\documentclass[12pt]{article}
\usepackage{psfrag,epsf}
\usepackage{enumerate}
\usepackage{natbib}
\usepackage{graphicx}
\usepackage{grffile}
\usepackage{url} % not crucial - just used below for the URL
\usepackage{amsmath,amsfonts,amsthm,epsfig,epsf}
\usepackage{comment}
\usepackage{mathtools}
\usepackage[noend]{algpseudocode}
\usepackage{algorithm}

\usepackage{color}
\usepackage{bbm}
\usepackage{theoremref}
\usepackage[inline]{enumitem}
\usepackage{xcolor}
\usepackage{bbm}
\usepackage{amsmath}
\usepackage{amssymb}
\usepackage{hhline}
\usepackage[utf8x]{inputenc}
\usepackage{multirow}
\usepackage{float}
\usepackage{siunitx}
\usepackage{slashbox}
\usepackage{array, makecell} %
\usepackage{titling}
\usepackage{etoolbox}
\usepackage{anyfontsize}

\usepackage{xr}
\makeatletter
\patchcmd{\@maketitle}{\LARGE}{\Large}{}{}
\makeatother

\newcommand{\ber}{\begin{eqnarray}}
\newcommand{\eer}{\end{eqnarray}}
\renewcommand{\baselinestretch}{1.2}

\newtheorem{corollary}{\noindent Corollary}
\newtheorem{proposition}{\noindent Proposition}
\newtheorem{definition}{\noindent Definition}
\newtheorem{claim}{\noindent Claim}
\newtheorem{remark}{\noindent Remark}

\newcommand{\be}{\begin{equation}}
\newcommand{\ee}{\end{equation}}
\newcommand{\bal}{\begin{align}}
\newcommand{\eal}{\end{align}}
\newcommand{\balnonum}{\begin{align*}}
\newcommand{\ealnonum}{\end{align*}}

\newcommand{\eps}{\epsilon}

\newcommand{\iid}{\overset{i.i.d.}{\sim}}

\newtheorem{lemma}{Lemma}[section]
\newcommand{\sample}{\mathcal{D}}
\newcommand{\dataset}{\mathcal{D}} % same as sample
\newcommand{\jointpdf}{f_{XY}}
\newcommand{\jointcdf}{F_{XY}}
\newcommand{\jointcdfhat}{\hat{F}_{XY}^n}
\newcommand{\xcdf}{F_{X}}
\newcommand{\xpdf}{f_{X}}
\newcommand{\ycdf}{F_{Y}}
\newcommand{\ypdf}{f_{Y}}
\newcommand{\xcdfhat}{\hat{F}_{X}}  % general estimators
\newcommand{\ycdfhat}{\hat{F}_{Y}}
\newcommand{\xcdfempir}{\hat{F}_{X}^n} % empirical distirbution has n
\newcommand{\ycdfempir}{\hat{F}_{Y}^n}

\newcommand{\wfun}{{w}} % weight function
\newcommand{\wmat}{\mathcal{W}} % weight matrix nXn
\newcommand{\truncregion}{\mathcal{A}} % region allowed by truncation

\newcommand{\indicator}[1]{\mathbbm{1}_{\{#1\}}} % indicator with condition
\newcommand{\setindicator}[1]{\mathbbm{1}_{#1}} % indicator of set

\newcommand{\nperm}{B}

\newcommand{\obsset}{o} % observed in a set
\newcommand{\expectset}{e}

\newcommand{\zuk}[1]{\textcolor{red}{[#1 --Or]}} % for comments
\newcommand{\micha}[1]{\textcolor{blue}{[#1 --Micha]}} % for comments

\newcommand{\blind}{1}

\newcommand\blfootnote[1]{%
	\begingroup
	\renewcommand\thefootnote{}\footnote{#1}%
	\addtocounter{footnote}{-1}%
	\endgroup
}

\begin{document}

\def\spacingset#1{\renewcommand{\baselinestretch}%
{#1}\small\normalsize} \spacingset{1}

%%%%%%%%%%%%%%%%%%%%%%%%%%%%%%%%%%%%%%%%%%%%%%%%%%%%%%%%%%%%%%%%%%%%%%%%%%%%%%

\date{} 

\vspace{-0.4in}
\if1\blind
{
	\title{\bf \fontsize{20}{25} \Large{Testing Independence under Biased Sampling}}
\vspace{-0.1in}	
	\author{Yaniv Tenzer, 
		 Micha Mandel and Or Zuk \\
		Department of Statistics, The Hebrew University of Jerusalem}
	\maketitle
} \fi
\vspace{0.2in}	
\blfootnote{*The authors thank Yair Heller, Yair Goldberg, Malka Gorfine, Sy Han Chiou, and Rebecca Betensky for fruitful discussions.
This research was supported by the NIH (grant no. R01NS094610) and by The Israel Science Foundation (grant No. 519/14).}

\if0\blind
{
%	\bigskip
%	\bigskip
	\bigskip
	\begin{center}
		{\LARGE\bf Testing Independence under Biased Sampling}
	\end{center}
%	\medskip
} \fi

\vspace{-0.8in}
%\smallskip
\begin{abstract}	
%	The text of your abstract. 200 or fewer words.
Testing for association or dependence between pairs of random variables is a
fundamental problem in statistics. In some applications, data are subject to selection bias that causes dependence between observations even when it is absent from the population. An important example is truncation models, in which observed pairs are restricted to a specific subset of the X-Y plane. Standard tests for independence are not suitable in such cases, and alternative tests that take the selection bias into account are required. To deal with this issue, we generalize the notion of quasi-independence with respect to the sampling mechanism, and study the problem of detecting any deviations from it. We develop two test statistics motivated by the classic Hoeffding's statistic, and use two approaches to compute their distribution under the null: (i) a bootstrap-based approach, and (ii) a permutation-test with non-uniform probability of permutations, sampled using either MCMC or importance sampling with various proposal distributions. We show that our tests can tackle cases where the biased sampling mechanism is estimated from the data, with an important application to the case of censoring with truncation. 
We prove the validity of the tests, and show, using simulations, that they perform well for important special cases of the problem and improve power compared to competing methods. The tests are applied to four datasets, two that are subject to truncation, with and without censoring, and two to positive bias mechanisms related to length bias.
\end{abstract}

\noindent%
{\it Keywords:} % 3 to 6 keywords, that do not appear in the title
quasi-independence, Markov chain Monte Carlo, permutation test, truncation, weighted distribution
\vspace{-0.0in}

%\newpage
\spacingset{1.45} % DON'T change the spacing!
\section{Introduction}
\label{sec:introduction}%
\vspace{-0.2cm} Testing independence of two random variables $X,Y$ is a fundamental statistical problem. Classical methods have focused on testing linear (Pearson’s correlation coefficient) or monotone (Spearman’s correlation, Kendall's tau) dependence, while other works focus on developing methods to capture complex dependencies (e.g. using Pearson's Chi-squared test). This classical problem keeps drawing attention from scholars with recent approaches focusing on omnibus tests employing computer-intensive methods \citep{gretton2008kernel,szekely2009brownian, heller2012consistent,heller2016consistent}.

A more challenging task is testing independence of two random variables when data is obtained through a general biased sampling mechanism. The most familiar example is that of truncation, where observations are restricted to a certain `observable region'. All cross-sectional samples are subject to some form of truncation, and the standard way of analysing such data is to assume independence between the truncation mechanism and the variables of interest (see \cite{chiou2018permutation} for more discussion and references). The results of the analysis can be highly biased if the independence assumption is violated. Yet, another important problem that exploits tests for truncated data is testing the Markov assumption in the Illness-Death model \citep{rodriguez2012nonparametric}.
The problem of quasi-independence was first dealt with in the framework of contingency tables (e.g., \cite{goodman1968analysis}) and was studied more recently in the framework of survival analysis where data are restricted by the condition $X\le Y$ (see, \cite{tsai1990testing} and the discussion below). \cite{bickel1991large} consider general biased regression models for which independence is equivalent to a zero regression coefficient. Another common framework where the problem naturally arises is in cross-sectional sampling designs;  Section \ref{sec:real_data} provides several examples.
Biased sampling in general, and truncation in particular, may imply dependence in the sample that does not exist in the population. This fact was acknowledged more than a century ago by \cite{elderton1913correlation} who studied the correlation between an  intellectually disabled child’s place in the family and the size of that family. They noticed that ``the size of the family must be as great or greater than the imbecile's place in it ... and there would certainly be correlation, if we proceeded to find it by the usual product moment method, but such correlation is, or clearly may be, wholly spurious".

In this paper, we study the problem of detecting dependency from general biased samples. We are given a sample of $n$ independent and identically distributed (i.i.d.) observations $\{(x_i, y_i)\}_{i=1}^n$ drawn from a joint distribution having a density $F^{(w)}_{X,Y}(dx,dy) \propto {\wfun(x,y)\jointcdf(dx,dy)}$, where $\wfun$ is a known non-negative function having a positive finite expectation with respect to $\jointcdf$. Here $\jointcdf$ is the joint distribution of the pair $(X,Y)$ in the population, and $F^{(w)}_{X,Y}$ is the distribution of observed pairs, tilted by $\wfun$, the sampling mechanism. The aim is to test the null hypothesis $H_0: \jointcdf(x,y)=\xcdf(x)\ycdf(y)$ for all $(x,y)$. However, because $\jointcdf$ is not identifiable on $\{(x,y):\wfun(x,y)=0\}$, the goal is restricted to testing quasi-independence defined as $\jointcdf(dx,dy)=\tilde{F}_{X}(dx)\tilde{F}_{Y}(dy)$ for all $(x,y)\in \{(x',y'):\wfun(x',y')>0\}$, for some functions $\tilde{F}_{X},\tilde{F}_{Y}$ \citep{tsai1990testing}. % check Tsai for definition of quasi-independence

Testing quasi-independence under a biased-sampling regime is challenging and previous works mainly focused on simple truncation models. \cite{tsai1990testing} considered the problem of testing quasi-independence under left-truncation (and right censoring) based on the conditional Kendall's tau correlation coefficient. \cite{efron1999nonparametric} and \cite{martin2005testing} extended this method to the settings of double-truncation.
\cite{chen2007sequential} suggested an importance sampling algorithm to estimate the P-value under truncation models. \cite{emura2010testing} constructed a log-rank type statistic for the left-truncation setting. %with optimal weights determined by the odds ratio function considered in \cite{chaieb2006estimating}.
\cite{chen1996product} proposed a conditional version of Pearson's product-moment correlation. These tests are powerful for monotone alternatives, but generally less powerful against non-monotone dependencies frequently encountered in real-life applications.

Some recent works accommodated non-monotone alternatives by utilizing local versions of Kendall's tau test \citep{rodriguez2012nonparametric, de2012markov}, and a weighted version of the local Kendall's tau \citep{rodriguez2016methods}. These tests are inefficient for small sample sizes, and their performance depends on the choice of pre-selected grids.

Testing quasi-independence by utilizing permutations was previously proposed by \cite{tsai1990testing}, \cite{efron1999nonparametric} and \cite{chiou2018permutation}. The test of \cite{chiou2018permutation} is based on a scan statistic that partitions the sample space into two groups according to a threshold value of $X$, and compares the distributions of $Y$ in the two groups. Although no formal results are established, the simulations indicate that the procedure has comparable power to the test of \cite{martin2005testing} for monotone alternatives, and performs much better for non-monotone alternatives. Yet, similarly to other aforementioned works, these works consider only the cases of one and two-sided truncation.
%In order to develop a valid test for a general bias function $\wfun$, a more general permutations sampling scheme is required in order to compute the distribution of the test statistic under the null.

The current paper describes a new family of tests of independence for data from a general biased sampling design. The tests are not restricted to monotone alternatives and are based on a scan statistic similar to the one suggested by \cite{heller2016consistent}. As in \cite{heller2016consistent}, P-values are calculated using permutations, but here the permutation distribution under the null hypothesis is not uniform and generating permutations is a challenging task. We consider two approaches for calculation of P-values, the first employs a Markov-Chain Monte Carlo (MCMC) algorithm to sample from the resulting weighted permutation distribution, and the second uses importance sampling.

In addition, an alternative bootstrap test of independence is considered that requires consistent estimation of the univariate marginal distributions under the alternative hypothesis. We identify two settings under which consistent estimators of the marginals are attainable (under both the null and the alternative hypotheses).

The rest of the paper is organized as follows. Section \ref{sec:preliminaries} introduces the problem of testing quasi-independence and its relation to previous works. Section \ref{sec:permutations} derives the permutation tests and studies some of their theoretical properties. Section \ref{sec:bootstrap} presents an alternative bootstrap-based approach. Section \ref{sec:teststat} introduces the adjusted Hoeffding Statistic and describes its computation. An inverse weighting test for the case of strictly positive $w$, and implementation of the tests to left-truncated right-censored data are also discussed. The new methods are compared in simulations and applied to real-life data sets in Sections \ref{sec:simulations} and \ref{sec:real_data}, respectively. Section \ref{sec:discussion} completes the paper with a discussion.

\section{Preliminaries}
\label{sec:preliminaries}
Let $\jointcdf$ be a bivariate distribution function with a density $\jointpdf$ and univariate marginals $\xpdf, \ypdf$. We consider $n$ independent pairs $(X_i,Y_i) \sim \jointcdf^{(w)}(x,y)$ of scalar continuous random variables, sampled from the joint density
\be
f^{(w)}_{X,Y}(x,y)=\wfun(x,y) \jointpdf(x,y)/\mathbb{E}_{\jointpdf} \{\wfun(X,Y)\},
\label{eq:weighted_density}
\ee
where $\wfun : \mathbb{R}^2 \to \mathbb{R}^+$ is a non-negative weight function such that $0 < \mathbb{E}_{\jointpdf}\{\wfun(X,Y)\} < \infty$. The marginals of the observed data are denoted by $\xpdf^{(w)}(x) = \int_{-\infty}^{\infty} \jointpdf^{(w)}(x,y) dy$ and $\ypdf^{(w)}(y) = \int_{-\infty}^{\infty} \jointpdf^{(w)}(x,y) dx$.
The weighted independent density is defined as $[\xpdf \ypdf]^{(w)}(x,y) = {\wfun(x,y) \xpdf(x) \ypdf(y)}/{\mathbb{E}_{\xpdf \ypdf} \{\wfun(X,Y)\}}$, with the corresponding weighted distribution $[\xcdf \ycdf]^{(w)}$.

An important special case is that of a truncated sample in which
\begin{align}
	\wfun(x,y) = \setindicator{\truncregion}(x,y) = \left\{ \begin{array}{lll}
	1  &   \quad (x,y) \in \truncregion  \\
	0 & \quad otherwise, \\
	\end{array} \right.
	\label{eq:truncation_function}
\end{align}	
for some set $\truncregion \subset \mathbb{R}^2$. This special form of $\wfun(x,y)$ arises frequently in practice and was previously investigated by several authors, as discussed in Section \ref{sec:introduction}.

%While there are known non-parametric tests for detecting any dependency between $X$ and $Y$ when sampling directly from $\jointcdf$ \citep{hoeffding1948non,heller2012consistent,szekely2007measuring,gretton2008kernel}, it is \textit{a priori} not clear which dependencies can be detected under biased sampling, i.e. when sampling from $\jointcdf^{(w)}$.

A strongly related concept to our problem is that of \emph{quasi-independence} in truncation models \citep{tsai1990testing}, which can be naturally extended to a general weight function $\wfun(x,y)$:
\begin{definition}(quasi-dependence)
	We say that the joint distribution $\jointcdf^{(w)}$ is quasi-independent with respect to the weight function $\wfun$, if there exist density functions $\tilde{\xpdf}$ and $\tilde{\ypdf}$, such that
\be \label{eq:quasi_independence}
	\jointcdf^{(w)}(x,y) \propto \int\displaylimits_{-\infty}^{x}\int\displaylimits_{-\infty}^{y} \wfun(s,t) \tilde{\xpdf}(s) \tilde{\ypdf}(t) dt ds, \quad \forall x,y \in \mathbb{R} \: .
\ee

Otherwise, we say that $\jointcdf^{(w)}$ is quasi-dependent with respect to the weight function $\wfun$.
\label{def:quasi_independence}
\end{definition} % quasi-independence

Based on the sample $\sample = \{(x_i,y_i)\}_{i=1}^n$, we aim at performing the following hypothesis testing for quasi-dependence:
\be
\begin{array}{rl}
H_0: & \jointcdf^{(w)} \textrm{\; is \; quasi-independent} \\
H_1: & \textrm{otherwise}
\end{array}
\label{eq:main_hypothesis_testing}
\ee

\begin{remark}
Quasi-dependence implies dependence. When $\wfun$ is strictly positive, quasi-dependence is simply dependence of $X$ and $Y$.
\end{remark}

\begin{remark}
Quasi-independence does not imply independence.  If $\wfun(x,y) = 0$ for some $(x,y) \in \mathbb{R}^2$, it is possible to have quasi-independence without independence (and then we must have either $\tilde{\xcdf} \neq \xcdf$ or $\tilde{\ycdf} \neq \ycdf$). For the important case of $w(x,y)=\mathbbm{1}(x,y)$, \cite{cheng2007nonparametric} discuss the identifiability problem and its implications.
%; see also the recent discussion by \cite{vakulenko2019nonidentifiability}.
\end{remark}

We denote by $\sample_x, \sample_y$ the unordered samples comprising of $\{x_1,\ldots,x_n\}$ and $\{y_1,\ldots,y_n\}$, respectively. For convenience, we often keep the indices of the original data, but not the coupling between $x,y$, and to this end we use
the unordered sample $(x, \sample_y)$ - i.e. we keep the original ordering of the $x_i$'s but only the marginal empirical distribution of the $y_i$'s.

\section{Permutation Test}
\label{sec:permutations}
\subsection{The Distribution of Permutations}
Under biased-sampling, different permutations are not equally likely under the null model, thus should not be uniformly sampled as in standard permutation tests. Therefore, the sampling mechanism should account for the discrepancy in weight of distinct permutations based on the data. Let $\pi(y)$ be the vector $y$ rearranged according to a permutation $\pi$, i.e. $\pi(y) = (y_{\pi(1)}, .., y_{\pi(n)})$, and
let $\pi(\sample)$ be the permuted sample: $\pi(\sample) \equiv \big((x_1, y_{\pi(1)}), .., (x_n, y_{\pi(n)})\big)$.
For a sample $\sample=\{(x_i, y_i)\}_{i=1}^n$, let $\wmat \in \mathbb{R}_{n\times n}$ be a weight-matrix defined by
$\wmat(i,j) \equiv \wfun(x_i, y_j).$
Let $S_n$ be the set of all permutations of $n$ elements, and for $\pi\in S_n$ consider the probability 
\be
P_{\wmat}(\pi) \equiv \frac{1}{per(\wmat)} \prod_{i=1}^n \wmat(i,\pi(i)),
\label{eq:prob_perm}
\ee
where $per(\wmat) = \sum_{\pi \in S_n} \prod_{i=1}^n \wmat(i, \pi(i))$ is the normalizing constant, given by the permanent of the matrix $\wmat$.

\setcounter{claim}{0}
\begin{claim}
	\label{claim:P_W_conditional}	
	Under $H_0$, the probability $P_{\wmat}(\pi)$ represents the probability of observing permuted datasets conditional on the marginal sets, that is, $P_{\wmat}(\pi(\sample)) = P_0(\pi \mid x, \sample_y)$, where $P_0$ denotes the probability under quasi-independence.	
%	\begin{proof}
%		For a general weighted model, we have
%		%
%		\be
%		\label{eq:generalperm}
%		P(\pi(\sample) \mid x, \sample_y) = \frac{\prod_{i=1}^n \jointpdf^{(w)}(x_i, y_{\pi(i)}) }{\sum_{\pi'\in S_n} %\prod_{i=1}^n \jointpdf^{(w)}(x_i, y_{\pi'(i)}) }.
%		\ee
%		%
%		Under the null, $\jointpdf^{(w)}(x,y)\propto w(x,y)\tilde{\xpdf}(x)\tilde{\ypdf}(y)$, hence
%		\begin{eqnarray*}	
%			P_0(\pi(\sample) \mid x, \sample_y) = \frac{\prod_{i=1}^n \tilde{\xpdf}(x_i) \tilde{\ypdf}(y_{\pi(i)}) %\wmat(i,\pi(i))}{\sum_{\pi' \in S_n} \prod_{i=1}^n \tilde{\xpdf}(x_i) \tilde{\ypdf}(y_{\pi'(i)}) \wmat(i,\pi'(i))}
%			= P_{\wmat}(\pi).
%		\end{eqnarray*}	
%	\end{proof}
\end{claim}
\begin{proof}
	For a general weighted model, we have
	\be
	\label{eq:generalperm}
	P(\pi(\sample) \mid x, \sample_y) = \frac{\prod_{i=1}^n \jointpdf^{(w)}(x_i, y_{\pi(i)}) }{\sum_{\pi'\in S_n} \prod_{i=1}^n \jointpdf^{(w)}(x_i, y_{\pi'(i)}) }.
	\ee
	Under the null, $\jointpdf^{(w)}(x,y)\propto w(x,y)\tilde{\xpdf}(x)\tilde{\ypdf}(y)$, hence
	\begin{eqnarray*}	
		P_0(\pi(\sample) \mid x, \sample_y) = \frac{\prod_{i=1}^n \tilde{\xpdf}(x_i) \tilde{\ypdf}(y_{\pi(i)}) \wmat(i,\pi(i))}{\sum_{\pi' \in S_n} \prod_{i=1}^n \tilde{\xpdf}(x_i) \tilde{\ypdf}(y_{\pi'(i)}) \wmat(i,\pi'(i))}
		= P_{\wmat}(\pi).
	\end{eqnarray*}	
\end{proof}

When $\wfun(x,y)$ is a truncation function, $P_{\wmat}(\pi)$ is simply the uniform distribution over the set of valid permutations, i.e. permutations $\pi$ yielding permuted datasets $\pi(\sample)$ which are consistent with the truncation.
The next lemma shows that permuted data points drawn from Equation \eqref{eq:generalperm} follow the distribution of $n$ independent copies of $(X,Y) \sim \jointcdf^{(w)}$.
\begin{lemma}
\label{lemma:doulbe_randomization}	
Let $\dataset \sim  [\jointcdf^{(w)}]^n$ and conditionally on $\dataset$ let $\pi$ be a permutation having the conditional probability law given in Equation \eqref{eq:generalperm}. Then $\pi(\dataset) \sim  [\jointcdf^{(w)}]^n$.

\begin{proof}
Using Equation \eqref{eq:generalperm} and the law of total probability, we have
\begin{eqnarray}
f_{\pi(\dataset)}\Big((x_1,y_1),\ldots,(x_n,y_n)\Big)& =&
 \sum_{\pi \in S_n} \frac{\prod_{i=1}^n \jointpdf^{(w)}(x_i, y_i)}{\sum_{\pi'} \prod_{i=1}^n \jointpdf^{(w)}(x_i, y_{\pi'(i)}) } \prod_{i=1}^n \jointpdf^{(w)}(x_i, y_{\pi(i)}) \nonumber \\
 & =& \prod_{i=1}^n \jointpdf^{(w)}(x_i, y_i).
\end{eqnarray}
\end{proof}
\end{lemma}
Recalling that under the null, Equation \eqref{eq:generalperm} reduces to \eqref{eq:prob_perm} (by Claim \ref{claim:P_W_conditional}), we have
\begin{corollary}
	\label{cor:alpha}	
Under the null, permuted data points drawn according to Equation \eqref{eq:prob_perm} follow the distribution of $n$ independent copies of $(X,Y) \sim [\tilde\xcdf \tilde\ycdf]^{(w)}$.
\end{corollary}

\subsection{The Weighted-Permutation Test of Independence}

Let $T(\sample)$ be any test statistic. The permutation test consists of comparing $T(\sample)$ to its null distribution over all permuted samples $\pi(\sample)$, and calculating the P-value by the proportions of permutations $\pi_i$ with
test statistic $T(\pi_i(\sample))$ exceeding $T(\sample)$. In practice, a large number of permutations, $B$, is sampled, and the P-value is approximated by
\be
{P}_{value} = \frac{1}{\nperm+1} \sum_{i=0}^{\nperm} \indicator{ T(\pi_i(\sample)) \geq T(\sample) },
\label{eq:empirical_pval}  % maybe not needed. Appears in Algorithm
\ee
where  $\pi_0$ is the identity permutation corresponding to $T(\sample)$. We formalize our weighted-permutation test of independence as shown in Algorithm \ref{alg:weighted_permutations}.

\begin{algorithm}
	\caption{{\small Weighted Permutation Test of Quasi-Independence}} {\small
	\begin{algorithmic}[1]
		\Statex {\bf Input:} $\sample$ - sample, $\wfun(x,y)$ - bias function, $T : \mathbb{R}^{2n} \to \mathbb{R}$ - a test statistic
		\Statex {\bf Parameters:} $\nperm$ - number of permutations
		\State Generate $\nperm$ permutations $\pi_1,..,\pi_{\nperm} \sim P_{\wmat}$. % using the MCMC scheme in Algorithm \ref{alg:mcmc_permutations}.
		\State Compute the test statistic $T_0 \equiv T(\sample)$.
		%\STATE Estimate the marginals $\hat{F}_X, \hat{F}_Y$ 		
		\For{$i=1$ to $\nperm$}		
			\State Compute the test-statistic for the permuted dataset, $T_i = T(\pi_i(\sample))$.
		\EndFor
		\State {\bf Output:} $P_{value} \equiv \frac{1}{\nperm+1} \sum_{i=0}^{\nperm} \indicator{T_i \geq T_0}$.
		\end{algorithmic}
\label{alg:weighted_permutations}}
\end{algorithm}

Corollary \ref{cor:alpha} assures that under the null distribution $P_{value} \sim U[\frac{1}{B+1},...,\frac{B}{B+1},1]$ the type-1 error probability of the weighted permutation test is at most $\alpha$. (The addition of $1$ to the denominator and numerator in stage $6$, i.e. including the original sample, is necessary to ensure type-1 error below $\alpha$, but can be neglected in practice for large $B$.)
%For any statistic $T$, the computation in Step $2$ depends mathematically only on the original sample $\sample$. However,
%as we will discuss in Section \ref{sec:teststat}, the appropriate statistics we should use take into account expectations under the null
%with biased sampling - to compute these expectations we will use the permuted datasets $\pi_1(\sample), .., \pi_{\nperm}(\sample)$ that are available to us after the completion of Step $1$.

In order to implement Algorithm \ref{alg:weighted_permutations}, a method to sample weighted permutations $\pi_i \sim P_{\wmat}$ is required. An MCMC algorithm that generates such permutations is discussed next.

\subsection{Sampling Permutations using MCMC}
\label{sec:mcmc_permutations}
The case of sampling uniformly from a restricted set of permutations, i.e., permutations $\pi$ with $P_{\cal W}(\pi)>0$, was considered by \cite{diaconis2001statistical}. Their algorithm deals with the important special case of truncation with a 0/1 weight function. To enable sampling from a general distribution, we utilize the Metropolis-Hasting (MH) algorithm \citep{metropolis1953equation,hastings1970monte}. Let $\pi_t = (\pi_t(1),..,\pi_t(n))$ be the permutation at step $t$. Define the neighbours of $\pi_t$ to be all permutations obtained from $\pi_t$ by a single swap, that is,
\[
Neig(\pi_t) \equiv\{\pi_t^{i \leftrightarrow j} \equiv \big(\pi_t(1),..\pi_t(j),..,\pi_t(i), .., \pi_t(n)\big),\hspace{0.04in} \forall i < j\}.
\]
We then proceed according to the standard MH algorithm: at each iteration we sample a permutation uniformly from this set $\pi_t^{i \leftrightarrow j} \sim U [Neig(\pi_t)]$, as well as generate a uniform random number $u$ on $[0,1]$. Finally, we accept the new permutation only if
\[
u \leq \frac{P(\pi_t^{i \leftrightarrow j})}{P(\pi_t)} = \frac{\wmat(i,\pi_t(j)) \wmat(j,\pi_t(i))}{\wmat(i,\pi_t(i))\wmat(j,\pi_t(j))}.
\]
A similar algorithm was suggested by \cite{efron1999nonparametric} for doubly truncated data. However, for truncated data the weights are all 0 or 1, making the problem much simpler.
Algorithm \ref{alg:mcmc_permutations} describes our MCMC approach step-by-step.
\begin{algorithm}
	\caption{{\small MCMC for Biased Sampling of Permutations}}  {\small
	\begin{algorithmic}[1]
		\Statex {\bf Input:} $\sample$ - sample, $\wfun(x,y)$ - bias function
		\Statex {\bf Parameters:} $\nperm$ - number of permutations, $M_0$ - 'burn-in' number of steps, $M$ - number of steps between two permutations.
		\State Compute $\wmat(i,j) = \wfun(x_i,y_j), \quad \forall i,j=1,..,n$.
		\State Set $\pi_0$ the identity permutation $\pi_0(i) = i$.
		\For{$t=0$ to $M_0 + \nperm M - 1$}
			\State Sample $\pi_t^{i \leftrightarrow j} \sim U [Neig(\pi_t)]$, and $u \sim U[0,1]$.		
			\If {$u \leq \frac{\wmat(i,\pi_t(j)) \wmat(j,\pi_t(i))}{\wmat(i,\pi_t(i))\wmat(j,\pi_t(j))}$}
				\State set $\pi_{t+1} \leftarrow \pi_t^{i \leftrightarrow j}$.
			\Else
				\State set $\pi_{t+1} \leftarrow \pi_t$.
			\EndIf
		\EndFor		
		\State {\bf Output:} The resulting $\nperm$ permutations $\pi_{M_0},\pi_{M_0+M},..,\pi_{M_0+\nperm M}$.
	\end{algorithmic}
	\label{alg:mcmc_permutations}}
\end{algorithm}

\subsection{Importance Sampling}
\label{sec:importance_sampling}
Due to the difficulty of sampling directly from the weighted permutations distribution $P_{\wmat}$, we propose here another approach: sample permutations $\pi_1,\ldots,\pi_B$ according to an importance probability law $P_{IS}$ such that $P_{IS}(\pi)>0$ whenever $P_{\wmat}(\pi)>0$, and calculate the P-value by
\begin{equation}
P_{value} \equiv \frac{\sum_{i=0}^{\nperm} \frac{ P_{\wmat}(\pi_i)}{ P_{IS}(\pi_i)} \indicator{ T_i \geq  T_0 } } {\sum_{i=0}^{\nperm}\frac{P_{\wmat}(\pi_i)}{P_{IS}(\pi_i)}}.
\label{eq:IS_pval}
\end{equation}
The unknown term $per(\wmat)$ appearing in $P_{\wmat} $ (see Equation \eqref{eq:prob_perm}) is cancelled in the above equation, thus enabling us to compute the $P_{value}$ even when $P_{\wmat}$ is known only up to a normalizing constant.
\cite{chen2007sequential} suggest this importance sampling algorithm for statistical inference under truncated data. It is based on a simple sequential method to generate permutations under $P_{IS}$. \cite{kou2009approximating} have generalized one of the approaches proposed in \cite{chen2007sequential}
to estimate the permanent of general weight functions. We have derived several sequential importance sampling approaches, similar to those of \cite{chen2007sequential} and \cite{kou2009approximating}, that are applicable for general weight functions $\wmat$, and investigated their performances in testing.

\cite{harrison2012conservative} shows that for any test statistic satisfying mild invariance properties, the test that includes the identity permutation
in the P-value calculation in Equation \eqref{eq:IS_pval} controls the type-1 error at level $\alpha$ (see his Theorem 1). %, we can show that the test is valid for test statistics satisfying mild invariance properties.
This result applies directly to our case by considering our approach as testing conditionally on the data $\sample$.

Although the correction above ensures validity, the importance sampling approach can perform very poorly if the importance distribution $P_{IS}$ is far from $P_{\wmat}$, for example when  $P_{IS}$ is taken to be the uniform distribution over $S_n$, because in such cases the $B$ sampled permutations have very low probability under $P_{\wmat}$. It is thus challenging to suggest a distribution $P_{IS}$ that is both easy to calculate and sample from, as well as close enough to $P_{\wmat}$ for general $\wmat$ - see Supp. Materials, Section \ref{sec:IS_appendix} for more details.  % \zuk{Correct reference here after we describe results elsewhere}. 
\section{Bootstrap-Based Test}
\label{sec:bootstrap}
\subsection{The Bootstrap Algorithm}
The permutation test bypasses the need to estimate the (unbiased) marginal distributions $\xpdf, \ypdf$, which can be difficult and even impossible when $\wfun$ vanishes on part of the support of $\jointpdf$. Nevertheless, when we can estimate the univariate marginals consistently from the data, a bootstrap test is a viable alternative. Briefly, we generate independent samples from the estimated null distribution and compute the test statistic for each such sample. We then reject the null hypothesis if the observed test statistic is greater than the $1-\alpha$ quantile of the resulting bootstrap distribution. The test is summarized by Algorithm \ref{alg:bootstrap}.
\begin{algorithm}
	\caption{{\small Bootstrap-Based Test of Quasi-Independence}} {\small
	\begin{algorithmic}[1]
		\Statex {\bf Input:} $\sample$ - sample, $\wfun(x,y)$ - bias function, $T : \mathbb{R}^{2n} \to \mathbb{R}$ - a test statistic
		\Statex {\bf Parameters:} $B$ - number of bootstrap samples
		\State Estimate the marginals $\xcdfhat, \ycdfhat$. 	% this is before as it is needed for computing T
		\State Compute the test statistic $T_0 = T(\sample)$.
		\For{$i=1$ to $B$}
			\State Generate a bootstrap sample $\sample_i$ by sampling with replacement $n$ i.i.d examples from $[\xcdfhat \ycdfhat]^{(w)}$. % $\int_{-\infty}^x\int_{-\infty}^y w(u,v)f_X(u)f_Y(v)dvdu$.
			\State Estimate the marginals $\hat{F}_{X,i}, \hat{F}_{Y,i}$ of the bootstrap sample $\sample_i$.
			\State Compute the test statistic $T_i = T(\sample_i)$.
		\EndFor		
		\State {\bf Output:} $P_{value} \equiv \frac{1}{B+1} \sum_{i=0}^{B} \indicator{ T_0 \ge T_i}$.
		%	\end{procedure}
	\end{algorithmic}
\label{alg:bootstrap}}
\end{algorithm}

As an alternative of estimating the marginal distributions, samples can be drawn under the null from the unbiased conditional distribution of $X$ given the observed $Y$ values $y_1,\ldots,y_n$; \cite{efron1999nonparametric} apply this approach to doubly truncated data.

\subsection{Estimating the Marginal Distributions}
\label{sec:bootstrap_marginal_estimation_null}

The next challenge is implementing step $1$ of Algorithm \ref{alg:bootstrap}, namely estimating the univariate marginals, given a known bias function $\wfun(x,y)$.
Naturally, under a bias-sampling regime, the underlying marginals may not be identifiable, unless additional modelling assumptions, either on $\jointcdf$ or $\wfun(x,y)$, are made.
%Next we demonstrate this principle and consider several such assumptions. In each case we derive and discuss the resulting estimators.

\subsubsection{Case 1: Estimating the Marginal Distributions Under Quasi-independence}
\label{sec:marginal_estimation_under_null}
For a valid test, it is enough to estimate $\xcdf$ and $\ycdf$ in the observable region under the null hypothesis of quasi-independence.  A general algorithm for estimation of the marginal densities in \eqref{eq:quasi_independence} under quasi-independence is developed next. \cite{bickel1991large} provide a somewhat similar algorithm for the case where $X$ is discrete, which reduces to the selection bias model of \cite{vardi1985empirical}.

Let $\tilde{X}\sim \tilde{F}_X, \tilde{Y}\sim \tilde{F}_Y$, where $\tilde{F}_X,\tilde{F}_Y$ are the cumulative distribution functions of $\tilde{f}_X,\tilde{f}_Y$ in \eqref{eq:quasi_independence}. Under quasi-independence, by the law of total probability, the density of $X$ can be written as ${f}_{X}(x)=\mathbb{E}\{w(x,\tilde{Y})\}\tilde{f}(x)/\mathbb{E}\{w(\tilde{X},\tilde{Y})\}$. Thus, an estimate of $\tilde{F}_Y$ yields an estimate for $\mathbb{E}\{w(x,\tilde{Y})\}$, which can be used to build an inverse weighting estimate for $\tilde{F}_X$:
\begin{equation} \label{eq:estim_update}
\widehat{\tilde{F}}_X(x)= \frac{\sum_{i=1}^n \indicator{x_{i}\le x}
\hat{\mathbb{E}}\{w(x_i,\tilde{Y})\}^{-1}} {\sum_{i=1}^n \hat{\mathbb{E}}\{w(x_i,\tilde{Y})\}^{-1}}.
\end{equation}
This estimate can be used in turn to estimate $\tilde{F}_Y$, suggesting an iterative procedure as described in Algorithm \ref{alg:est_marg}.
For the important case $\wfun(x,y) = \indicator{x<y}$, the algorithm reduces to the standard product-limit (PL) estimator for left and right truncated data, implemented, for example, in the DTDA package of R \citep{moreira2010dtda}. For more details and an extension to more than two variables see Supp. Materials, Section \ref{sec:IterativeAlgorithm_appendix}.

\begin{algorithm}
	\caption{{\small Estimation of Marginal Distributions Under Quasi-independence}} {\small
	\begin{algorithmic}[1]
		\Statex {\bf Input:} $\sample$ - sample, $\wfun(x,y)$ - bias function, $d(F_1,F_2)$ - distance function.
		\Statex {\bf Parameters:} $\epsilon$ - convergence criterion.
		\State Generate initial estimates $\tilde F^{new}_X,\tilde F^{new}_Y$ and set $\tilde F^{old}_X=\tilde F^{old}_Y\equiv 0$.
\While{$d(\tilde F^{old}_X,\tilde F^{new}_X)+d(\tilde F^{old}_Y,\tilde F^{new}_Y)>\epsilon$}
\State Set $\tilde F^{old}_X=\tilde F^{new}_X$ and $\tilde F^{old}_Y=\tilde F^{new}_Y$.
		\State Calculate $\mathbb{E}_{\tilde F^{new}_{Y}}\{w(x,{Y})\}$, and update $\tilde F^{new}_X$ using \eqref{eq:estim_update}.
\State Calculate $\mathbb{E}_{\tilde F^{new}_{X}}\{w({X},y)\}$, and update $\tilde F^{new}_Y$ using the equivalent for $\tilde{F}_Y$ of \eqref{eq:estim_update}.
\EndWhile
\State {\bf Output:} $\tilde F^{new}_X$ and $\tilde F^{new}_Y$.
	\end{algorithmic}
\label{alg:est_marg}}
\end{algorithm}

While Algorithm \ref{alg:est_marg} provides a general procedure to estimate a distribution under independence, for testing purposes it may result in low power. Consider the truncation model $\wfun(x,y) = \indicator{x<y}$. The PL estimators are consistent under the null hypothesis, but using them in our test leads to low power, even in seemingly very extreme situations of a strong dependence. For a test to perform reasonably well for moderate sample sizes, $\xcdf$ and $\ycdf$ should be estimated well not only under the null, but also under the alternative hypothesis (see Section \ref{sec:simulations} and Supp. Materials, Section B). The next example demonstrates this claim.

\noindent \textbf{Example: Difficulties in Detecting Quasi-Independence}
%\label{sec:bootstrap_example_wrong_marginal_estimation}

Consider the case where data is generated from a uniform
bivariate distribution over $\{0 < x, y < 1: \hspace{0.03in} |x − y| < 0.3\}$, and let $\wfun(x,y)=\indicator{x<y}$ be the standard truncation model.

We drew $n=500$ samples from this model and estimated the univariate marginals CDFs using the PL estimators. The left panel of Figure \ref{example} shows the sampled data points. The green and red curves in the right panel are the resulting PL estimates of $\tilde{F}_X$ and $\tilde{F}_Y$, respectively. Because the unbiased variables $\tilde{X}$ and $\tilde{Y}$ are exchangeable, they share the same underlying marginal distribution, depicted by the blue line. The product-limit estimates differ considerably from the true marginal distribution. When such CDFs generate the truncated data, the probability of a selection $(\tilde X < \tilde Y)$
is small, and when it happens, the values of $\tilde X$ and $\tilde Y$ tend to be close, yielding a scatter plot somewhat similar to the observed data. Indeed, the middle panel of Figure \ref{example} shows pairs obtained by generating independent
variables from the estimated product-limit curves $\xcdfhat$, $\ycdfhat$ and retaining only observations satisfying $\tilde X< \tilde Y$.
This example shows that independent variables under selection bias can produce data similar to that obtained by strongly dependent variables. Applying the bootstrap test using estimates of the marginal that are consistent only under the null independence assumption may result in a test with low power.

%
\begin{comment}
\begin{figure}[!h]
\begin{minipage}{1.\columnwidth}
	\begin{center}		
		\begin{tabular}{cc}
			\begin{tabular}{c}
				\includegraphics[trim={0 2 0 0}, width=0.3\columnwidth]{../Figures/simulations/UniformStrip_rho_0.3_scatter.png}
			\end{tabular}
			&
			\hspace{-0.1in}
			\begin{tabular}{c}
				%\vspace{0.1in}
				\includegraphics[width=0.3\columnwidth]{../Figures/simulations/UniformStrip_rho_0.3_KM_marginal.png}
			\end{tabular}\\
			
			\begin{minipage}{0.374\columnwidth}
				\vspace*{0.25in}
				\hspace*{-0.45in}
				\begin{tabular}{c}
					\includegraphics[trim={0 0 0 1.5cm}, width=1.27\columnwidth]{../Figures/simulations/UniformStrip_rho_0.3_marginal_scatter.png}
				\end{tabular}
			\end{minipage}
			&
			\hspace{0.2in}
			\begin{minipage}{0.3\columnwidth}
				\vspace{-0.1in}
				
				\caption{(top left): a scatterplot of samples drawn from the true underlying distribution. (bottom left): a scatterplot of samples drawn from the truncated independence distribution, using the PL estimates. (top right): the PL estimates of the univariate marginal CDFs $\xcdfhat$ and $\ycdfhat$ are biased and do not resemble the true underlying CDF $\xcdf=\ycdf$.}
				\label{example}	
			\end{minipage}
		\end{tabular}
		%\caption{}
	\end{center}
\end{minipage}	
\end{figure}
\end{comment}

\begin{figure}[!h]
%		\begin{center}		
			\begin{tabular}{ccc}
					\includegraphics[trim={0 2 0 0}, width=0.3\columnwidth]{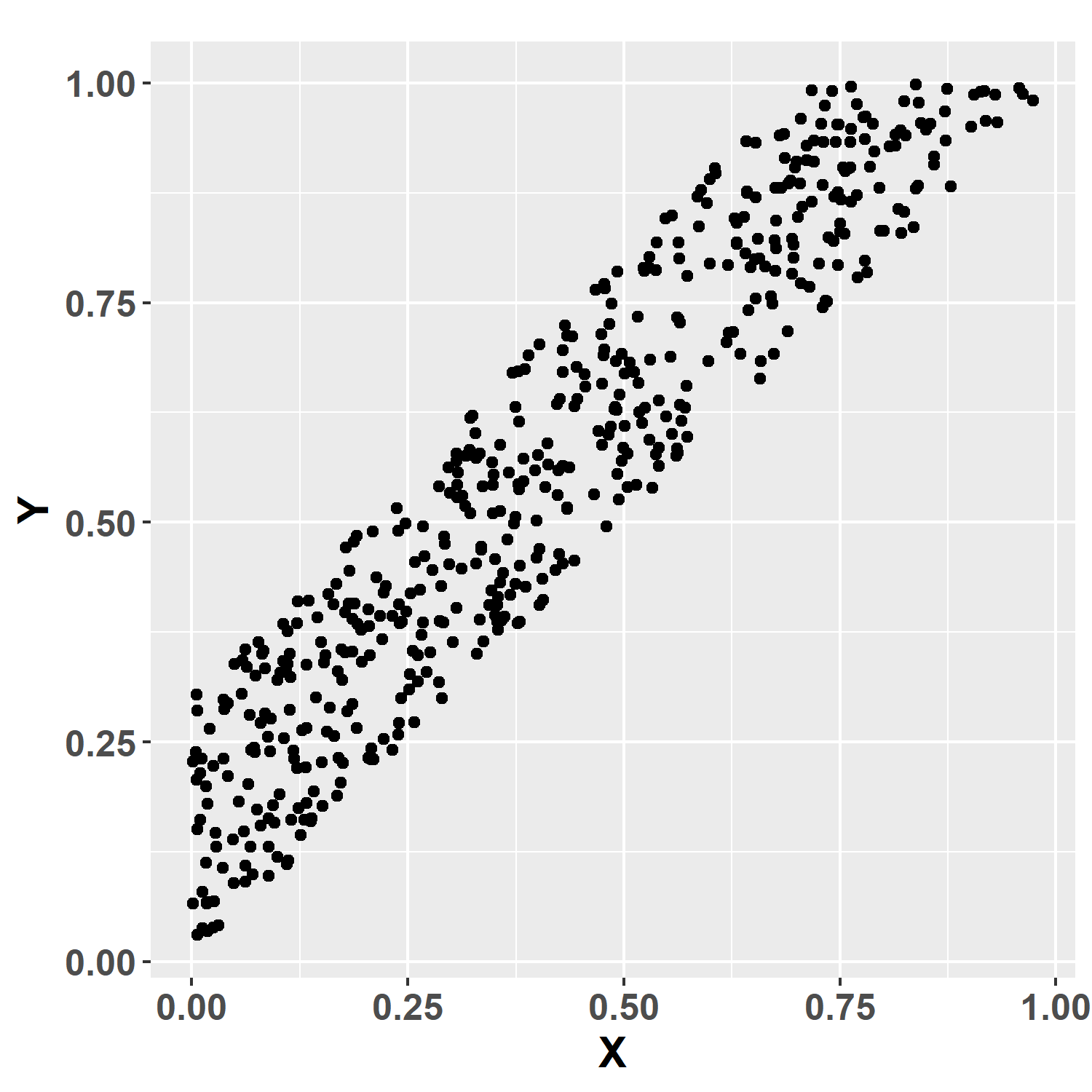}
				&
				\hspace{-0.1in}
				\includegraphics[trim={0 0 0 1.5cm}, width=0.3\columnwidth]{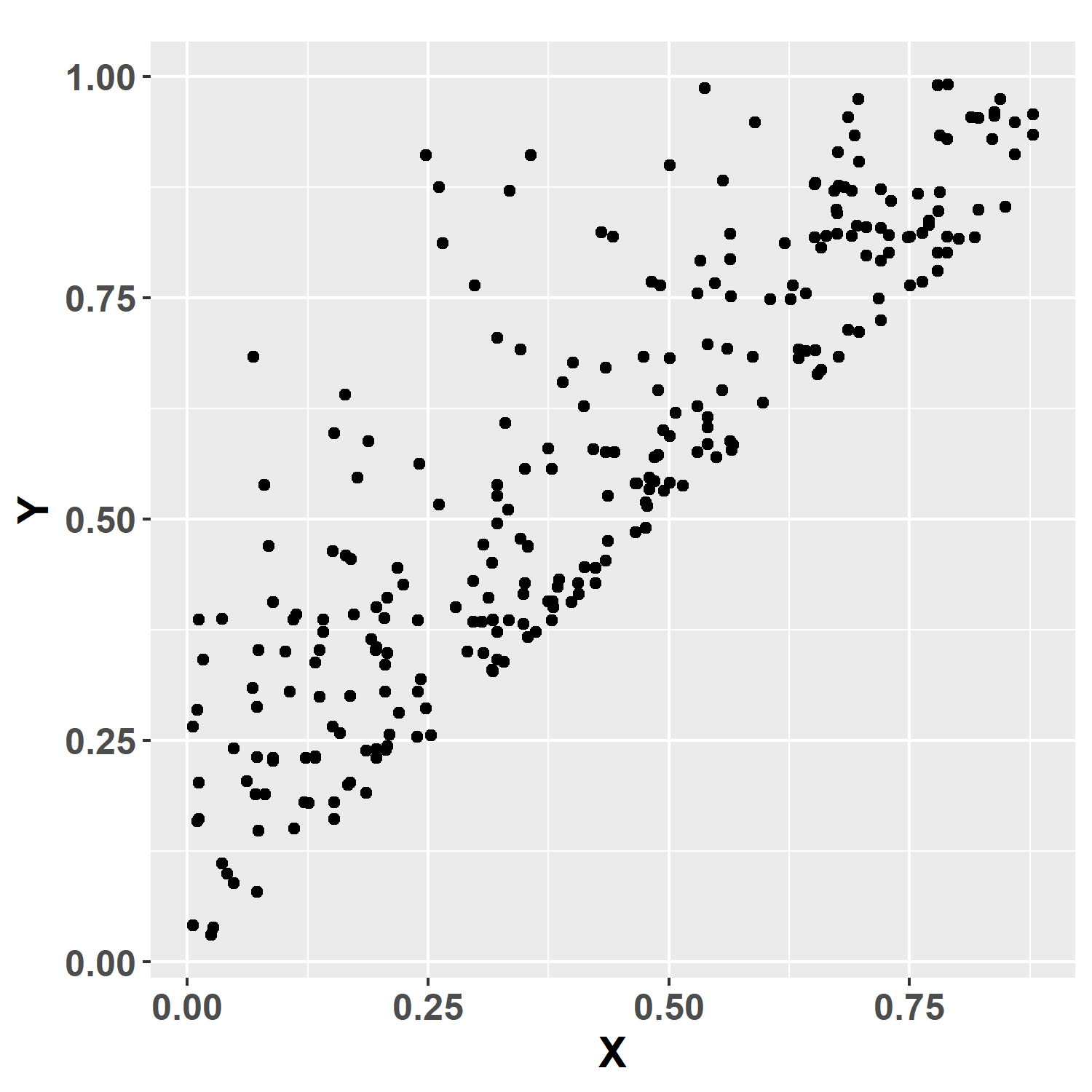}
				&	
					\hspace*{-0.1in}
					\includegraphics[width=0.3\columnwidth]{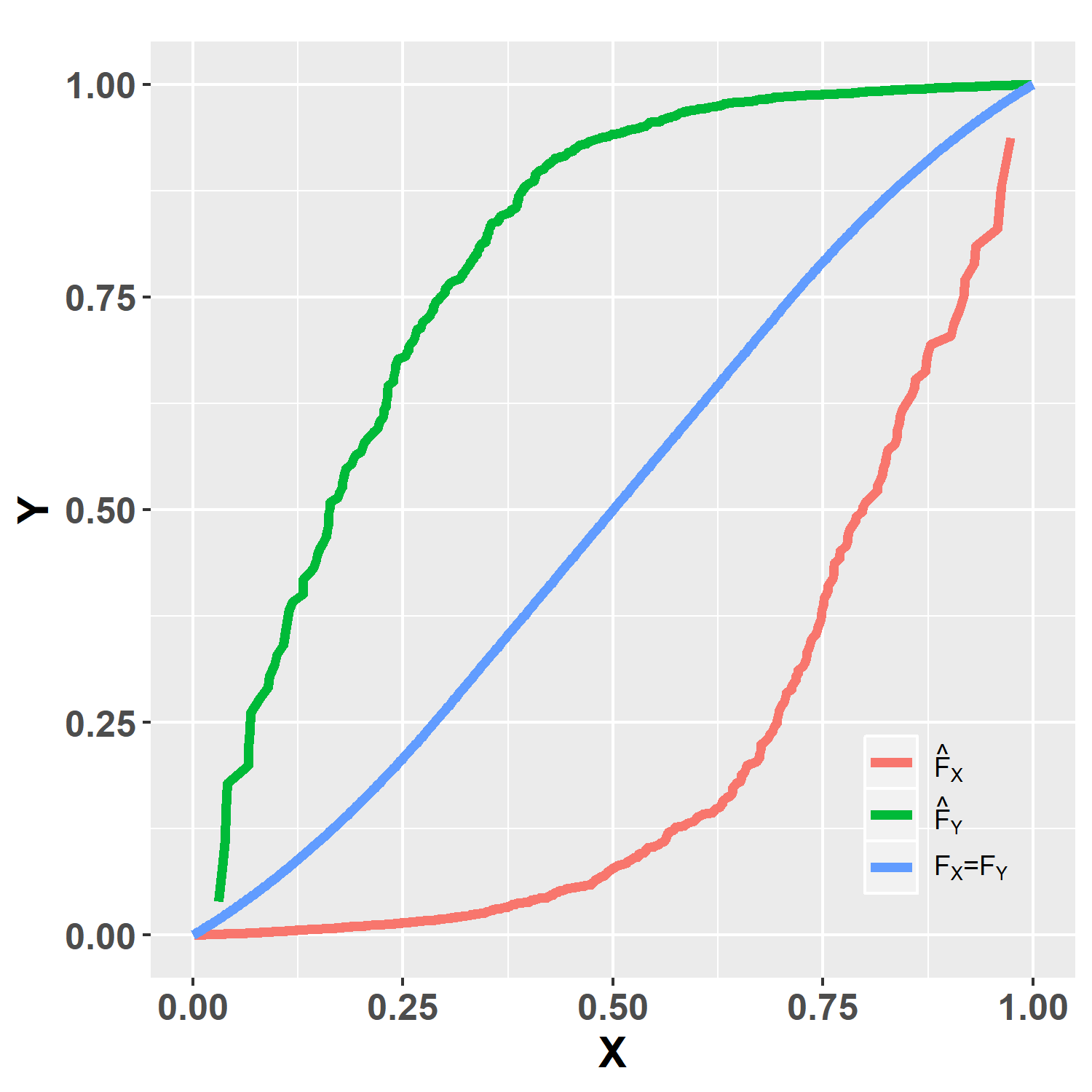}
		\end{tabular}								

%		\end{center}
	\caption{(left): a scatterplot of samples drawn from the true underlying distribution. (middle): a scatterplot of samples drawn from the truncated independence distribution, using the PL estimates. (right): the PL estimates of the univariate marginal CDFs $\xcdfhat$ and $\ycdfhat$ are biased and do not resemble the true underlying CDF $\xcdf=\ycdf$.}
\label{example}		
\end{figure}

%\subsection{Estimating Marginal Distributions Under Biased Sampling}
%\label{bootstrap_examples}
%
%In this section we discuss two settings in which the univariate marginals $\xcdf, \ycdf$ can be estimated consistently under biased-sampling.
%First, we consider the case of strictly positive $\wfun(x,y)$. In this case, deriving consistent estimators for the univariate marginals is relatively easy.
%Then, we consider the truncation function $\wfun(x,y) = \indicator{x<y}$. In this case, additional assumptions regarding the structure of the joint distribution $\jointcdf$ are needed to allow consistent estimators of the marginals.

A possible solution is to find estimators for the marginal CDFs that are consistent also under the alternative hypothesis of quasi-dependence. However, as the model is not identifiable under the alternative \citep{cheng2007nonparametric}, such estimators can be calculated only under additional assumptions, either on $\wfun(x,y)$, or on the underlying joint distribution (or both). We next demonstrate this through two different settings.
 %(i) When the biased sampling function $\wfun$ is strictly positive, and (ii) When $\wfun$ is a truncation function and $\jointcdf$ is an exchangeable distribution.

\subsubsection{Case 2: Strictly Positive $\wfun(x,y)>0$}

The problem of estimating non-parametrically a general multivariate distribution $F$ using weighted data is well known (e.g., \cite{vardi1985empirical}) and for $w>0$ the non-parametric maximum likelihood estimator (NPMLE) is given by:
\be
\jointcdfhat(t, s) = \frac{\sum_{i=1}^{n} \indicator{X_i\leq t, Y_i\leq s}w(X_i, Y_i)^{-1}}{\sum_{i=1}^{n}w(X_i, Y_i)^{-1}}.
\label{eq:non_parametric_MLE}
\ee
Estimators for $\xcdf$ and $\ycdf$ can be then obtained by marginalization of Equation \eqref{eq:non_parametric_MLE},
\be
\xcdfempir(t) =\jointcdfhat(x, \infty) = \frac{\sum_{i=1}^{n} \indicator{X_{i}\leq t}w(X_i,Y_i)^{-1}}{\sum_{i=1}^{n}w(X_i, Y_i)^{-1}}.
\label{eq:marginal_estimation_under_positive_weight_function}
\ee

The estimator above can be used whenever $w>0$ in the entire support of $\jointcdf$.
%, even if $w = 0$ for other $(x,y)$ values,
%and in particular if $w(X_i, Y_j)=0$ for some $i \neq j$ in our sample.
%
By the law of large numbers, $n^{-1}{\sum_{i=1}^{n} \indicator{X_{i}\leq t}w(X_i, Y_i)^{-1}}\to F_X(t)/\mathbb{E}_{\jointpdf} \{\wfun(X,Y)\}$ a.s. and $n^{-1}{\sum_{i=1}^{n}w(X_i, Y_i)^{-1}} \to 1/\mathbb{E}_{\jointpdf} \{\wfun(X,Y)\}$ a.s. so by the continuous mapping theorem $\xcdfempir(t) \to \xcdf(t)$ a.s. By similar arguments, $\ycdfempir(s) \to \ycdf(s)$ a.s.

\subsubsection{Case 3: Left Truncation $\wfun(x,y) = \indicator{x<y}$}
In contrast to the former case, when $\wfun(x,y)$ is a truncation function, estimating the marginals under quasi-independence is more challenging due to the actual loss of data and identifiability issues. Nevertheless, for certain types of truncation mechanisms, additional assumptions on the joint distribution may allow to reconstruct the marginals. In particular, the most familiar type of truncation in the statistical literature is left (or right) truncation, described by $\wfun(x,y) = \indicator{x<y}$.
The next proposition shows that under exchangeability, the empirical distribution function of the joint sample $X_1,\ldots,X_n,Y_1,\ldots,Y_n$ is a consistent estimator for the marginals (the proof is in the Supp. Methods, Section \ref{sec:SI_proofs}).
\begin{proposition}
Let $\jointcdf(x,y)$ be an exchangeable joint distribution having a density $\jointpdf(x,y) \equiv \jointpdf(y,x)$ and let $\sample \sim [\jointcdf^{(w)}]^n$ be a sample with the truncation weight function $\wfun(x,y) = \indicator{x<y}$. Let $\hat{F}^{(w),n}_X, \hat{F}^{(w),n}_Y$ be the empirical CDFs of $\xcdf^{(w)}, \ycdf^{(w)}$, respectively. Define:
\be
\xcdfempir(x) = \ycdfempir(x) = \frac{\hat{F}^{(w),n}_X(x)+\hat{F}^{(w),n}_Y(x)}{2}=\frac{1}{2n}\Big[\sum_{i=1}^n \indicator{X_i\le x} + \sum_{i=1}^n \indicator{Y_i\le x} \Big].
\label{eq:exchange_estimator}
\ee
Then $\xcdfempir, \ycdfempir \to \xcdf=\ycdf$  a.s.
\label{prop:consistent_estimator_under_exchangeability}
\end{proposition}

\section{Test Statistics}
\label{sec:teststat}

\subsection{The Adjusted Hoeffding Statistic}

While the permutation and bootstrap approaches can be applied with any test statistic, our goal is to modify an existing omnibus test to weighted models in general and to truncation in particular. For the latter, most tests used to date are tailored to specific alternatives, such as monotone dependence. A recent new approach studied by \cite{chiou2018permutation} can test against a general alternative, but uses either significant computational resources or permutations with different sample sizes, so its significance level is not guaranteed. We are inspired by some popular non-parametric tests of independence such as \cite{thas2004nonparametric}, \cite{heller2012consistent}, and \cite{heller2016consistent},  and for concreteness, we describe the approach of the latter, which generalizes a modified version of \cite{hoeffding1948non}. Our problem requires two major modifications. First, biased sampling should be taken into account when computing the null distribution of the test statistic using a bootstrap or permutations resampling approach; this was addressed in the previous sections. Second, the test statistic compares observed counts  with their expectations under the null, and the computation of these expectations needs to be modified to accommodate biased sampling.

The test belongs to a family of tests which compare the observed counts $o_{\truncregion}$ to the expected counts $e_{\truncregion}$ for different sets $\truncregion \in \mathbb{R}^2$.
As in \cite{heller2016consistent}, our test statistic is based on Pearson's Chi-squared statistics, and we consider all partitions defined by the data $\sample$. Specifically, each data point $(x_i,y_i) \in \sample$ defines a partition of $\mathbb{R}^2$ into four quadrants:
\be
Q_{i}^{jk} \equiv \Big\{(x',y') \in \mathbb{R}^2 \: : \: \indicator{x'>x_i} = j, \indicator{y'>y_i} = k \Big\} \quad   j,k \in \{0,1\}.
\label{eq:data_quartiles_def}
\ee
For example, $Q_{i}^{00} = (-\infty, x_i] \times (-\infty, y_i]$ and $P\big((X,Y)\in Q_{i}^{00}\big) = \jointcdf^{(w)}(x_i,y_i)$.

Let $(x_i, y_i)$ be a point in the sample $\sample$. For a quadrant $Q_i^{jk}$, we denote the observed number of points by $o_{i}^{jk} \equiv o_{Q_i^{jk}}$ and the expected number of points under the null by $e_{i}^{jk} \equiv e_{Q_i^{jk}}$.  We then compute, for each quadrant, the scaled squared difference between the observed and expected number of points under $H_0$. Finally, we sum over all the sample points to get our test statistic: 
%(other aggregation methods like maximization or soft-max instead of summation may also be considered)
\be
T = \sum_{i=1}^n \sum_{j,k \in \{0,1\}} \frac{(o_{i}^{jk}-e_{i}^{jk})^2}{e_{i}^{jk}}.
\label{eq:modifed_hoeffding}
\ee

Estimating the expected values $e_{i}^{jk}$ requires the estimation of the null distribution, which may become highly non-trivial in the biased sampling setting. First, $\xcdf$ and $\ycdf$ may be un-identifiable, and therefore using plug-in estimators of $\xcdf, \ycdf$ may give a poor approximation of the distribution of the test statistic under the null. Second, evaluation of expectations or probabilities under the null may require computationally costly integration of the null distribution, $[\xpdf \ypdf]^{(w)}$, as opposed to a simple multiplication of the empirical marginals in the standard setting. If the expectations $e_{i}^{jk}$ are not estimated correctly, the WP test is still valid according to Corollary \ref{cor:alpha} and the bootstrap approach can be also applied, but power can be severely reduced, as is shown in the Supp. Materials, Section \ref{sec:fast_bootstrap}.

\subsubsection{Computing Expectations under Biased Sampling}
%
%Whether utilizing Pearson's statistic or other test statistics, we face the problem of computing the expected count in a certain cell under the null, given the sample $\sample$ and the resulting permutation distribution $P_{\wmat}$.
A natural approach is to estimate the marginals $\xcdf$ and $\ycdf$ under the null independence model and use them to calculate the expected count in a certain cell. However, as discussed in Section \ref{sec:bootstrap}, this approach works well only for special models. We therefore suggest here an alternative method that directly estimates the expected counts.

Let $P_{ij} \equiv P_{\wmat}(\pi(i)=j) =
\sum_{\pi \in S_n} P_{\wmat}(\pi) \indicator{ \pi(i)=j }$, and define the  Bernoulli random variables $\xi^\pi_{ij} \equiv \indicator{\pi(i) = j }$,
%%
%\begin{align}
%\xi^\pi_{ij} = \left\{ \begin{array}{lll}
%1 & \quad \pi(i) = j \\
%0 & \quad otherwise \\
%\end{array} \right.
%\label{eq:bernoulli_permutations}
%\end{align}
%%
so $P_{ij} = \mathbb{E}(\xi^\pi_{ij})$.
%\begin{comment}
%We then have:
%\be
%P_{ij} = \sum_{\pi \in S_n} P_{\wmat}(\pi) \indicator{ \pi(i)=j } =
%\frac{per(\wmat^{(\neg i, \neg j)})}{per(\wmat)} \wmat(i,j),
%\label{eq:P_ij}
%\ee
%%
%where $\wmat^{(\neg i, \neg j)}$ is the matrix obtained by deleting the $i$-th row and $j$-th column of $\wmat$.  % we can remove it!
%\end{comment}
Let $\truncregion \subset \mathbb{R}^2$ be an arbitrary set. Given a sample $\sample$, for any permutation $\pi$ of the data, denote the number of points in $\truncregion$ under $\pi$ (i.e., after permuting the data set $\sample$) by
$\obsset_{\truncregion}(\pi)= \sum_{i=1}^n \indicator{ (x_i,y_{\pi(i)}) \in {\truncregion} }$,
and let $\expectset_{\truncregion} \equiv \mathbb{E}_{P_{\wmat}} \{\obsset_{\truncregion}(\pi)\}$ be the expected number of data points in $\truncregion$, under the permutations distribution $P_{\wmat}$.

%\end{definition}

The $P_{ij}$ values determine the expected number $\expectset_{\truncregion}$ for any set $\truncregion$ via the following claim:
\begin{claim}
For any $\truncregion \subset \mathbb{R}^2$, the expected number of points $\expectset_{\truncregion}$ under the permutations distribution $P_{\wmat}$ is given by:
\be
\expectset_{\truncregion} = \sum_{i,j=1}^n \indicator{(x_i,y_j) \in \truncregion} \mathbb{E}_{P_{\wmat}}(\xi^\pi_{ij}) = \sum_{i,j=1}^n  \indicator{(x_i,y_j) \in \truncregion} P_{ij} .
\label{eq:expected_region}
\ee
\label{claim:expected_permutations}
\begin{proof}
By definition, we have:
\be
\expectset_{\truncregion} = \mathbb{E}_{P_{\wmat}} \big[\sum_{i=1}^n \indicator{ (x_i,y_{\pi(i)}) \in \truncregion }\big] =
\sum_{i,j=1}^{n} \indicator{(x_i,y_j) \in \truncregion} P_{\wmat}\big( \pi(i)=j \big) =
\sum_{i,j=1}^{n} \indicator{(x_i,y_j) \in \truncregion} P_{ij}.
\ee
\end{proof}
\end{claim}

The probabilities $P_{ij}$ can be easily estimated using the MCMC scheme described in Algorithm \ref{alg:mcmc_permutations}: let $\pi_0$ be the identity permutation and $\pi_1, .., \pi_{\nperm}$ be the sampled permutations, and define the following estimator:
\be
\hat{P}_{ij} \equiv \frac{1}{B+1} \sum_{b=0}^B \indicator{\pi_b(i)=j}.
\label{eq:P_ij_estimator}
\ee
When we sample permutations using an importance distribution $P_{IS}$ (see Section \ref{sec:importance_sampling}),
the above estimator for ${P}_{ij}$ is replaced by:
\be
\hat{P}_{ij}^{(IS)} \equiv \frac{\sum_{b=0}^B \indicator{\pi_b(i)=j} \frac{P_{\wmat}(\pi_b)}{P_{IS}(\pi_b)} }{\sum_{b=0}^B  \frac{P_{\wmat}(\pi_b)}{P_{IS}(\pi_b)}}.
\label{eq:P_ij_estimator_IS}
\ee

Plugging Equation \eqref{eq:P_ij_estimator} (or Equation \eqref{eq:P_ij_estimator_IS}) into Equation \eqref{eq:expected_region} gives an estimator of $\expectset_{\truncregion}$,
\be
{\hat{\expectset}}_{\truncregion} = \sum_{i,j=1}^{n} \indicator{(x_i,y_j) \in \truncregion} \hat{P}_{ij},
\label{eq:estimate_expected}
\ee
 which can be used in the Chi-squared statistic.

For the bootstrap approach, the estimate of the null distribution is used in a straightforward manner. Consider, for example, the bottom-left quadrant with respect to a point $(x_i,y_i)$, $Q_{i}^{00}$; given estimators $\xcdfhat, \ycdfhat$ of the univariate CDFs, a natural estimator for the mass (up to a normalizing constant) that the null puts on $Q_{i}^{00}$ is given by
\be
\widehat{e_{i}^{00}} = n[\xcdfhat \ycdfhat]^{(w)} (x_i,y_i) .
\label{eq:expected_quardant}
\ee

\subsection{An Inverse Weighting Statistic for Strictly Positive $\wfun$}
\label{sec:inverse_weight_stat}
When $\wfun$ is strictly positive, a test can utilize an inverse weighting approach.
For a set  $\truncregion \in \mathbb{R}^2$, define the {\it inverse weighted} observed and expected counts
$o_{\truncregion}^{(\wfun)}$ and $e_{\truncregion}^{(\wfun)}$, respectively:
\begin{align}
o_{\truncregion}^{(\wfun)} &= \sum_{\ell=1}^n \setindicator{\{(x_{\ell}, y_{\ell}) \in \truncregion\}} \wfun(x_{\ell}, y_{\ell})^{-1}, \nonumber \\
e_{\truncregion}^{(\wfun)} &= n E_{[\xcdf \ycdf]^{(w)}} \{ \setindicator{\truncregion} \wfun(X,Y)^{-1} \}.
\label{eq:weighted_obs_exp}
\end{align}

%  for each quadrant $Q_i^{00}$:
%For the partitioned of the sample space defined by the i'th observation, let the contribution of the region% $Q_i^{00}$ to the statistic be
For the quadrants $Q_i^{jk}$, we can compute estimates for the expected weighted counts ${e_i^{jk}}^{(\wfun)} \equiv e_{Q_i^{jk}}^{(\wfun)}$ by multiplying the corresponding marginal inverse weighted counts. For example, for $Q_i^{00}$:
\be
{e_i^{00}}^{(\wfun)} = \frac{(\sum_{\ell=1}^n \wfun(x_{\ell}, y_{\ell})^{-1} \indicator{x_{\ell} \le x_i}) (\sum_{\ell=1}^n \wfun(x_{\ell}, y_{\ell})^{-1} \indicator{y_{\ell} \le y_i}) }{\sum_{\ell=1}^n \wfun(x_{\ell}, y_{\ell})^{-1}} .
\label{eq:weighted_counts}
\ee
Similarly to Equation \eqref{eq:modifed_hoeffding}, the weighted statistic is given by
\be
T^{(w)} = \sum_{i=1}^n \sum_{j,k \in \{0,1\}} \frac{({o_{i}^{jk}}^{(w)}-{e_{i}^{jk}}^{(w)})^2}{{e_{i}^{jk}}^{(w)}}.
\label{eq:weighted_statistic}
\ee

\subsection{Unknown $\wfun$: Left Truncated Right Censored Data}
\label{sec:censoring}
In some applications the biased sampling function $\wfun$ is unknown. However, our methodology is still applicable when  $\wfun$ can be estimated consistently. In particular, we can tackle the important case of left truncation with censoring.
Consider the standard left-truncation right-censoring model where $(X_i,Y_i, C_i)$ are $n$ independent triplets, the joint density of $(X_i,Y_i)$ is proportional to $\jointpdf(x,y)\setindicator{\{x<y\}}$ for some density $\jointpdf$, and $(X_i,Y_i)$ are independent of the censoring variables $C_i$. We observe triplets $(X_i,\min(Y_i,X_i+C_i),\Delta_i$), where $\Delta_i=\setindicator{\{Y_i< X_i+C_i\}}$ are the censoring indicators with $\Delta_i=0$ and $1$ for censored and uncensored observations, respectively.  We suggest testing independence based on the uncensored observations, where censored observations are used only for estimation of the weight function.
Specifically, the conditional density of an uncensored observation is simply
\be
\frac{S(y-x)\jointpdf(x,y)\setindicator{\{x<y\}}}{\mathbb{E}_{\jointpdf}[S( Y- X)\setindicator{\{ X< Y\}}]},
\label{eq:censoring}
\ee
where $S(t)=P(C>t)$ is the survival function of $C$. Thus, the density of uncensored observations has exactly the form of Equation \eqref{eq:weighted_density}, with $w(x,y)=\setindicator{\{x<y\}} S(y-x)$, a function involving a continuous and a truncated part. The methods developed in previous sections are flexible enough to accommodate such functions. The survival function $S$ can be estimated by the standard Kaplan-Meier estimator applied to the data $\{(\min(C_i,Y_i-X_i),1-\Delta_i), \: i=1,\ldots,n \}$.

\section{Simulation Studies}
\label{sec:simulations}
We investigated, using simulation, the performances of the weighted permutation (WP) and bootstrap tests, and compared them to that of \cite{tsai1990testing}'s and the minimum $P_{value}$ (minP2) test of \cite{chiou2018permutation}.
The latter are applicable only to truncated data of the form $\wfun(x,y)=\indicator{x<y}$.

We implemented the simulations in $R$, with time consuming parts implemented in {\it c++} using the $rcpp$ package (\cite{eddelbuettel2011rcpp}). Scripts reproducing all figures and tables are available online at \url{https://github.com/YanivTenzer/TIBS}. P-values for the minP2 test were calculated using the package permDep in R \citep{permDep}, version $1.0.3$ (Aug. 14th, 2019) from \url{https://github.com/stc04003/permDep}. %%at \url{https://github.com/YanivTenzer/TIBS}

We calculated the rejection rate (power) at a significance level $\alpha=0.05$ by averaging results of $500$ replications. As the tests are computationally demanding, we used small sample sizes of $n=100$ and $n=200$ observations for uncensored and censored settings, respectively, in order to perform extensive simulations under different settings. We used $\nperm=1000$ permuted or bootstrap null datasets for all tests except  minP2 for which we used only $\nperm=100$ null datasets and $100$ replications, as it was much slower.
%We also compared the running times of the different tests, which could become a crucial issue when analyzing large datasets, or testing independence for many pairs of random variables, by averaging the times over the $500$ repetitions for each dataset,
The average running times of the tests on a standard laptop with an i7 2.8Ghz dual core Intel processor were $\approx\!0.02$ seconds for Tsai's test,  $\approx\!0.11$ seconds for the new weighted permutation test, $\approx\!2.86$ seconds for the bootstrap  test
%is considerably slower, with $\sim\!2.86$ seconds per test, due to the added computational burden of re-estimating the expectations under the null for each bootstrap sample.
and $\approx\!93.45$ seconds for the minP2 test. % is the slowest, with $\sim\!93.45$ seconds per test under these settings.

\subsection{Truncation, $\wfun(x,y)=\indicator{x<y}$}
\label{sec:truncation}

We study the performances of the tests under truncation for various dependence models with and without censoring. The censoring variable, $C$, was sampled from Gamma distributions, with the shape and scale parameters set such that roughly $25 \%-30\%$ of the observations were censored for each model.

We first simulated data under monotone dependence models with an exchangeable joint distribution, where consistent estimators of the marginals exist (as shown in Section \ref{sec:bootstrap}) and we expect the bootstrap procedure to perform well.  We generated $X$ and $Y$ from a standard bivariate Gaussian distribution with different correlations $\rho$ (Norm($\rho$)). In addition, we generated $X$ and $Y$ with  standard Gaussian marginal distributions under two copula models: (i)
The Gumbel copula ({\bf GC}) with dependence parameter $\theta=1.6$ (Kendall's $\tau = 0.375$), and (ii) The Clayton copula ({\bf CC}) with dependence parameter $\theta=0.5$ ($\tau = 0.2$).
Although both the Gumbel and Clayton copulas produce monotone dependence structures, the two are different in nature - while the former provides upper tail dependence structure the latter produces lower tail dependence \citep{nelsen2007introduction}.
Figure \ref{fig:monotone_exchangeable} in the Supp. Materials presents scatterplots of simulated pairs from the three models.

The results are summarized in Table \ref{tab:simulation_study}, with the test having the highest power shown in boldface. As expected, in the Gaussian settings, under the null distribution (i.e., $\rho=0$), all tests achieve the correct $\alpha=0.05$ error rate. Under the alternative, the bootstrap procedure demonstrates favorable performance in all three settings. In the Gaussian settings, for $\rho<0$ the WP test consistently outperforms minP2. The minP2 has the lowest power in this setting. For $\rho>0$, the WP test has poorer performance, probably due to the difficulty of sampling permutations consistent with the truncation, while Tsai's test shows the second highest power, after the bootstrap.

\begin{table}
	{\tiny	
		\begin{center}
			\scalebox{1}{
				\begin{tabular}{clccccccc}
					\hline
                    & & \multicolumn{4}{c}{Uncensored ($n=100$)} & \multicolumn{3}{c}{Censored ($n=200$)} \\
                    \hline
					Setting  & Model & Tsai & minP2 & WP & Bootstrap & WP & MinP2 & Tsai \\
					\hline
					\multirow{10}{*}{\parbox{3.5cm}{\scriptsize Monotone \\ Exchangeable}}
					& Norm($-0.9$) & {\bf 1}      & {\bf 1} & 	{\bf 1}	& {\bf 1}	 & \bf 1 & \bf 1   & \bf 1    \\
					%& Norm($-0.8$) & {\bf 1}      & {\bf 1} & 	{\bf 1}	& {\bf 1}	     \\
					& Norm($-0.7$) & 0.998	& 0.960	  & 0.998	&  {\bf 1}	  & \bf 1 & 0.828 & \bf 1  	    \\
					%& Norm($-0.6$) & 0.926	& 0.770	  & 0.922	&  {\bf 1}	    	    \\
					& Norm($-0.5$) & 0.764	& 0.510	  & 0.742	&  {\bf 0.998}	& 0.813 & 0.284 & \bf 0.895 \\
					%& Norm($-0.4$) & 0.506	& 0.220	  & 0.512	&  {\bf 0.970}	 \\
					& Norm($-0.3$) &  0.284	& 0.130	  & 0.278	& {\bf 0.828}	& 0.273 & 0.096 & \bf 0.366	 \\
					%& Norm($-0.2$) &  0.158	& 0.060	  & 0.160	& {\bf 0.520}		 \\
					%& Norm($-0.1$) &  0.052	& 0.030	  & 0.064	& {\bf 0.204}		 \\
					& Norm($0.0$) &  0.056	& 0.050	  & 0.064	& 0.058	    &  0.052 & 0.042  & 0.046 \\
					%& Norm($0.1$) &  0.046	& 0.060	  & 0.040	& {\bf 0.132}	   \\
					%& Norm($0.2$) &  0.098	& 0.040	  & 0.070	& {\bf 0.406}	   \\
					& Norm($0.3$) & 0.178	& 0.110	  & 0.118	&  {\bf 0.780}	& 0.097 & 0.083 &\bf 0.202		 \\
					%& Norm($0.4$) & 0.254	& 0.180	  & 0.188	&  {\bf 0.960}		 \\
					& Norm($0.5$) & 0.352	& 0.140	  & 0.194	&  {\bf 1}	& 0.211 & 0.103 & \bf 0.404\\
					%& Norm($0.6$) & 0.440	& 0.130	  & 0.248	&  {\bf 1}	 \\
					& Norm($0.7$) & 0.498	& 0.290	  & 0.262	&  {\bf 1}	  & 0.389 & 0.097 & \bf 0.729	     \\
					%& Norm($0.8$) & 0.594	& 0.210	  & 0.272	&  {\bf 1}	  	     \\
					& Norm($0.9$) & 0.658	& 0.260	  & 0.236	&  {\bf 1}	 & 0.495 & 0.076 & \bf 0.907	    \\
					& GC ($\theta=1.6$) & 0.196 & 0.130 & 0.104 & {\bf 1} & 0.079 & 0.080 & \bf 0.192\\
					& CC ($\theta=0.5$) & 0.110 & 0.130 & 0.074 & {\bf 0.782} & 0.126 & 0.127 & \bf  0.181 \\
					\hline
					\multirow{2}{*}{\parbox{3.5cm}{\scriptsize Monotone \\ Non-Exchangeable}}& LD ($\rho=0.0)$ & 0.042 & 0.010 & 0.046  & \textcolor{gray}{0.996} & 0.005 &0.048 &0.051 \\
					& LD ($\rho=0.4)$ & {\bf 0.634} & 0.330 & 0.578  & \textcolor{gray}{0.704} &  0.060 & 0.056 & \bf 0.250 \\
					\hline
					\parbox{3.5cm}{\scriptsize Non-monotone \\ Exchangeable}
					& CLmix($0.5$) & 0.278 & {0.140} & {\bf 0.412} &  0.338 & \bf 0.398& 0.120 & 0.308 \\ 	
					\hline
					%& Time (sec.) & 0.15 &  0.56 &  4.21 &  0.38 &  0.38 &  19.61 \\
					\multirow{12}{*}{\parbox{3.5cm}{\scriptsize Non-monotone, \\  Non-exchangeable}}
					& CNorm($-0.9$)&   0.992  & {\bf 1} & {\bf 1} & \textcolor{gray}{1} &  \bf 1 & 0.987 & 0.897 \\
					%& CNorm($-0.8$)&  0.916 & 0.930 & {\bf 0.984} & \textcolor{gray}{1} \\
					& CNorm($-0.7$)&  0.844 & 0.950 & {\bf 0.992} & \textcolor{gray}{0.998} &  \bf 0.972 & 0.609 & 0.635 \\
					%& CNorm($-0.6$)&  0.734 & 0.520 & {\bf 0.836} & \textcolor{gray}{0.862} \\
					& CNorm($-0.5$)& 0.514 & 0.600 & {\bf 0.794} & \textcolor{gray}{0.988}  &   \bf 0.626 & 0.262 & 0.321 \\
					%& CNorm($-0.4$)& 0.342 & 0.190 & {\bf 0.408} & \textcolor{gray}{0.686} \\
					& CNorm($-0.3$)& 0.176 & 0{\bf .290} & 0.272 & \textcolor{gray}{0.982} &    \bf 0.212 & 0.138 & 0.157 \\
					%& CNorm($-0.2$)& 0.090 & 0.120 & {\bf 0.130} & \textcolor{gray}{0.654} \\
%					& CNorm($-0.1$)& 0.066 & 0.100 & {\bf 0.082} & \textcolor{gray}{0.986} \\
					& CNorm($0.0$)& 0.042 & 0.020 & 0.046 & \textcolor{gray}{0.986} &     0.047  & 0.053 &  0.055 \\
%					& CNorm($0.1$)& 0.056 & 0.050 & {\bf 0.056} & \textcolor{gray}{0.990} \\
					%& CNorm($0.2$)& 0.078 & 0.050 & {\bf 0.086} & \textcolor{gray}{0.888} \\
					& CNorm($0.3$)& 0.132  & 0.170 & {\bf 0.216} & \textcolor{gray}{0.998} &    \bf 0.228 & 0.035 & 0.089 \\
					%& CNorm($0.4$)& 0.216  & 0.110 & {\bf 0.286} & \textcolor{gray}{0.992} \\
					& CNorm($0.5$)& 0.266 & 0.370 & {\bf 0.654} & \textcolor{gray}{1} & \bf 0.753 & 0.037 & 0.163 \\
					%& CNorm($0.6$)& 0.412 & 0.290 & {\bf 0.672} & \textcolor{gray}{1} \\
					& CNorm($0.7$)& 0.490 & 0.790 & {\bf 0.992} & \textcolor{gray}{1} & \bf 0.999 & 0.080 & 0.268 \\
					%& CNorm($0.8$)& 0.572 & 0.600 & {\bf 0.988} & \textcolor{gray}{1} \\
					& CNorm($0.9$)& 0.698 & {\bf 1} & {\bf 1} & \textcolor{gray}{1} & \bf 1 & 0.298 & 0.367 \\
					\hline
					\hline
				\end{tabular}	
			}
		\end{center}				
		\caption{Power at a significance level of $\alpha=0.05$ for left-truncated data ($\wfun(x,y)=\indicator{x<y}$) for uncensored (left 4 columns) and censored (right 3 columns) models.  Norm($\rho$) - Bivariate normal distribution with correlation $\rho$. GC - The Gumbel Copula, CC - the Clayton Copula, with dependence parameter $\theta$. LD - Lifetime Distribution with correlation $\rho$, CLmix - an exchangeable mixture of two Clayton copulas, CNorm($\rho$) - non-exchangeable joint distribution with dependence parameter $\rho$. \label{tab:simulation_study}}
	} % end scriptsize 	
\end{table}

Following \cite{chiou2018permutation}, we next simulated data from a lifetime distribution (LD) where the joint distribution is non-exchangeable having marginal distributions $X \sim exp(5)$ and $Y \sim Weibull(3,8.5)$. We specified the dependence of $(X,Y)$ through a normal copula, where the strength of dependence is determined by the correlation parameter $\rho$. We simulated data under independent ($\rho=0$) and dependent ($\rho=0.4$) $\jointcdf$ before truncation.  Figure \ref{monotone_non_exchangeable} in the Supp. Materials displays scatterplots of simulated pairs from both models.

The results are shown in the second part of Table \ref{tab:simulation_study}.
Although $X$ and $Y$ are not exchangeable, we applied the bootstrap procedure (shown in light gray) as well, using marginal estimates according to Equation \eqref{eq:exchange_estimator}, in order to investigate the impact of model miss-specification on its performance. The devastating impact of model miss-specification under the null distribution for the bootstrap procedure is now apparent: the test does not retain the desired rejection rate $\alpha=0.05$ under the null hypothesis ($\rho = 0$). For $\rho=0.4$, Tsai's test has the highest power, followed by the WP test and then minP2.

The third simulation study evaluates the performance of the tests under non-monotone dependence. Starting with a non-monotone exchangeable model, we simulated data from a mixture of two Clayton copulas with dependence parameters $\theta=0.5 \hspace{0.03in}(\tau = 0.2)$ and $\theta=-0.5 \hspace{0.03in}(\tau=-0.333)$, respectively, and equal population proportions. Figure \ref{non_monotone_exchangeable} in the Supp. Materials presents scatterplots of simulated pairs from the model. The third part of Table \ref{tab:simulation_study} presents the results. Our bootstrap approach has the highest power, followed by the WP test, Tsai's test and lastly minP2. %All tests have relatively low power under this model, probably due to the small sample size of $n=100$ observations.
It is somewhat surprising that Tsai's test outperforms here minP2, as the former is tailored to monotone alternatives. %

Finally, we considered a non-monotone and non-exchangeable model.
We used a normal copula, CNorm($\rho$), with varied correlation coefficient $\rho$ to specify the joint distribution of $X$ and $Y$, where $X \sim Weibull(0.5,4)$ and $Y \sim U[0,16]$,  set such that $\mathbb{E}(X)=\mathbb{E}(Y)=8$, and retained only pairs satisfying $Y \geq X$. Figure \ref{fig:non_monotone_non_exchangeanle} in the Supp. Materials displays scatterplots of simulated pairs from the models.
The last part of Table \ref{tab:simulation_study}  presents the results. As expected, the bootstrap procedure (light gray) does not retain the desired rejection rate under the null hypothesis. Under the alternative, both the WP and minP2 tests outperform Tsai's procedure across the entire range of the dependence parameter $\rho$. This behaviour is expected because both tests were designed to detect non-monotone dependency. Our method has the highest power for all values of $\rho$. minP2 is more powerful than Tsai's test for strong (absolute) correlations and less powerful for weaker correlations.

%\subsection{Left Truncation Right Censoring}
%\label{sec:truncation_censoring}
%Next we consider the left truncation right censoring regime. Recall that in this setting the joint density is of the form $f^{(w)}_{X,Y}(x,y) \propto S(y-x)f(x,y)\setindicator{\{x<y\}}$, where $S(t)=P(C>t)$ is the survival function of $C$, a censoring variable independent of $X, Y$ (see Section ~\ref{sec:censoring} for details).
%Our weighted permutations test can be adjusted to accommodate such a setting, by first estimating $S(t)$, using the censored samples, then applying the test to the uncensored observations, with the estimated bias function $\hat{w}(x,y) \equiv \indicator{x<y} \hat{S}(y-x)$. To evaluate the performance of this approach, we compare it with minP2 and Tsai's tests, that similarly to our method, can accommodate censored data. In our simulations we use the same distributions considered in previous section \zuk{All of them?} \yaniv{yes, all of them.}, however, this time with censoring introduced into the data generation process. In each setting, the censoring variable $C$ was sampled from the gamma distribution, with different shape and scale parameters for each setting, so that roughly $25 \%-30\%$ of the samples are censored.
%%\zuk{What is the distribution of $C$?}
%Table \ref{tab:simulation_study_LTRC} shows the result achieved by the various methods.

\subsection{Strictly Positive Bias Functions}
\label{sec:strictly_positive}
Our new tests can be used to detect dependency in a general weighted model. To study the performance of our approach, we consider two cases of positive biased sampling.
In the first case we took $w(x,y)=x+y$, where $X$ and $Y$ were sampled from the log-normal bivariate distribution with zero mean, unit variance and correlation $\rho$. Recall that for a strictly positive $\wfun$, the inverse weighted Hoeffding statistic from Section \ref{sec:inverse_weight_stat} can be used as an alternative to the adjusted Hoeffiding statistic. We also compared here the bootstrap and MCMC approaches to the importance sampling approach from Section \ref{sec:importance_sampling}, with four different importance sampling distributions, described in the Supp. Methods, Section \ref{sec:IS_appendix}.
We applied six different sampling methods for p-value calculation for each of the two test statistics, resulting in twelve different tests. %The methods used were the permutation MCMC approach, the Bootstrap, and importance sampling of permutations with four different distributions.
% therefore end up with three additional tests: the Inverse-Weighted Permutations test (IWP), the Inverse-Weighted Importance-Sampling Permutations test (IWISP) and the Inverse-Weighted Bootstrap test (Bootstrap-IW).
When applying the Bootstrap, we estimated the univariate marginals using the weighted estimators in Equation \eqref{eq:marginal_estimation_under_positive_weight_function}.
Table \ref{tab:strictly_positive} shows the rejection rates of the various tests at a significance level $\alpha=0.05$, for sample size $n=100$ and for $\rho=0$ (independence) and $\rho=0.2$. All tests except the bootstrap seem to maintain the significance level at approximately $\alpha =0.05$ under the null. Under the alternative, the importance sampling approach with a uniform and 'grid' importance distribution is most powerful, but the MCMC approach is not far behind.

In the second example, we considered a bivariate Gaussian distribution $Norm(\rho)$ for $X$ and $Y$ as in Section \ref{sec:truncation}, and with $\wfun(x,y) \propto Norm(-\rho)$. The biased sampling function here masks the dependence so the observed pairs, $(X_i,Y_i)$, are independent Gaussian random variables. This example shows that biased sampling can not only create spurious dependencies, but can also mask true dependencies. Nevertheless, knowing $\wfun$ we can apply our tests and detect the dependence, as is shown in Table \ref{tab:strictly_positive}.
As expected, for all tests, as $\rho$ increases (we used only $\rho>0$ due to symmetry) the power increases quite rapidly. The importance sampling with a uniform distribution is usually the most powerful, with the MCMC approach very close or superior for strong correlation. The bootstrap approach shows poor performance for this case. As in the previous example, here too the inverse weighting statistic is inferior to the adjusted Hoeffding statistic. The relative success of the importance sampling schemes for both examples can be explained by the small sample size. As shown in the Supp. Materials, Section \ref{sec:IS_appendix}, when the sample size increases the importance sampling distributions become unrepresentative of the distribution $P_{\wmat}$, resulting in poor performance, whereas the MCMC approach is much more robust to changes in sample size.

\begin{table}[tb]
	\centering
	\begingroup\tiny
	\begin{tabular}{rrrrrrr}
		\hline
Model		& WPIS & WPIS & WPIS & WPIS & WP & Bootstrap \\
IS-Dist.		& Kou-McCullagh & Uniform & Monotone & Grid &  &  \\
		\hline
	LogNormal($\rho=0$)	 & 0.048 & 0.050 & 0.016 & 0.050 & 0.056 & 0.072 \\
	IW	 & 0.050 & 0.038 & 0.016 & 0.036 & 0.060 & 0.080 \\
	LogNormal($\rho=0.2$) & 0.606 & {\bf 0.676} & 0.004 & {\bf 0.676} & 0.602 &  0.632 \\
	IW	 & 0.384 & 0.432 & 0.000 & 0.428 & 0.382 & 0.414 \\	
	\hline
	Norm(0.0) & 0.046 & 0.048 & 0.032 & 0.040  & 0.046 & 0.004 \\
		IW    & 0.056 & 0.058 & 0.032 & 0.054 & 0.050 & 0.004 \\
	Norm(0.1) & {\bf 0.084} & {\bf 0.084} & 0.052 & 0.080 & 0.082 &  0.006 \\
	IW        & 0.076 & 0.076 & 0.032 & 0.076 & 0.078 &  0.008 \\
%	Norm(0.2) & 0.146 & {\bf 0.152} & 0.080 & 0.134 & 0.142 &  0.016 \\
%	IW        & 0.132 & 0.140 & 0.066 & 0.132 & 0.138 &  0.014 \\
	Norm(0.3) & 0.296 & 0.322 & 0.134 & {\bf 0.350} & 0.302 & 0.052 \\
	IW        & 0.282 & 0.296 & 0.136 & 0.318 & 0.298 & 0.024 \\
%	Norm(0.4) & 0.368 & {\bf 0.436} & 0.174 & 0.434 & 0.358 & 0.082 \\
%	IW        & 0.350  & 0.366 & 0.160 & 0.378 & 0.352 & 0.032 \\
	Norm(0.5) & 0.582 & {\bf 0.640} & 0.298 & 0.630 &  0.582 &  0.166 \\
	IW        & 0.502 & 0.526 & 0.228 & 0.532 & 0.492 & 0.038 \\
%	Norm(0.6) & 0.770 & 0.776 & 0.428 & {\bf 0.802} & 0.770 & 0.230  \\
%	IW        & 0.660 & 0.662 & 0.326 & 0.636 & 0.644 &  0.046 \\
	Norm(0.7) & 0.856 & {\bf 0.860} & 0.516 & 0.850 & 0.852 &  0.334  \\
	IW        & 0.730 & 0.706 & 0.456 & 0.684 & 0.726 &  0.078 \\
%	Norm(0.8) & 0.934 & 0.886 & 0.654 & 0.902 & {\bf 0.938} &  0.526 \\
%	IW        & 0.808 & 0.756 & 0.536 & 0.768 & 0.806 &  0.086 \\
	Norm(0.9) &  {\bf 0.964} & 0.914 & 0.712 & 0.918 & {\bf 0.964} &  0.672 \\
	IW        & 0.876 & 0.800 & 0.662 & 0.786 & 0.878 & 0.114 \\
		\hline
	\end{tabular}
	\caption{Strictly positive bias functions; estimated power at a significance level of $\alpha=0.05$, and sample size $n=100$ of the permutations (WP), uniform importance sampling permutations (WPIS), and bootstrap tests. For each parameter settings and sampling method two statistics were applied, shown in separate lines: The adjusted Hoeffding statistic, with expectations estimated using permutations/bootstrap samples, and the inverse weighting statistic (IW).  The top four rows represent a LogNormal distribution with $\wfun(x,y)=x+y$. The bottom rows represent a Gaussian distribution with correlation $\rho$, with $w(x,y)$ proportional to a Gaussian density with correlation $-\rho$. 	\label{tab:strictly_positive}}
	\endgroup	
\end{table}

%\input{experiments_new}
%###############################################################################################################################
\section{Real-Life Datasets}
\label{sec:real_data}

We applied the various tests to four data sets, shown in Figure \ref{fig:real_datasets}. P-values for the WP and bootstrap tests are based on $\nperm=10^5$ samples, and for the minP2 test on $\nperm=10^4$, due to computational restrictions.

%used $M_0=??$ 'burn-in' and $M=??$ skip parameters \zuk{To fill values}.

\begin{enumerate}
\item \textbf{Time from Infection to AIDS in HIV Carriers} - A classical example of truncated data occur in AIDS retrospective studies, where the time from HIV infection to AIDS ($Y$) is restricted to be smaller than the time from HIV to sampling ($X$) \citep{lagakos1988nonparametric}. Here we analyze data on 295 AIDS cases, available in the DTDA package of R \citep{moreira2010dtda}. By design, the sample comprised only patients satisfying $0\leq X\leq Y$. The new WP, Tsai's and minP2 tests all obtained significant P-values of $0.001$, $0.005$, and $0.002$, respectively, suggesting that dependence exists between the two time variables. 
To examine the effect of considering the truncation mechanism when testing for independence, we also performed a WP test with a constant $\wfun$, and the P-value remained significant and in fact was reduced to $10^{-5}$.

% We remove this one
\begin{comment}
\item \textbf{Length of Hospitalization in Intensive-Care-Units} - Data on $137$ patients that were hospitalized in Intensive-Care-Units (ICUs) on a random day were collected in five Israeli hospitals (see \cite{mandel2010competing}). Let $Y$ and $X$ denote the length of stay in the ICU and the time from admission to sampling. Since only patients having $0 \leq X \leq Y$ are observable, the data are truncated with $\wfun(x,y)=\indicator{x \leq y}$. Previous works (e.g. \cite{mandel2010competing}) analyzed the data assuming that $X$ and $Y$ are quasi-independent, and it is therefore important to test the validity of this assumption. The new WP test, Tsai's test and minP2 test were all insignificant at the standard $\alpha=0.05$ significance level, with P-values of $0.298, 0.108$ and $0.346$, respectively.
%We also compared the running times of the three tests, which were $16.5, 0.09$ and $1370.3$ seconds for WP, Tsai's test and minP2, respectively (recall that minP used $10$-fold fewer randomized samples, and would have taken $\approx 13,700$ seconds with $B=100,000$).
\end{comment}

\item \textbf{Survival in the Channing House Community} - \cite{hyde1977testing} analyzed survival data for residents of the Channing House retirement community in Palo Alto, CA; the full dataset is available in the boot $R$ package \citep{R_boot}. The survival time, $Y$, of a resident is left truncated by the entering age to the community, $X$, and is right censored by the age at the end of followup. After removing five observations that were not consistent with the criterion $X<Y$, the data consisted of $n=457$ individuals, 282 ($61.7\%$) of which were censored. Quasi independence of survival time and entering time was tested applying the approach described in Section \ref{sec:censoring}.
%The entire sample was used to estimate the survival function $S(t)$ by the standard product-limit approach for left-truncated right-censored data.
%hen, we used the $175$ uncensored observations and applied the WP test for quasi independence using the weighting function $\wfun(x,y) =\setindicator{\{x<y\}} S(y-x)$ (see Equation \eqref{eq:censoring}).
Tsai's and minP2 tests were calculated for comparison.
The WP, Tsai's and minP2 tests all obtained non-significant P-values of
$0.854$,  $0.099$ and $0.140$, respectively, showing no evidence for dependence between entering age and survival. As expected, when we ignored the truncation mechanism, a naive WP test of independence considering only the censoring and using $\wfun(x,y)=S(y-x)$ (see Equation \eqref{eq:censoring}) yielded a spurious signal of dependence, with a P-value of $10^{-5}$, demonstrating the need to account for biased sampling when testing.

\item \textbf{Time before and after infection in Intensive-Care-Units} -
Data on $137$ patients that were hospitalized in Intensive-Care-Units (ICUs) on a random day were collected in five Israeli hospitals as part of a national cross-sectional study \citep{mandel2010competing}. Infection data were collected from admission to the ICU until discharge or 30 days, whatever comes first. An important question was whether the time of infection is associated with the remaining time in the ICU.
To test this independence hypothesis, we use the sub-sample of patients who admitted to the ICU without delay and who acquired infections during their first 30 days of hospitalization in the ICU. Thus, we test independence of $X$, the time from admission to the ICU to infection and $Y$, the time from infection to discharge from the ICU. Due to the sampling mechanism, the data are length biased according to the total length of stay in the ICU, yielding the weight function $\wfun(x,y)=x+y$.
The new WP and bootstrap tests were both significant at the standard $\alpha=0.05$ level, with P-values $0.039$ and $0.031$, respectively, indicating that the time of acquiring infection shows a significant effect on prolonging the remaining time in the ICU.
When we used a WP test while ignoring the biased sampling function $\wfun$, the signal for dependence disappeared, and we got a P-value of $0.605$. Thus, for this dataset the biased sampling masks the true dependence between $X$ and $Y$, and our test that takes $\wfun$ into account was able to reveal it.

%The running times were $6.3$ and $80.6$ seconds for the WP and bootstrap tests, i.e. the WP was about $13$-fold faster, mainly because under the bootstrap we need to re-estimate the marginal distributions for each bootstrap sample.

\item \textbf{Time to Promotion to the Rank of Full Professor} - The data consists of cross-sectional records on all faculty members of the Hebrew University of Jerusalem who were employed in $1998$. We tested whether the age at promotion to the associate professor rank depends on the service time in that rank. Let $A_{AP}$ and $A_{FP}$ denote the age at promotion to the ranks associate professor (AP) and full professor (FP) respectively; we test independence between $X = A_{AP}$ and $Y = A_{FP}−A_{AP}$ for associate professors promoted to full professor before the age of $65$, back then the retirement age in Israel. We used the sub-sample of $306$ faculty members who were promoted to the FP rank after 1980 and were younger than $65$ at sampling time (1998). As in \cite{mandel2012cross}, we assume that professors will stay in the university until age $65$. Assuming a stable entrance process to the AP rank, the cross-sectional study design leads to length biased sampling according to the length of service at the FP rank. The restriction of the data to professors who promoted after $1980$ resulted in the weight function $\wfun(x,y) = \min(65−x-y, 18) \indicator{x+y<65}$%, where $\indicator{x+y<65}$ indicates that faculty members older than 65 at sampling time were not included in the sample. 
Thus, the weight is neither a truncation function nor strictly positive, and the only test applicable is the permutation test of Section \ref{sec:permutations}.

We applied the WP test with $\nperm=10^5$ permutations, % The running time was $121.3$ seconds.
and obtained a very small P-value of $0.00002$, meaning that the age of promotion to associate professor rank does depend on the service time in that rank.
Ignoring the biased sampling function $\wfun$ still yielded a significant P-value of $0.00060$.
\end{enumerate}

\begin{comment}
\vspace{-0.3in}
\begin{figure}[H]
	%\begin{minipage}{1.\columnwidth}
	\begin{center}		
		\begin{tabular}{cccc}
						\hspace{-0.1in}$AIDS$ & \hspace{0.1in}$Channing$ $House$ & $Huji$ & $Infection$ \\
			\hspace{-0.3in}
			\includegraphics[width=0.26\columnwidth]{Figures/real_data/AIDS.png}
			&
			\hspace{-0.2in}
			\includegraphics[width=0.26\columnwidth]{Figures/real_data/ChanningHouse.png}
			&	%	\\				
			\hspace{-0.2in}
			\includegraphics[width=0.26\columnwidth]{Figures/real_data/huji.png}
			&
%			&
%			\includegraphics[width=0.3\columnwidth]{../Figures/real_data/ICU.png}\\ % & \hspace{-0.15in}
			\hspace{-0.2in}
			\includegraphics[width=0.26\columnwidth]{Figures/real_data/Infection.png}			
		\end{tabular}
	\end{center}
	\vspace{-0.4in}
	\caption{\footnotesize Scatterplots of four real datasets. Red points indicate the true data. Black '+' signs represent points $(x_i, y_{\pi(i)})$ for randomly drawn permutations $\pi$ under biased sampling: $\wfun(x,y)= \indicator{x<y}$ for the AIDS dataset, $\wfun(x,y)=\setindicator{\{x<y\}} S(y-x)$ for the Channing housing dataset, $\wfun(x,y)=min(65−x-y, 18) \indicator{x+y<65}$ for the Huji dataset and $\wfun(x,y)=x+y$ for the Infection dataset.}	
	\label{fig:real_datasets}	
\end{figure}
\end{comment}

\vspace{-0.3in}
\begin{figure}[h]
	%\begin{minipage}{1.\columnwidth}
	\begin{center}		
		\begin{tabular}{cc} %cc
			\hspace{-0.4in}$AIDS$ & \hspace{-0.3in}$Channing$ $House$ \\  
			\hspace{-0.3in}
			\includegraphics[width=0.45\columnwidth]{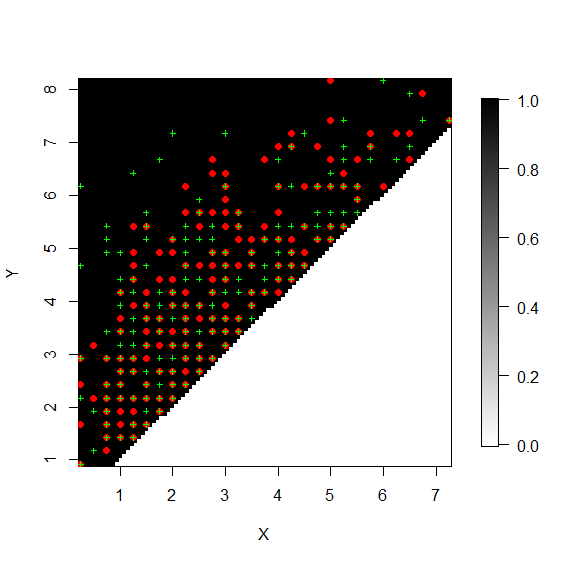}
			&
			\hspace{-0.2in}
			\includegraphics[width=0.45\columnwidth]{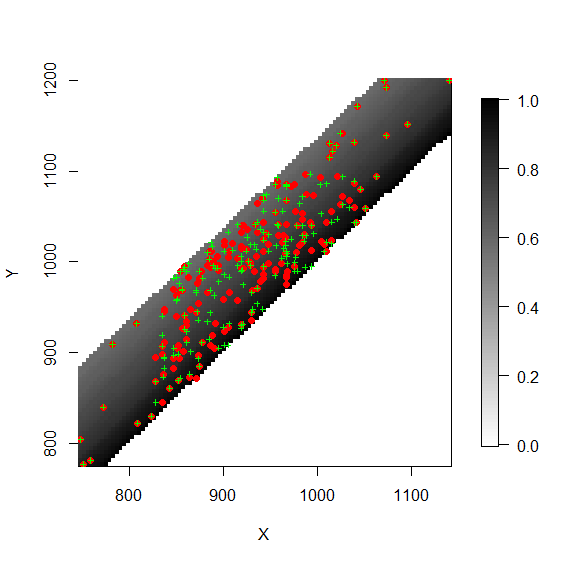} \\
			\hspace{-0.4in}	$Huji$ & \hspace{-0.5in} $Infection$ \\
%			&	%	\\				
			\hspace{-0.2in}
			\includegraphics[width=0.45\columnwidth]{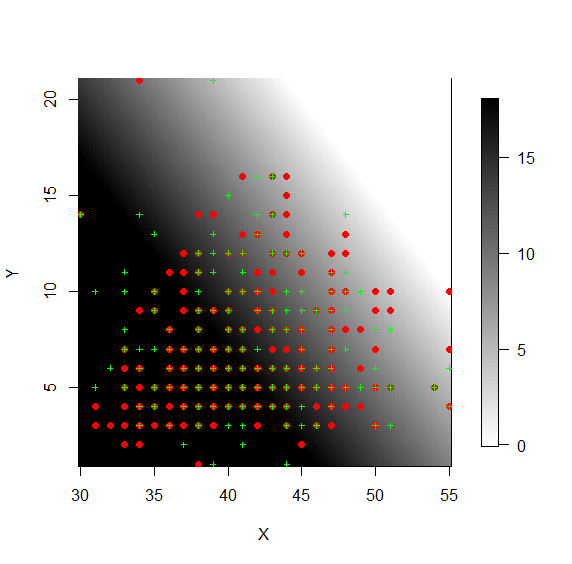}
			&
			%			&
			%			\includegraphics[width=0.3\columnwidth]{../Figures/real_data/ICU.png}\\ % & \hspace{-0.15in}
			\hspace{-0.2in}
			\includegraphics[width=0.45\columnwidth]{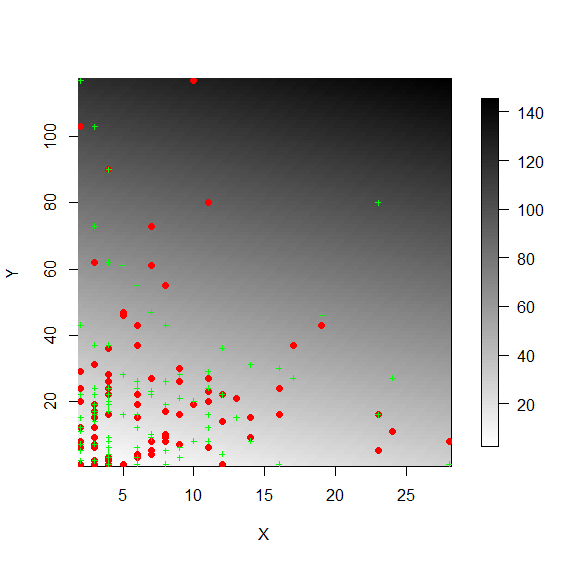}			
		\end{tabular}
	\end{center}
	\vspace{-0.4in}
	\caption{\footnotesize \footnotesize Scatterplots of four real datasets. Red points indicate the true data. Green '+' signs represent points $(x_i, y_{\pi(i)})$ for randomly drawn permutations $\pi$ under biased sampling. The biased sampling functions are shown as background in grayscale, with $\wfun(x,y)= \indicator{x<y}$ for the AIDS dataset, $\wfun(x,y)=\setindicator{\{x<y\}} S(y-x)$ for the Channing housing dataset, $\wfun(x,y)=min(65−x-y, 18) \indicator{x+y<65}$ for the Huji dataset and $\wfun(x,y)=x+y$ for the Infection dataset.}	
	\label{fig:real_datasets}	
\end{figure}

\section{Discussion}
\label{sec:discussion}

%In this work we tackle the problem of testing quasi-independence under a general biased sampling regime. We introduce two new tests, namely, the weighted permutations and bootstrap tests. These tests are based on a scan statistic that measures the deviation between the observed and expected number of points within a set of quadrants determined by the data. Many other tests can be considered, utilising our MCMC algorithm.

%As far as we know, our work is the first to consider testing quasi-independence of a general weighted model. Previous works focus on testing independence only under different truncation models. We demonstrate the merit of our proposed tests, using simulated and real-life data sets, and showed that even for truncated data they attain similar and often higher power in most settings considered here, compared to the recent minP2 test. Our test can handle censored observations under truncation, and it would be interesting to explore this approach further to cases of censoring and general biased sampling.

%Lastly, we conjecture that, under the assumption of quasi-dependence, both the WP and bootstrap test are consistent. To prove that, the limit of the permuted data distribution under $P_{\wmat}$ should be analyzed for general dependence structure. This is a difficult problem that is beyond the scope of the current paper.
%%%%%%%%%%%%%%%%%%%%%%%%%%%%%%%%%%%%%%%%%%%%%%%%%%%%%%%

In this paper we address the problem of testing quasi-independence in the presence of a general bias function and in possibly
non-monotone settings. 
%Among the statistical procedures that assume quasi-independence assumption to hold, are the popular
%Cox's proportional hazard model, the familiar product-limit estimator for truncated data and others.
%It is therefore essential to establish quasi-independence before applying those procedures.
As demonstrated using real-life data sets, testing independence naively, while ignoring the bias function, can either create spurious dependence or mask true dependence.

We introduce two general machineries to simulate samples under the biased null distribution, namely, the permutations and bootstrap approaches, and examine the challenges possessed by both.
Concretely, the former requires drawing permutations from a (general) non-uniform distribution, while the latter requires consistent estimation of the univariate marginals under the alternative hypothesis.
We tackle the first challenge by utilising an MCMC scheme and an importance sampling methodology.
Our simulation study indicates that for large sample size, the latter approach may suffer some degradation in performance, due to the difficulty in finding a proposal distribution suited for a general bias function.

On the bootstrap front, we identify two settings in which consistent estimators can be derived, both under the null and the alternative.
We also introduce a new algorithm for estimating the marginal CDFs under the null, for a general bias function.
This is of independent interest in cases where the null hypothesis is not rejected.

Importantly, both the permutations and bootstrap approaches can be combined with different statistics, thus result in a different test.
As shown in simulations and real-life data sets, the choice of statistic can affect the power of the resulting test.

%next, censoring
An appealing feature of the methodologies developed here is that they can be easily adapted to cases where the bias function is not known, but can be estimated from the data.
An important instance of such cases is that of censoring. In particular, in observational studies, left truncation right censoring settings are frequently encountered.
Our tests can accommodate such settings and it would be interesting to explore this approach further to more general cases of censoring and biased sampling.

We demonstrate the merit of our proposed tests, using simulated and real-life data sets, and showed that even for truncated data they attain similar and often higher power in most settings considered here, compared to the recent minP2 test.

Lastly, some theoretical aspects of the proposed algorithms are yet to be explored and left for future work.
In particular we conjecture that, under the assumption of quasi-dependence, both the WP and bootstrap test are consistent. 
\bibliographystyle{plainnat}
%\setcitestyle{square}
\newpage
\clearpage
\bibliography{bib_biased_sampling_R1}
\newpage
\clearpage
\section*{Supplementary Materials} %  for Testing Independence with Biased Sampling}
\label{sec:supplement_material}

\renewcommand{\thesubsection}{\Alph{subsection}}

\subsection{Proofs}
\label{sec:SI_proofs}

The proof of Claim 	\ref{claim:P_W_conditional} is brought below:
\begin{proof}
	For a general weighted model, we have
	\be
	\label{eq:generalperm}
	P(\pi(\sample) \mid x, \sample_y) = \frac{\prod_{i=1}^n \jointpdf^{(w)}(x_i, y_{\pi(i)}) }{\sum_{\pi'\in S_n} \prod_{i=1}^n \jointpdf^{(w)}(x_i, y_{\pi'(i)}) }.
	\ee
	Under the null, $\jointpdf^{(w)}(x,y)\propto w(x,y)\tilde{\xpdf}(x)\tilde{\ypdf}(y)$, hence
	\begin{eqnarray*}	
		P_0(\pi(\sample) \mid x, \sample_y) = \frac{\prod_{i=1}^n \tilde{\xpdf}(x_i) \tilde{\ypdf}(y_{\pi(i)}) \wmat(i,\pi(i))}{\sum_{\pi' \in S_n} \prod_{i=1}^n \tilde{\xpdf}(x_i) \tilde{\ypdf}(y_{\pi'(i)}) \wmat(i,\pi'(i))}
		= P_{\wmat}(\pi).
	\end{eqnarray*}	
\end{proof}

\noindent The proof of Proposition \ref{prop:consistent_estimator_under_exchangeability} is brought below:
\begin{proof}
	% Proof from Micha:
	Since $X$ and $Y$ are continuous exchangeable random variables, $P(X< Y)=1/2$ and therefore
	$$
	\jointpdf^{(w)}(x,y)=2\indicator{x < y} \jointpdf(x,y).
	$$
	Calculating the weighted marginal, we have
	$$
	\xpdf^{(w)}(t)=2\int_t^\infty \jointpdf(t,y)dy=2\int_t^\infty \jointpdf(y,t)dy,
	$$
	due to exchangeability. Similarly
	$$
	\ypdf^{(w)}(t)=2\int_{-\infty}^t \jointpdf(x,t)dx .
	$$
	Thus (due to the continuity assumption),
	$$
	\frac{\xpdf^{(w)}(t)+\ypdf^{(w)}(t)}{2}=\int_t^\infty \jointpdf(y,t)dy+\int_{-\infty}^t \jointpdf(x,t)dx = \ypdf(t).
	$$
	The result follows by applying the Glivenko–Cantelli Theorem on $\hat{F}^{(w),n}_X$ and $\hat{F}^{(w),n}_Y$ .
\end{proof}
%
%\end{proposition}
%

\subsection{Simulated Data}
%\subsection{Monotone Settings}
This section of the supplementary materials contains the estimated power and scatterplots for simulated data of the various settings presented in Section \ref{sec:simulations}.

\subsubsection*{Monotone-Exchangeable Joint Distribution}
\begin{figure}[!h]
	%\begin{minipage}{1.\columnwidth}
	\begin{center}		
		\begin{tabular}{cccc}			
			\hspace{0.01in}$Norm(\rho=0.9)$ & \hspace{0.07in} $Norm(\rho=-0.9)$ & $GC(\theta=1.6)$ & \hspace{0.15in}$CC(\theta=0.5)$  \\
			\hspace{-0.3in}
			\includegraphics[width=0.26\columnwidth]{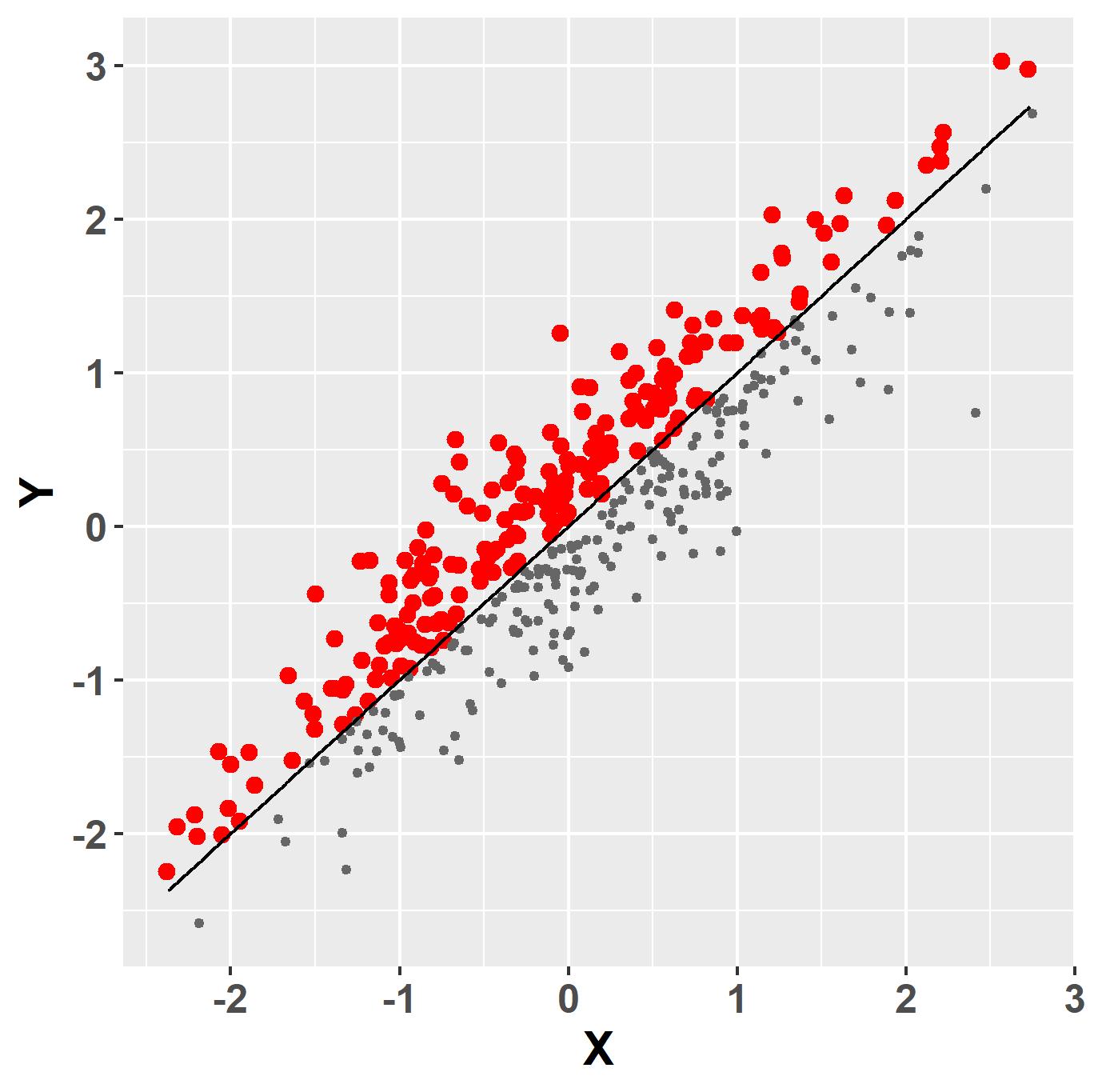}	% new figures			
			&
			\hspace{-0.2in}
			\includegraphics[width=0.26\columnwidth]{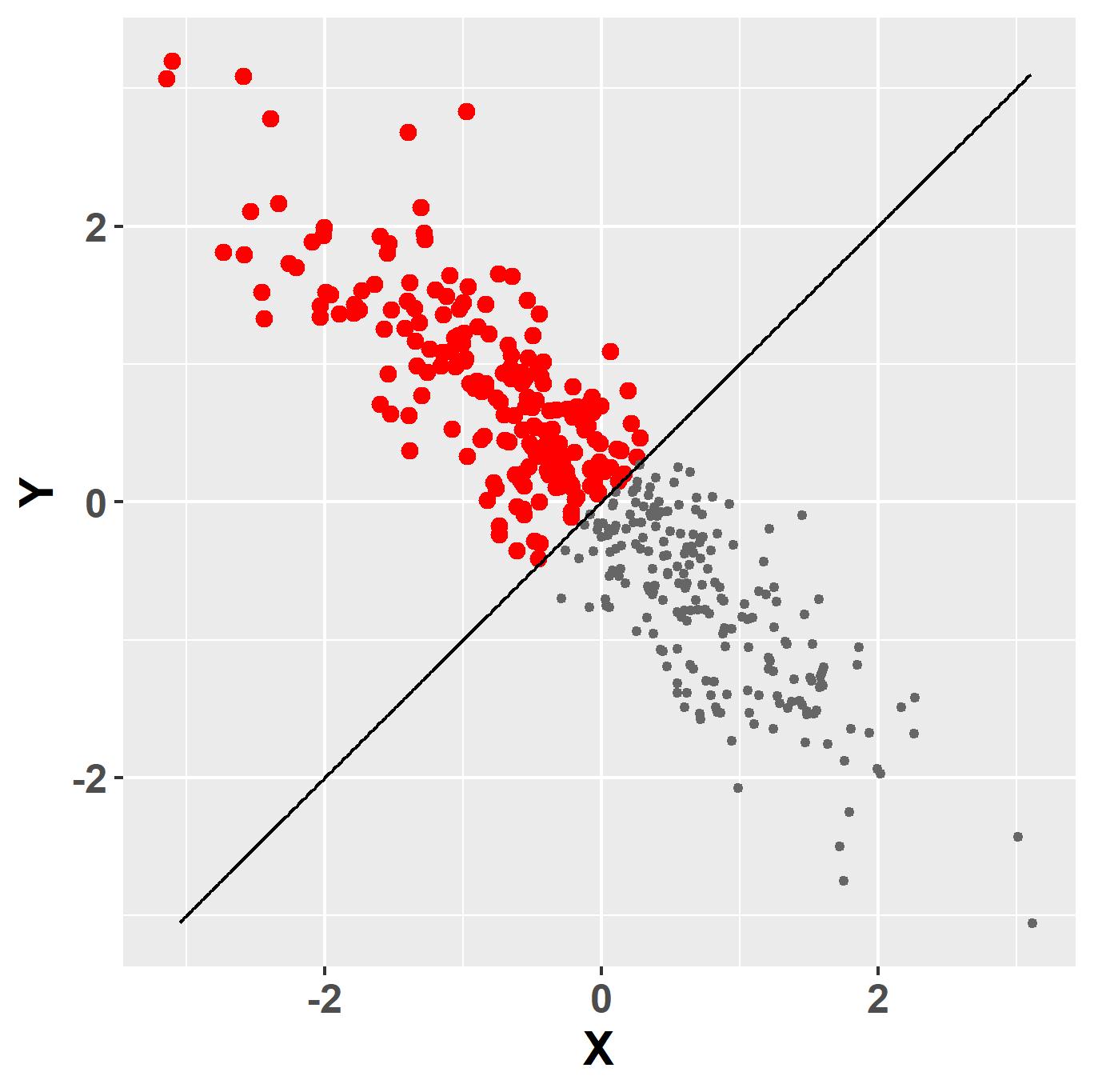}				
			&
			\hspace{-0.2in}
			\includegraphics[width=0.26\columnwidth]{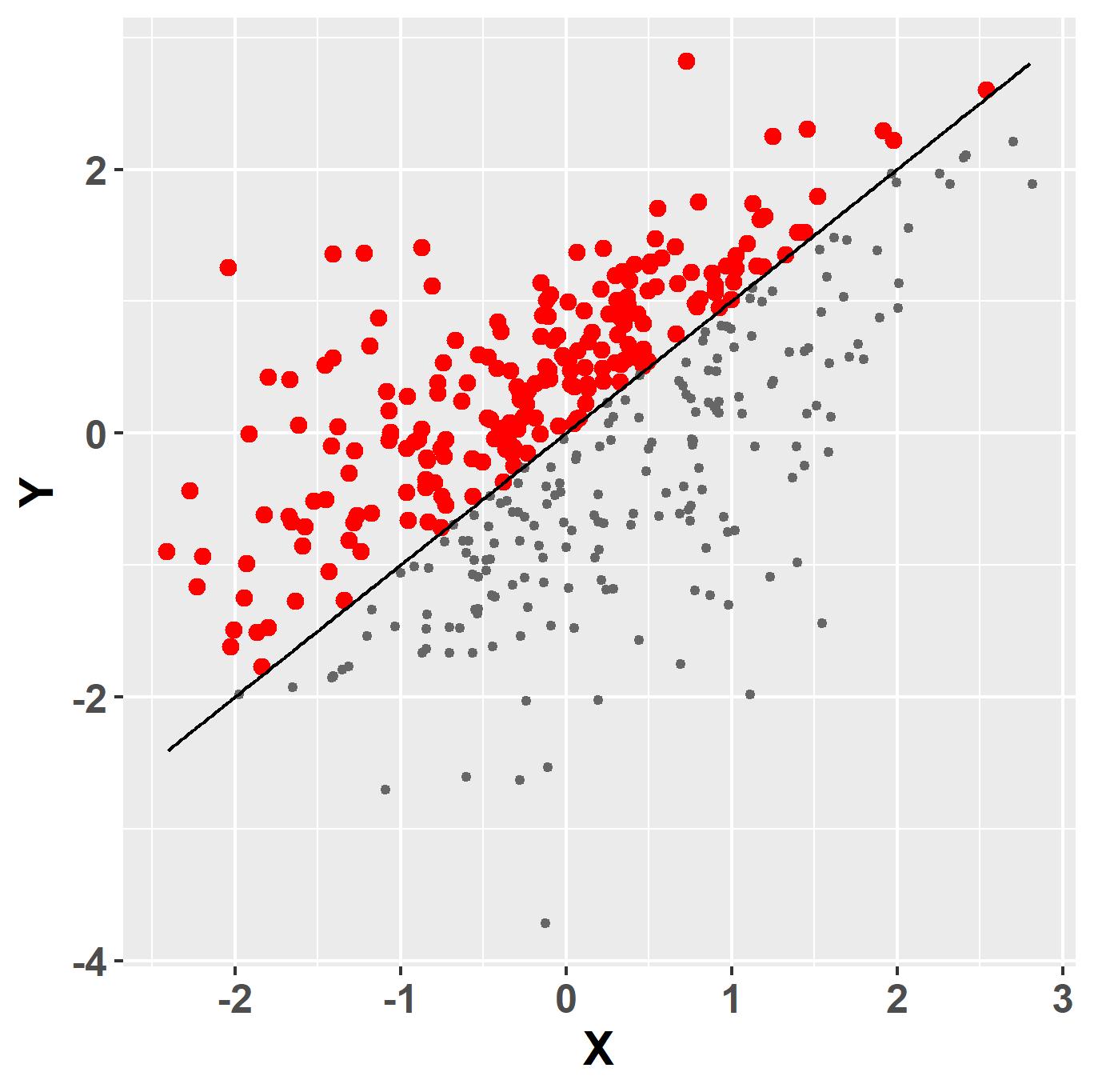}
			&
			\hspace{-0.2in}
			\includegraphics[width=0.26\columnwidth]{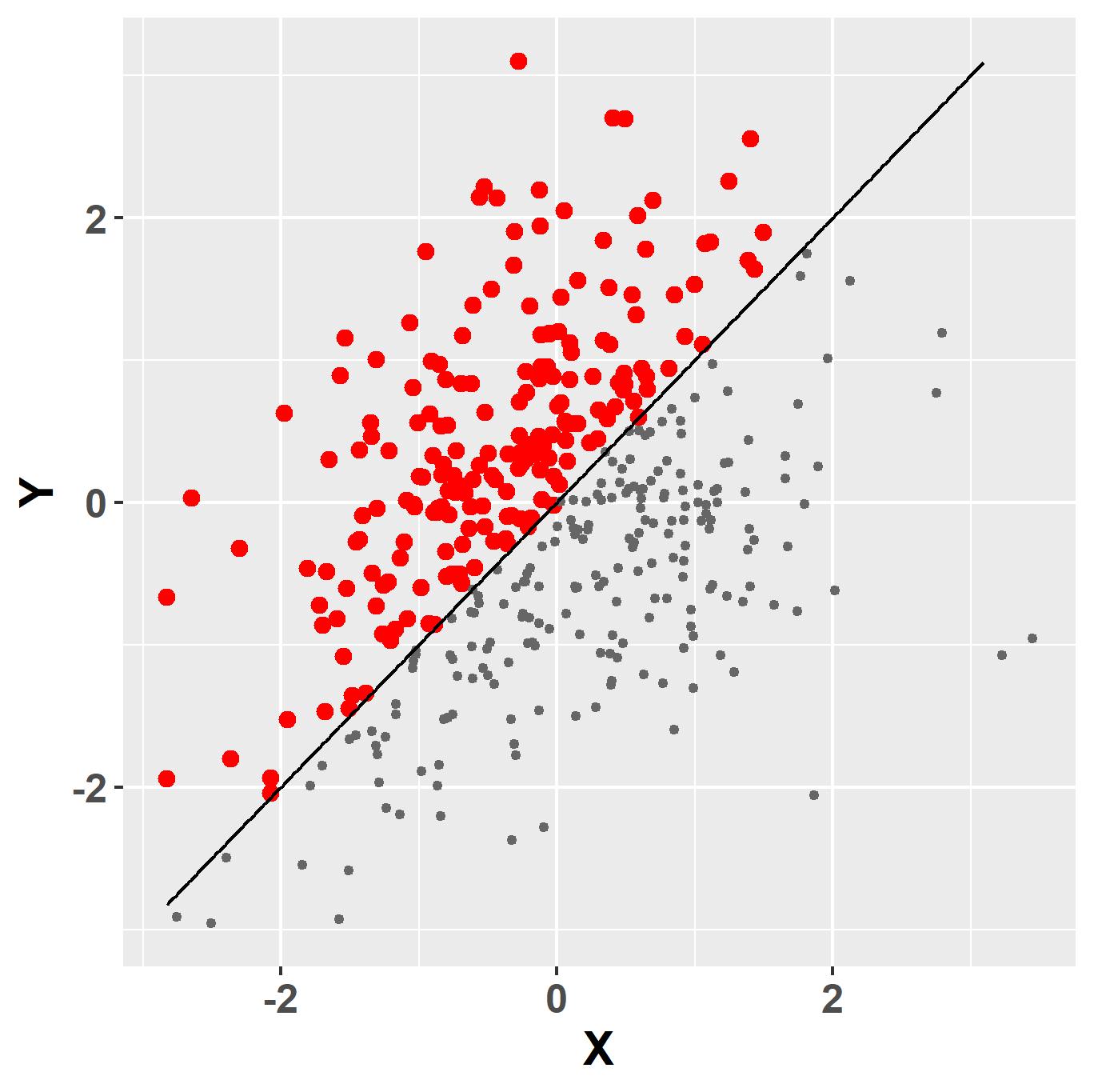} \\			
		\end{tabular}
	\end{center}
	%\end{minipage}
	\vspace{-0.2in}
	\caption{\footnotesize Scatterplots of $400$ samples generated from the monotone-exchangeable distributions described in Section \ref{sec:truncation} of the main manuscript, under truncation $\indicator{X\leq Y}$. Observed points are shown in red. Points excluded due to the truncation are shown in gray. All plots show data with standard Normal margins. The left two plots are for a bi-variate Gaussian distribution. For the negative correlation $(\rho=-0.9)$ the dependence between $X$ and $Y$ is immediately apparent even when observing only data after truncation. For the positive correlation $(\rho=0.9)$ the dependence is harder to detect by observing only such data, in similar to Figure \ref{fig:example} in the main text. The right two figures show data for the Gumbel copula (GC) with dependence parameter $\theta=1.6$ ,and the Clayton copula (CC) with dependence parameter $\theta=0.5$.}	
	\label{fig:monotone_exchangeable}
\end{figure}

Figure ~\ref{fig:monotone_exchangeable} displays the scatterplots of samples generated from the monotone-exchangeable distributions described in Section \ref{sec:truncation} of the main manuscript. $X, Y$ were generated from a standard Gaussian distribution. We examined three possible dependence structures, as specified by three different copulas:
\begin{itemize}
	\item Gaussian, with correlation parameter $\rho \in \{-0.9, -0.8,\ldots, 0.9\}$
	\item Gumbel, with dependence parameter $\theta=1.6$
	\item Clayton, with dependence parameter $\theta=0.5$
\end{itemize}

\subsubsection*{Monotone Non-Exchangeable Joint Distribution}
\begin{figure}[!h]
	%\begin{minipage}{1.\columnwidth}
	\begin{center}		
		\begin{tabular}{cc}
			\hspace{0.2in}$LD(\rho=0)$ & \hspace{0.2in}$LD(\rho=0.4)$\\
			\includegraphics[width=0.302\columnwidth]{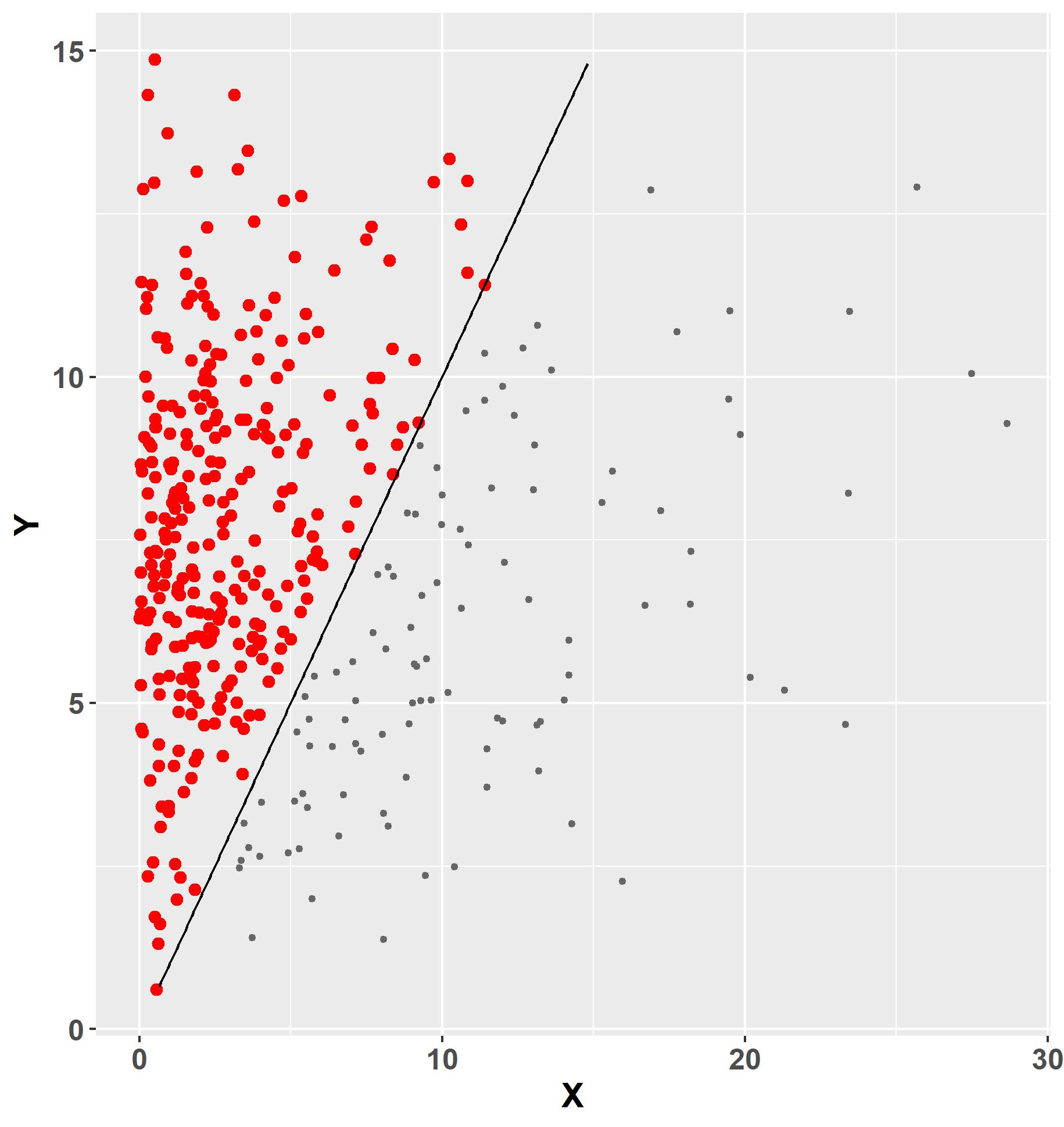} % new graphs
			&
			\hspace{-0.1in}
			\includegraphics[width=0.3\columnwidth]{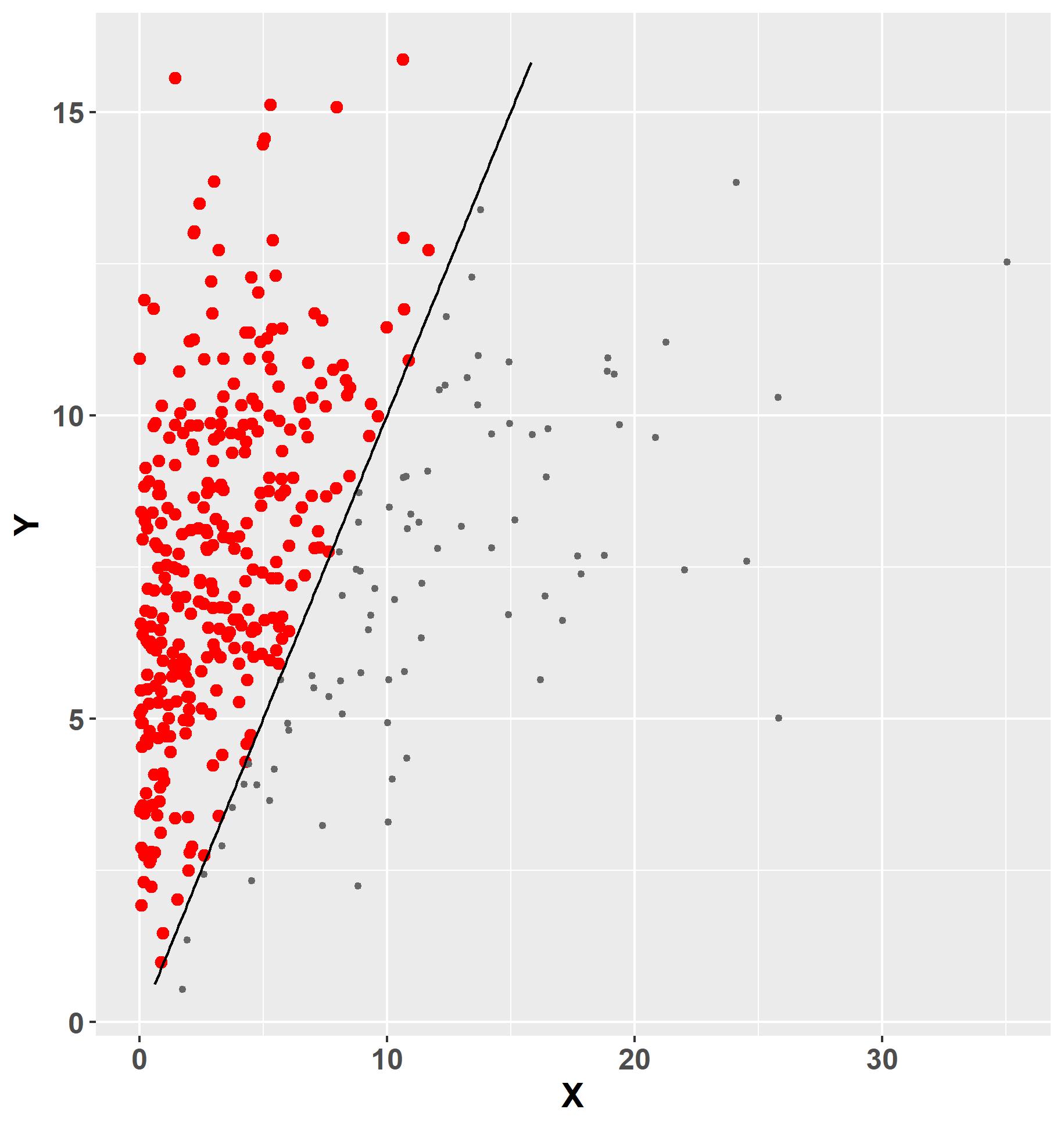} \\			
		\end{tabular}
	\end{center}
	\vspace{-0.2in}
	\caption{\footnotesize Scatterplots of samples generated from the monotone-non-exchangeable lifetime distribution (LD) described in Section \ref{sec:truncation} of the main manuscript, with otherwise the same settings as Figure \ref{fig:monotone_exchangeable}. We used a Gaussian copula with dependence parameter $\rho=0$ (left) and $\rho=0.4$ (right), where $X \sim Weibull(8.5, 3)$  and $Y \sim exp(0.2)$.}	
	%\end{minipage}
	\label{monotone_non_exchangeable}	
\end{figure}

Figure ~\ref{monotone_non_exchangeable} displays the scatterplots of samples generated from the monotone non-exchangeable distribution appear in Section \ref{sec:truncation} of the main manuscript. We generated pairs ($X,Y$) with $X \sim exp(0.2)$ and $Y \sim Weibull(3,8.5)$. We specified the dependence of $(X,Y)$ through a normal copula, where the strength dependence is determined by the correlation parameter $\rho$. We considered two levels of dependence as measured by the pre-truncation $\rho= 0$ and $\rho=0.4$.

\subsubsection*{Non-Monotone Exchangeable Joint Distribution}
\begin{figure}[!h]
%\begin{minipage}{1.\columnwidth}
	\begin{center}		
		\begin{tabular}{c}
			%\begin{tabular}{c}
				\hspace{0.15in}$CLmix(0.5)$\\
				\includegraphics[width=0.302\columnwidth]{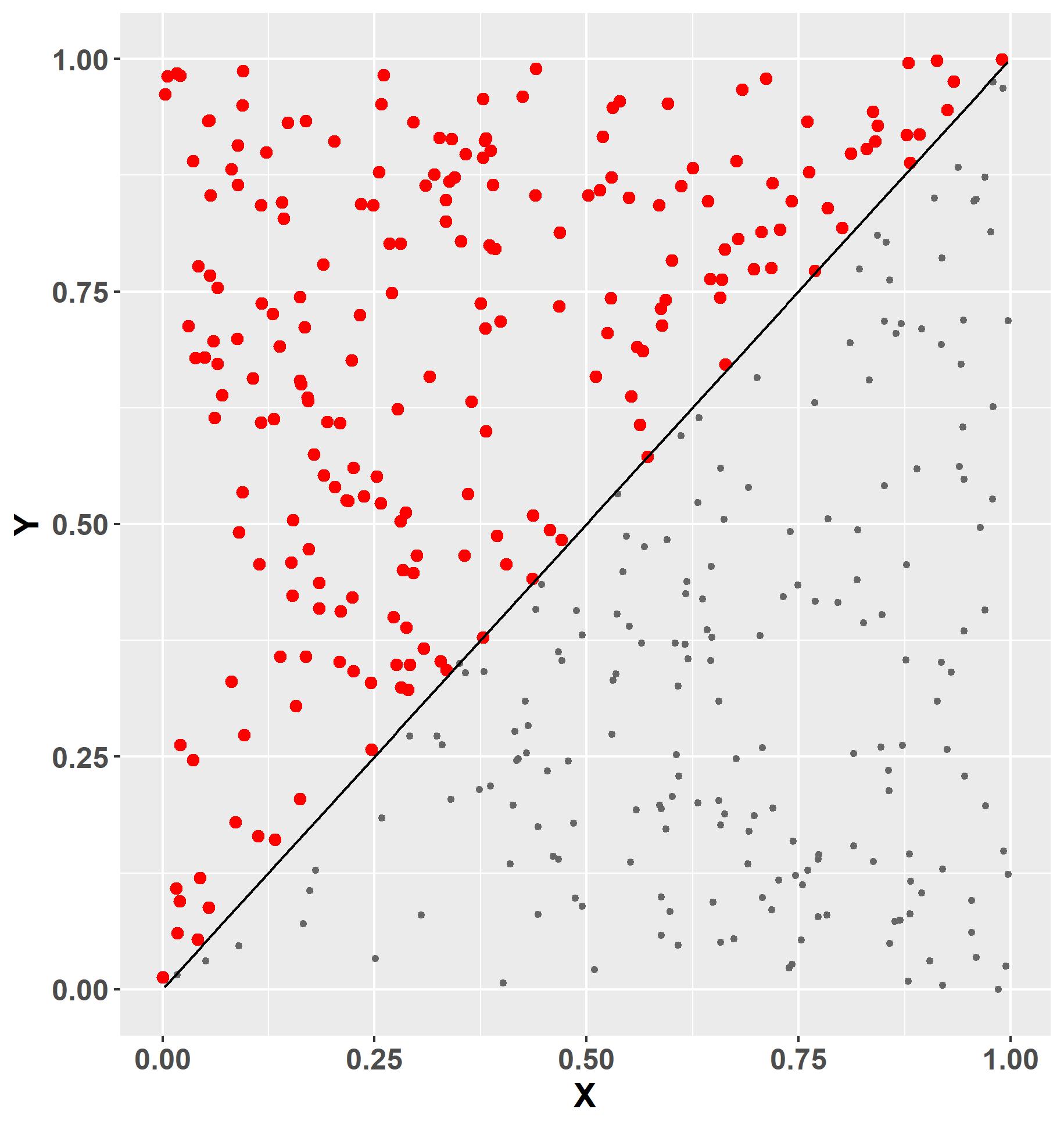} 	% new figure
			%\end{tabular}
		\end{tabular}
	\end{center}
%\end{minipage}
\vspace{-0.2in}
\caption{\footnotesize Scatterplots of samples generated from the  non-monotone exchangeable distribution described in Section \ref{sec:truncation} of the main manuscript.  $X, Y$ were generated from a mixture of two Clayton copulas with dependence parameters $\theta=0.5$ and $\theta=-0.5$.}
\label{non_monotone_exchangeable}		
\end{figure}

Figure ~\ref{non_monotone_exchangeable} displays the scatterplot of samples generated from the non-monotone exchangeable distribution described in Section \ref{sec:truncation} of the main manuscript. $X, Y$ were generated from a mixture of two Clayton copulas with dependence parameter $\theta=0.5$ and $\theta=-0.5$.

\subsubsection*{Non-Monotone Non-Exchangeable Joint Distribution}
\begin{figure}[!h]
%\begin{minipage}{1.\columnwidth}
	\begin{center}		
		\begin{tabular}{cccc}			
			%\begin{tabular}{c}
			\hspace{0.01in}$CNorm(-0.9)$ & \hspace{0.07in} $CNorm(0)$ & \hspace{0.01in} $CNorm(0.5)$ & \hspace{0.15in}$CNorm(0.9)$ \\
%				\includegraphics[width=0.25\columnwidth]{Figures/scatterplots/new_scatterlots/nonmonotone_nonexchangeable_minus_09.jpeg}
%			&
%			\hspace{-0.1in}
%				\includegraphics[width=0.25\columnwidth]{Figures/scatterplots/new_scatterlots/nonmonotone_nonexchangeable_0.jpeg}  %	$CNorm(0.5)$ & \hspace{0.1in}$CNorm(0.9)$ \\
%				\includegraphics[width=0.25\columnwidth]{Figures/scatterplots/new_scatterlots/nonmonotone_nonexchangeable_05.jpeg} & \hspace{-0.1in}
%				\includegraphics[width=0.25\columnwidth]{Figures/scatterplots/new_scatterlots/nonmonotone_nonexchangeable_09.jpeg}			\\  % below are the new figures
				\hspace{-0.3in}
				\includegraphics[width=0.26\columnwidth]{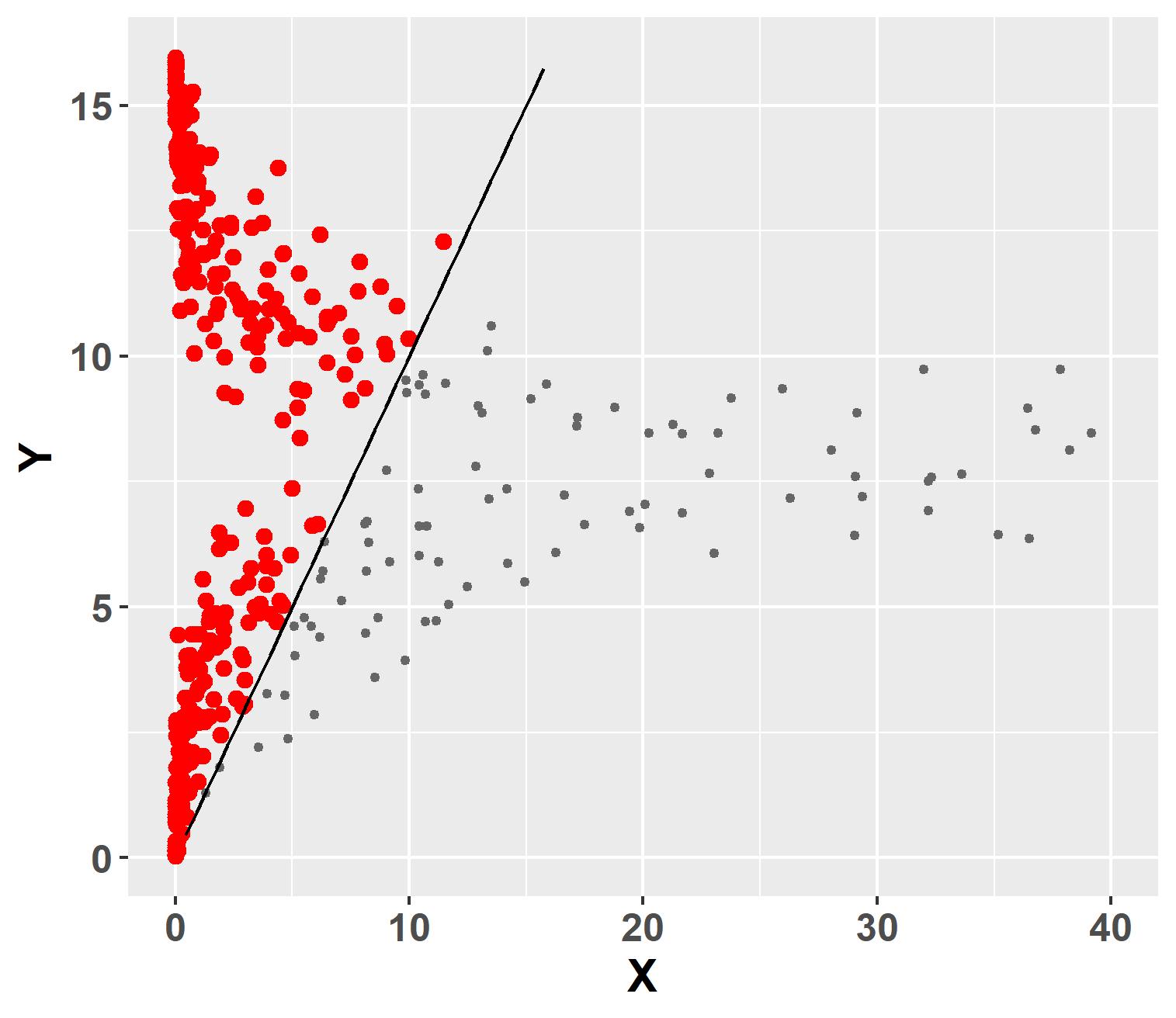}
				& \hspace{-0.2in}
				\includegraphics[width=0.26\columnwidth]{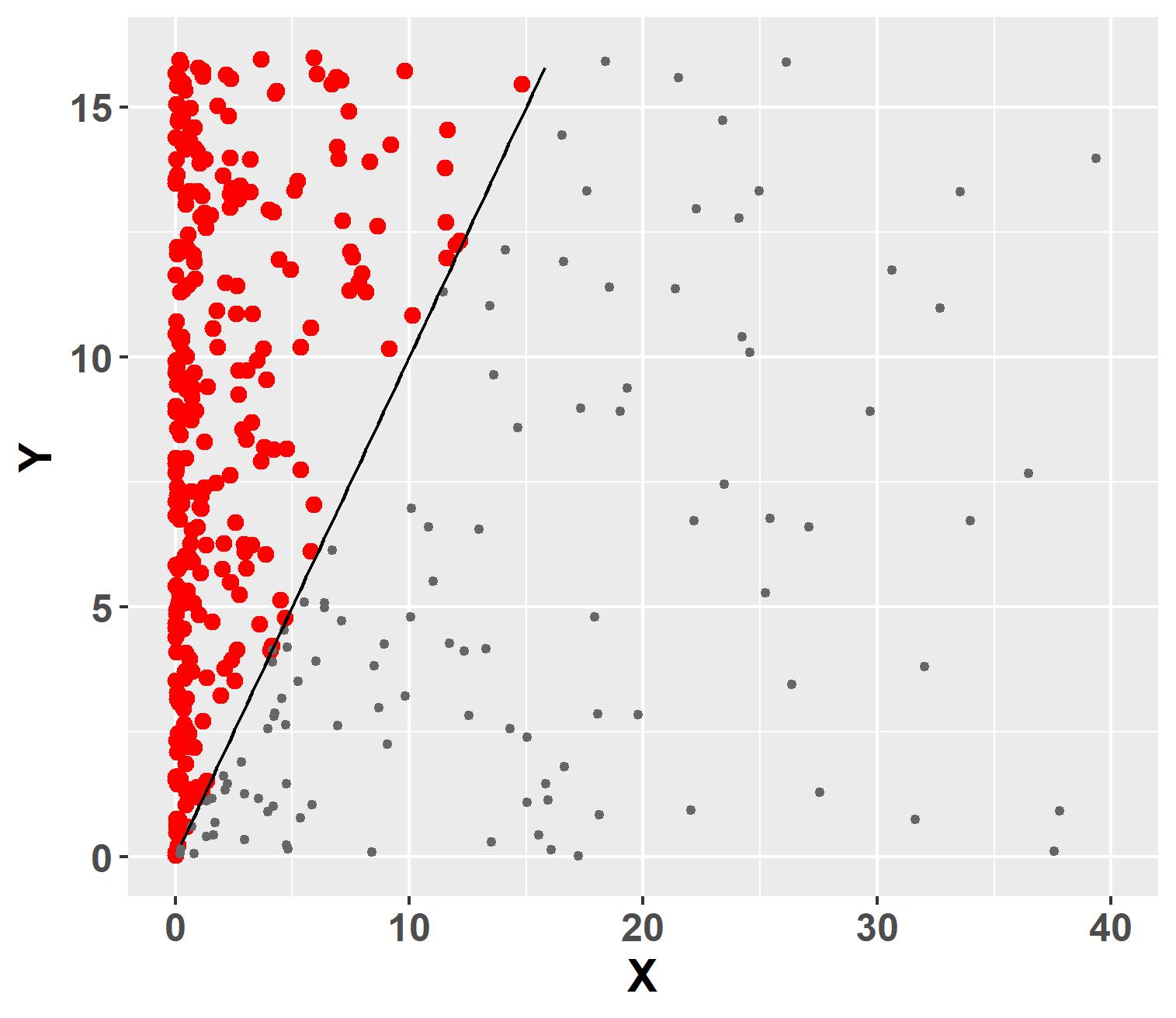}				
				& \hspace{-0.2in}
				\includegraphics[width=0.26\columnwidth]{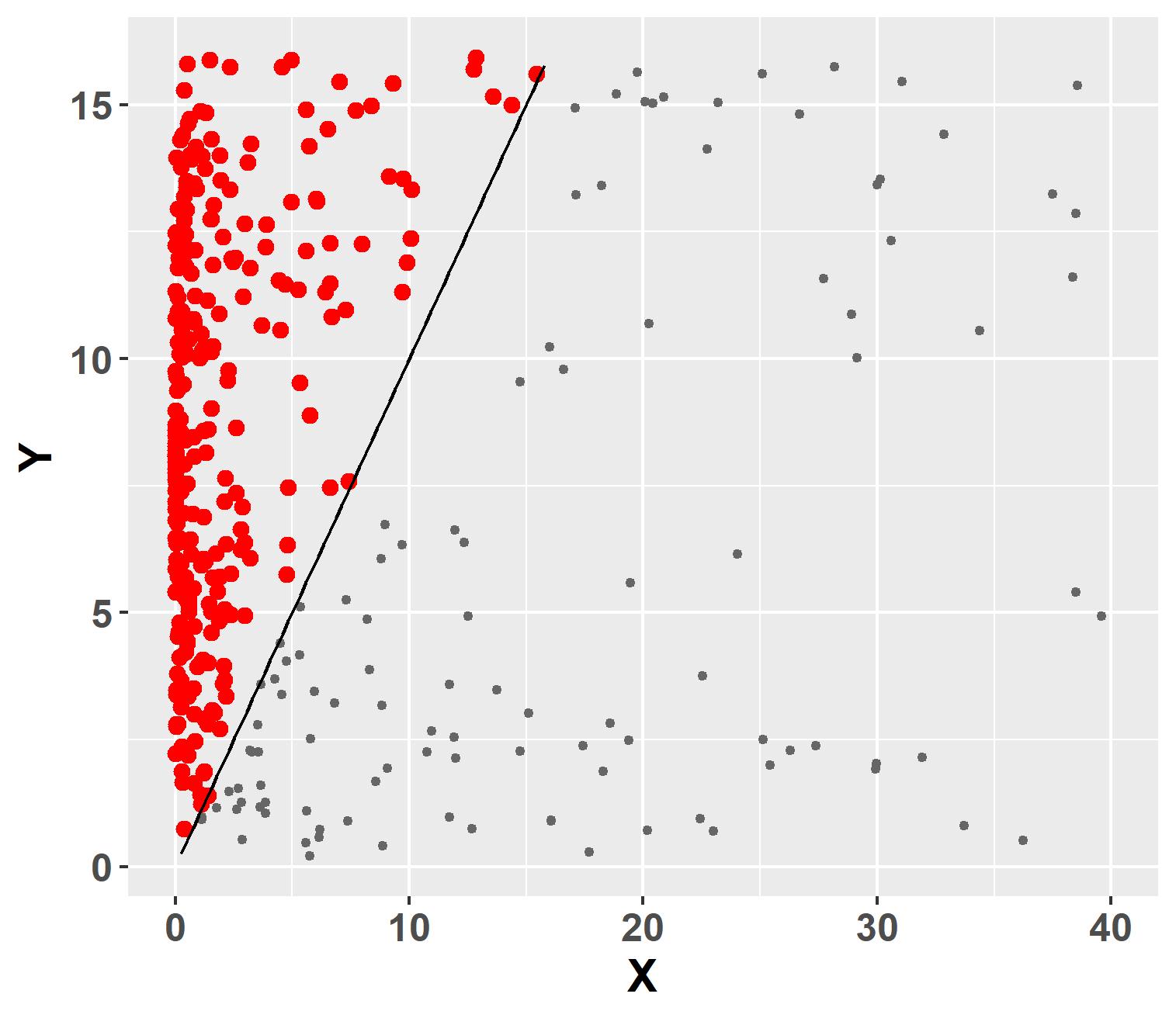}				
				& \hspace{-0.2in}
				\includegraphics[width=0.26\columnwidth]{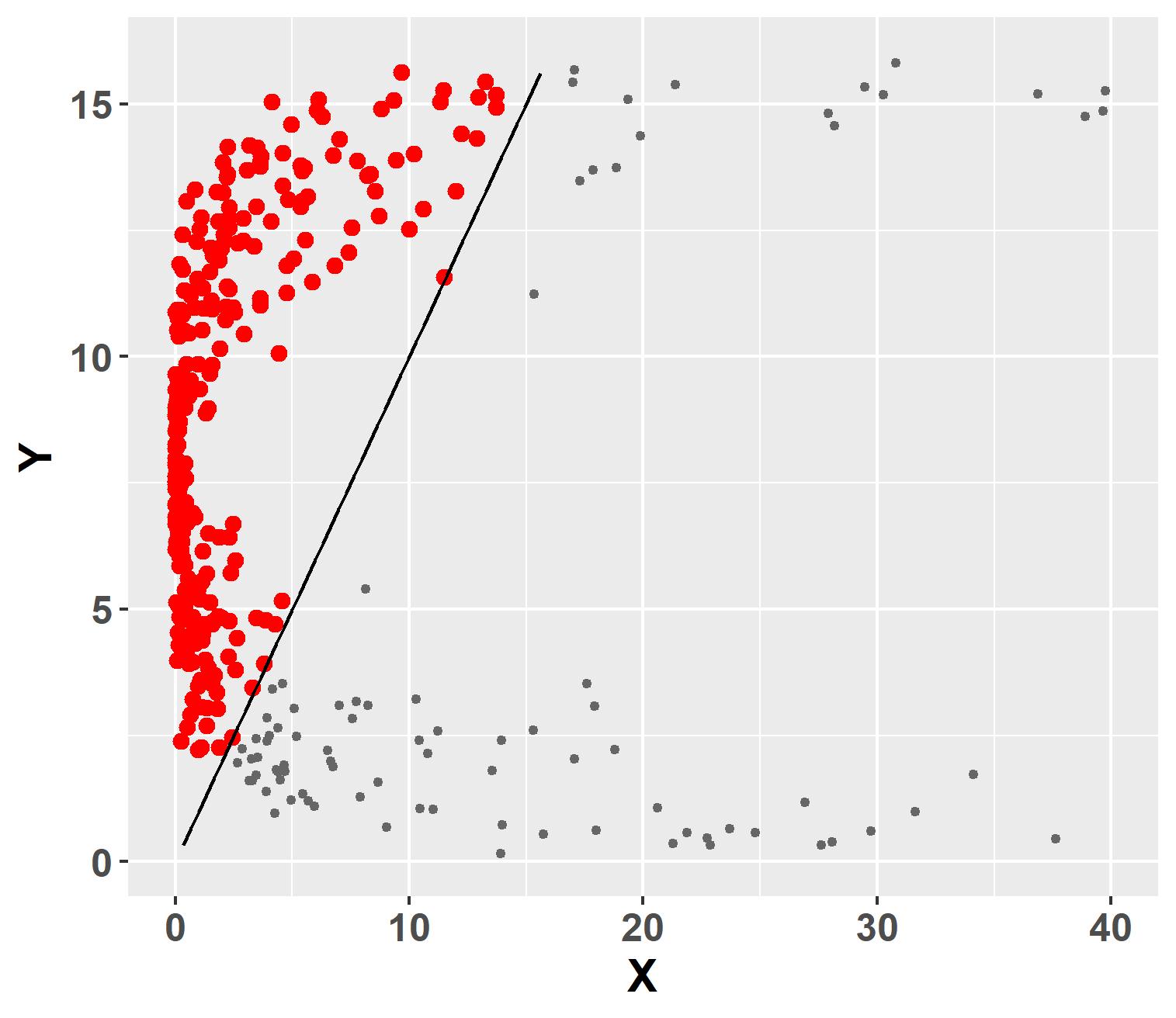}				
		\end{tabular}
	\end{center}
\vspace{-0.4in}
\caption{\footnotesize Scatterplots of samples generated from the non-monotone non-exchangeable distributions described in Section \ref{sec:truncation}. We used a normal copula with varied correlation coefficient $\rho$ to specify the joint distribution of $X, Y$, with $X \sim Weibull(0.5,4)$ and $Y \sim U[0,16]$. Due to the long tail of the Weibull distribution, we cut the $X$-axis at $40$, and points with $X_i>40$ are not shown.}	
\label{fig:non_monotone_non_exchangeanle}	
\end{figure}

Figure ~\ref{fig:non_monotone_non_exchangeanle} displays the scatter plot of samples generated from the non-monotone non-exchangeable distributions described in Section \ref{sec:truncation}. We used a normal copula with varied correlation coefficient $\rho$ to specify the joint distribution of $X, Y$, with $X \sim Weibull(0.5,4)$ and $Y \sim U[0,1]$.

\subsubsection{Strictly Positive Biased Sampling}
\begin{figure}[!h]
	%\begin{minipage}{1.\columnwidth}
	\begin{center}		
		\begin{tabular}{ccc}
			\hspace{0.2in}$LogNormal(\rho=0)$ & \hspace{0.2in}$LogNormal(\rho=0.5)$ & \hspace{0.2in}$LogNormal(\rho=0.9)$ \\
			\includegraphics[width=0.3\columnwidth]{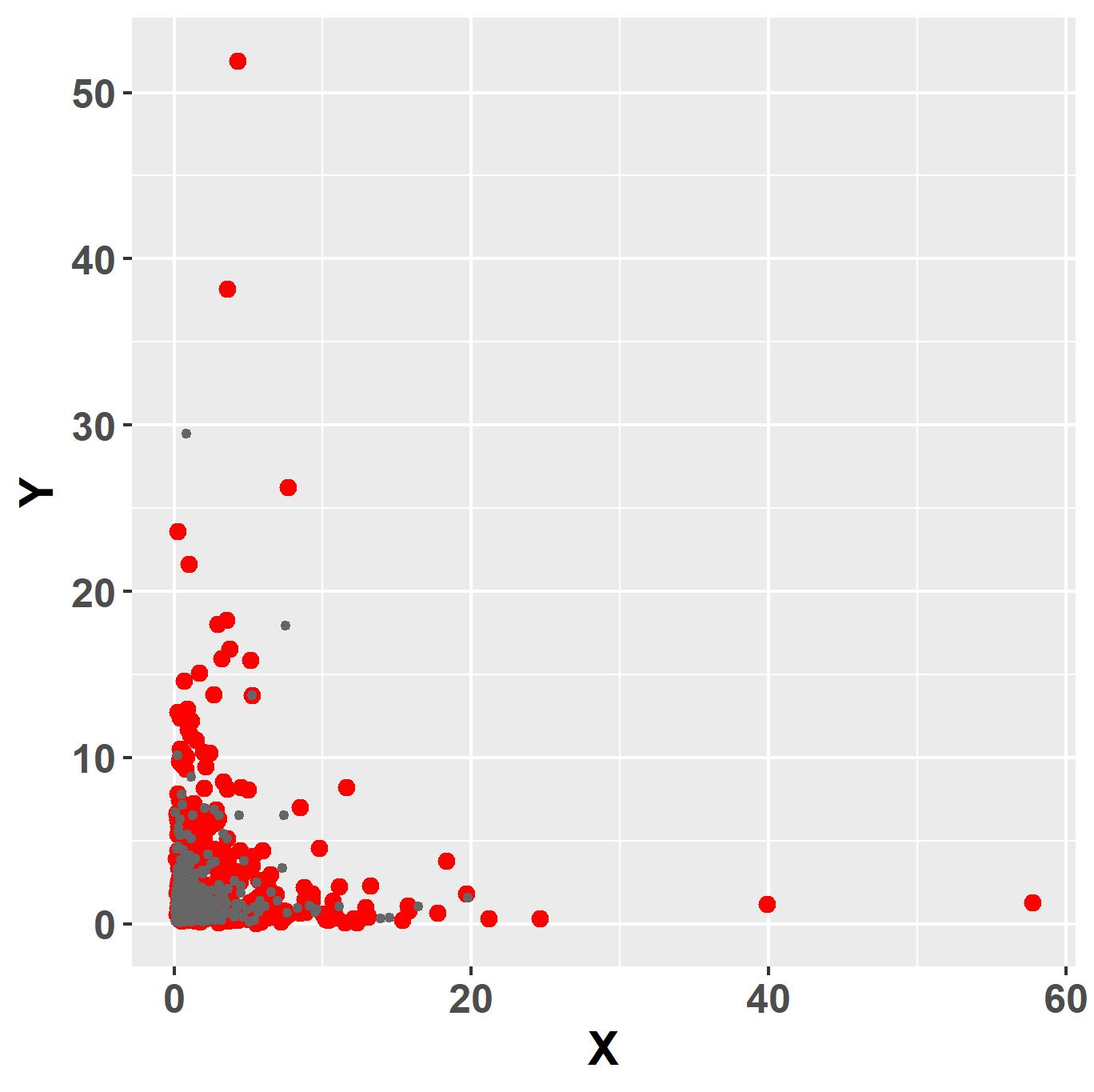}
			&
			\hspace{-0.1in}
			\includegraphics[width=0.3\columnwidth]{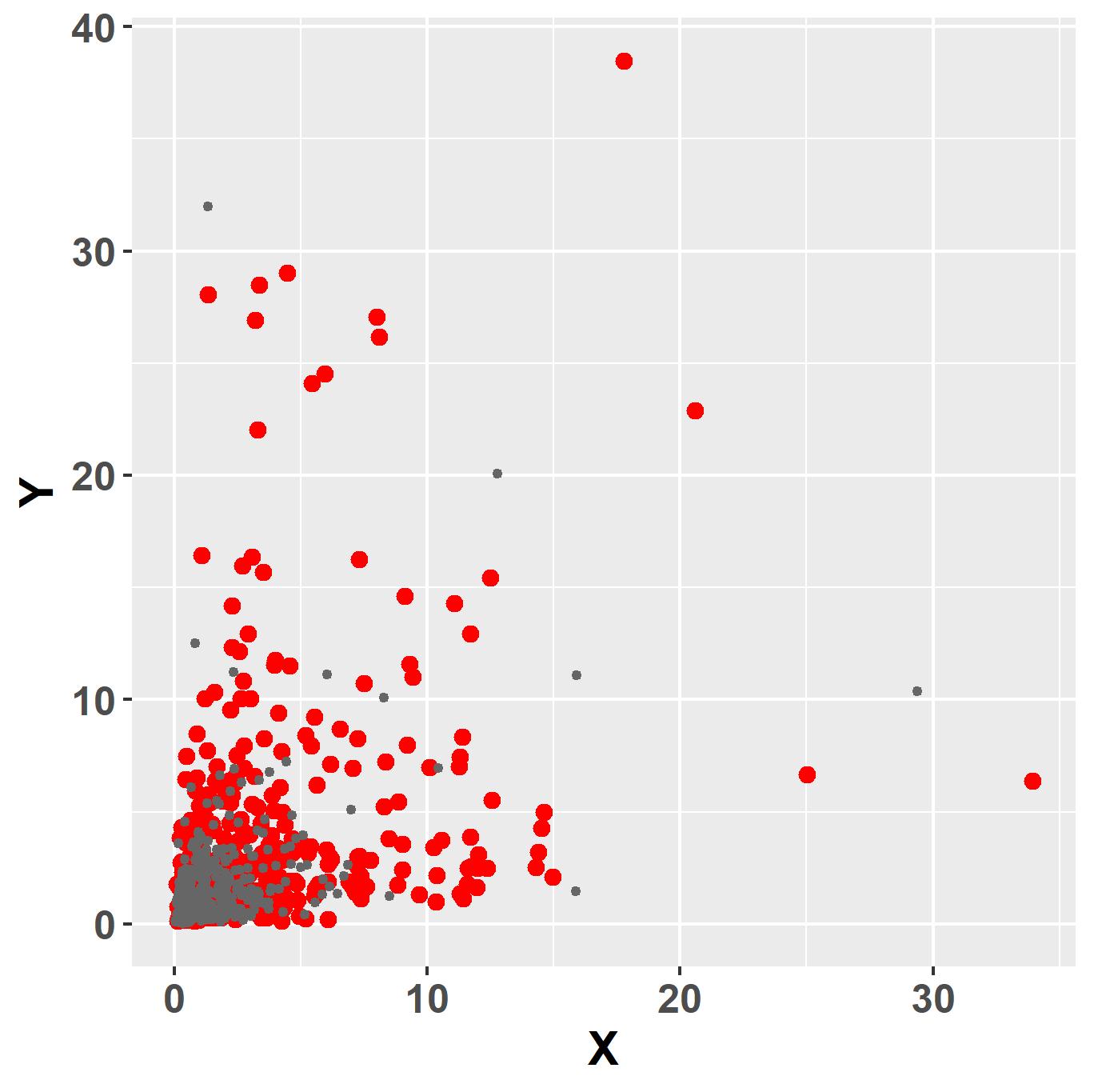}
			&
			\hspace{-0.1in}
			\includegraphics[width=0.3\columnwidth]{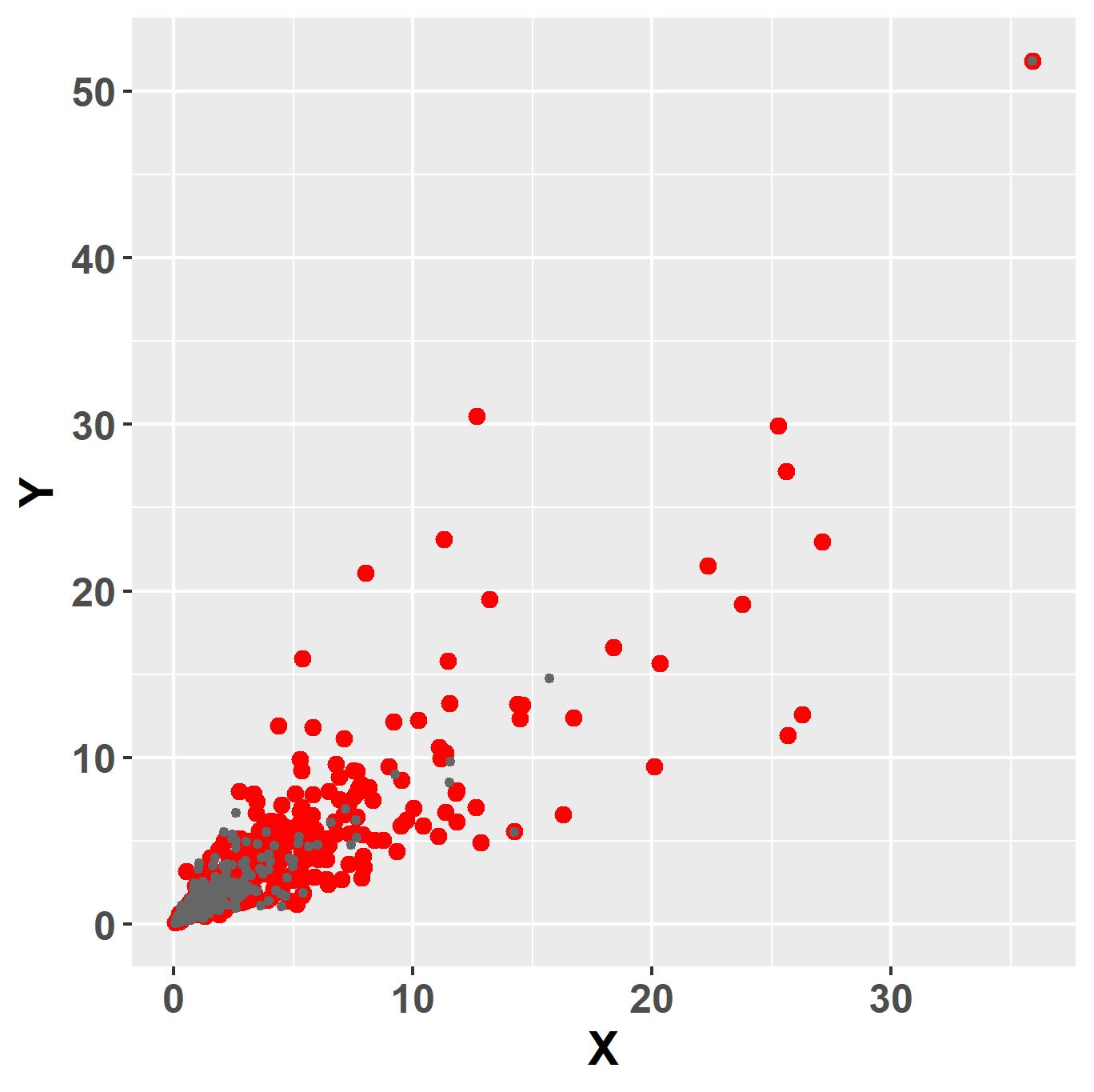} \\
			\hspace{0.2in}$Normal(\rho=0)$ & \hspace{0.2in}$Normal(\rho=0.5)$ & \hspace{0.2in}$Normal(\rho=0.9)$ \\			
			\includegraphics[width=0.3\columnwidth]{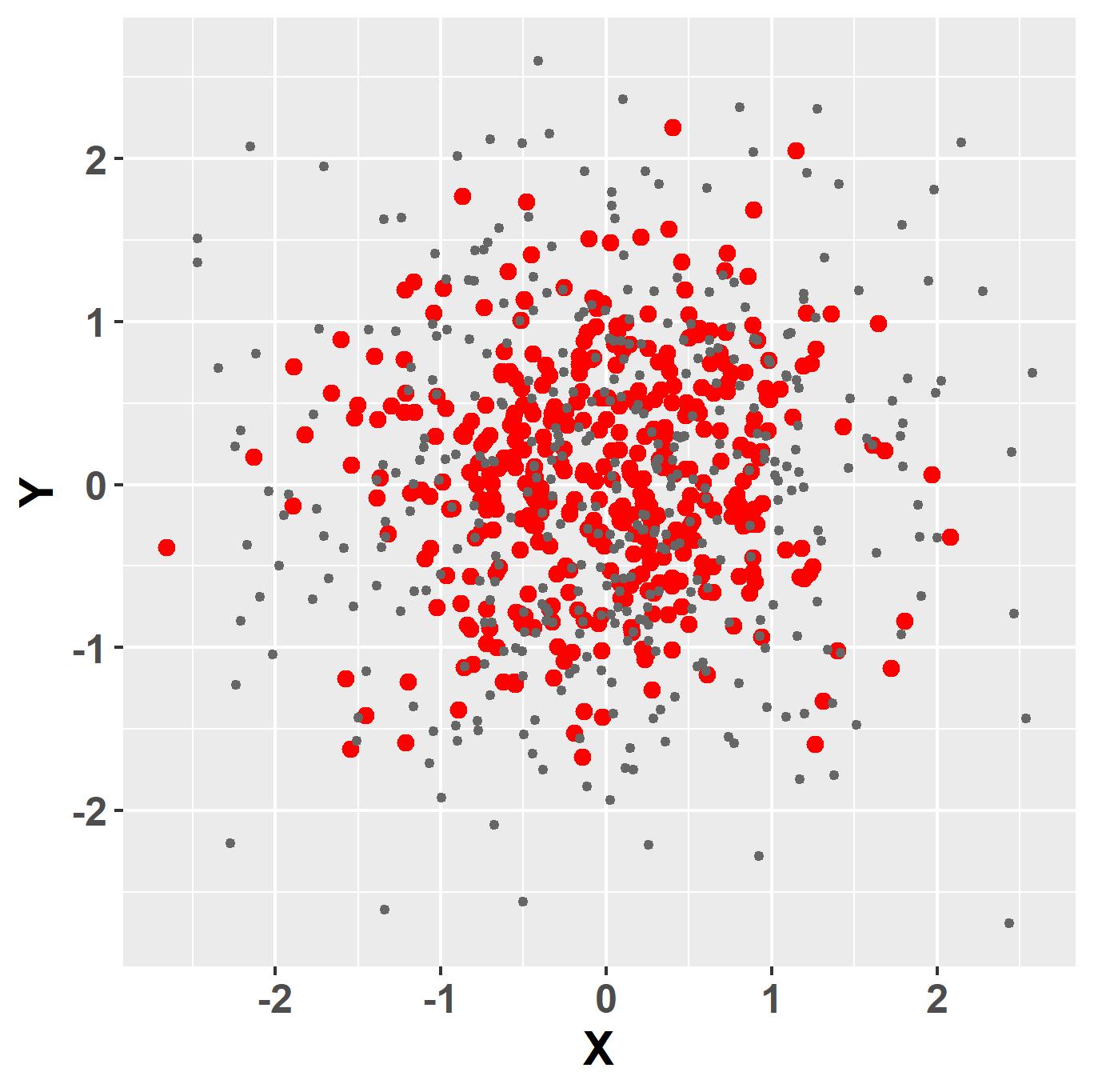}
			&
			\hspace{-0.1in}
			\includegraphics[width=0.3\columnwidth]{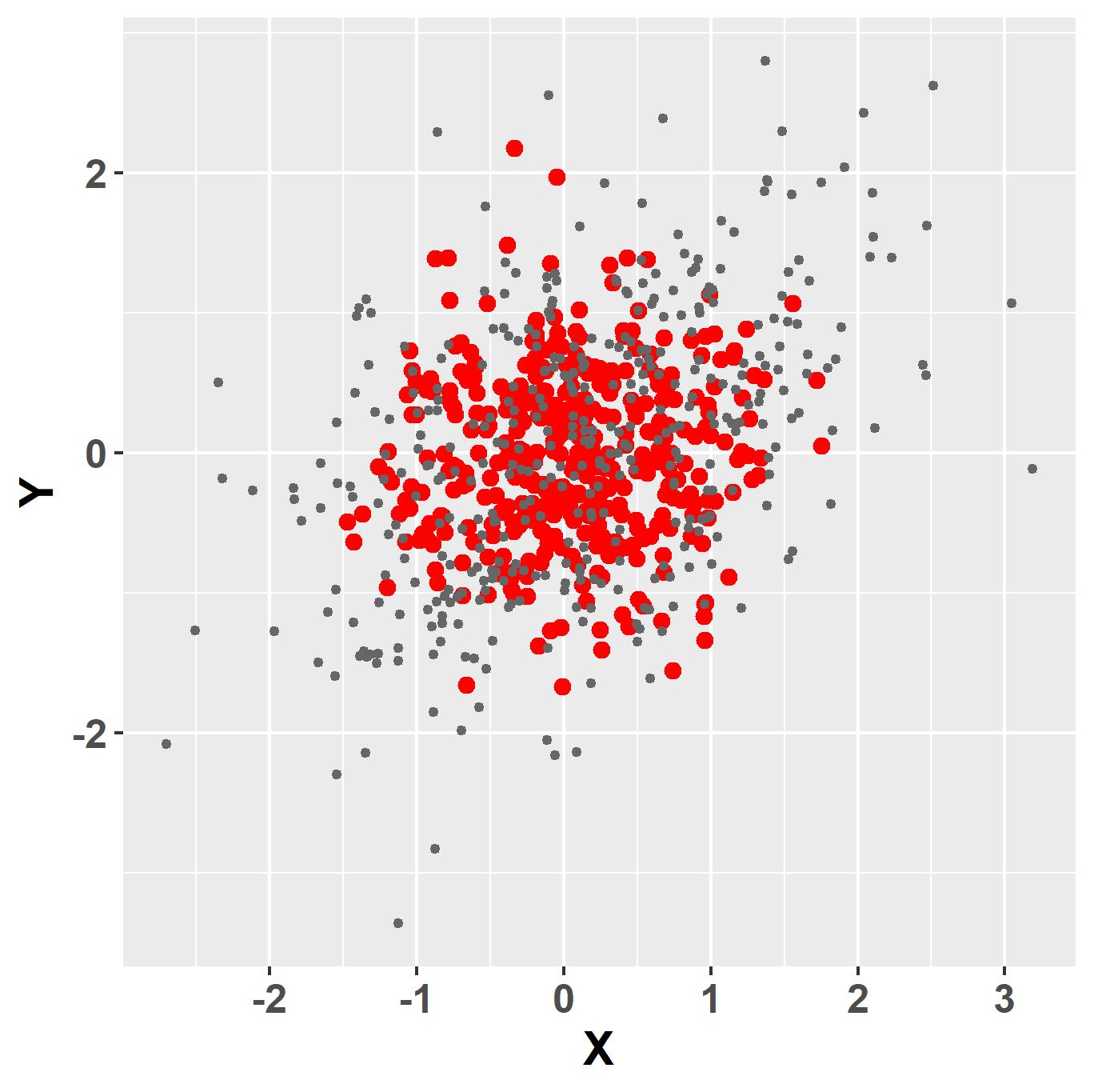}
			&
			\hspace{-0.1in}
			\includegraphics[width=0.3\columnwidth]{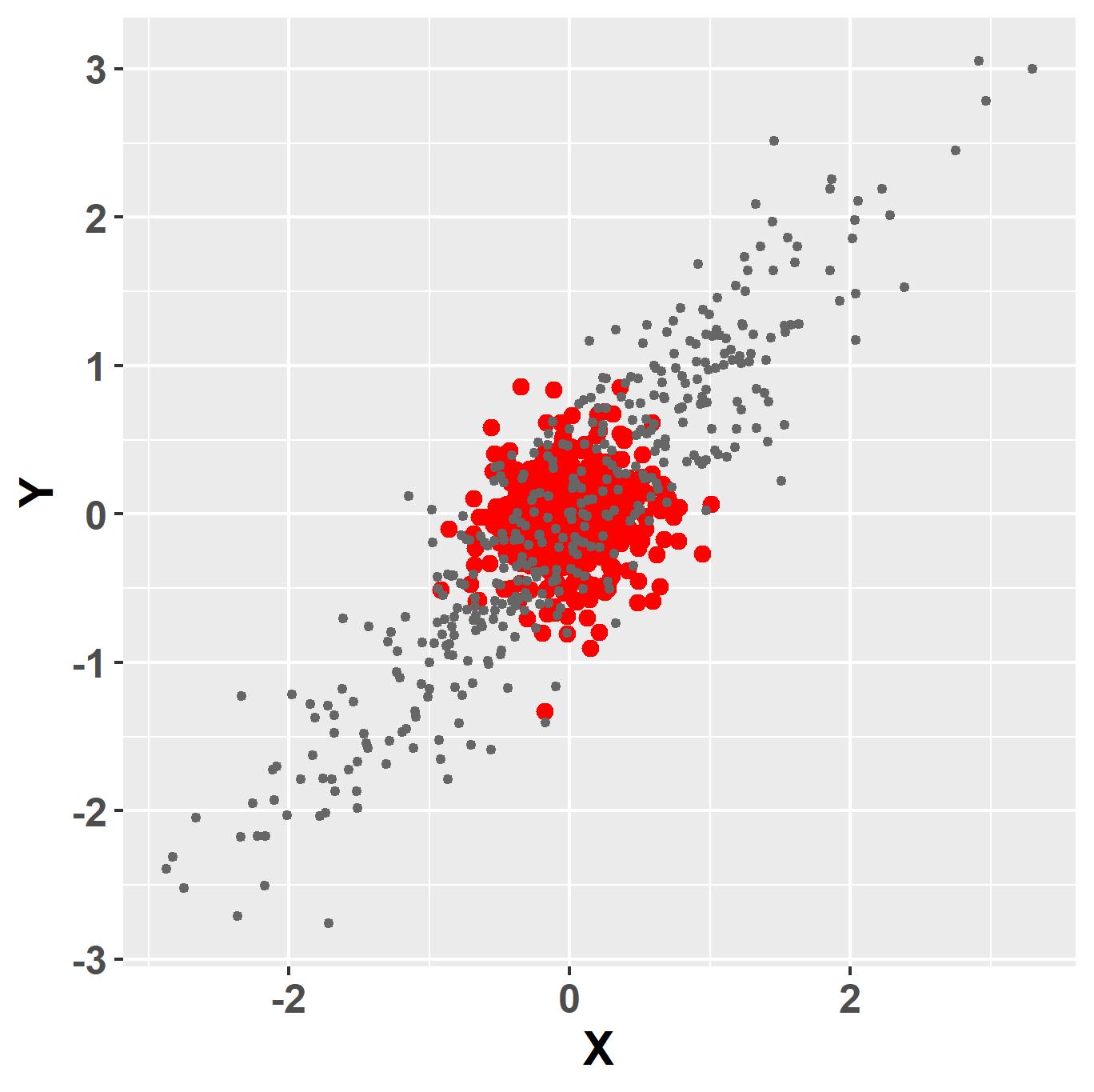}
			
		\end{tabular}
	\end{center}
	\vspace{-0.2in}
	\caption{\footnotesize Scatterplots of samples generated from distributions with positive biased sampling. Points sampled from $\jointpdf^{(\wfun)}$ are shown in red. Points sampled from $\jointpdf$ are shown in gray.
		The top panels show data for the LogNormal distribution described in Section \ref{sec:strictly_positive} of the main manuscript, sampled under the bias function $\wfun(x,y)=x+y$ for different values of $\rho$. The sampling mechanism increases the values of the observed $(x_i,y_i)$ and creates spurious dependence even under the null ($\rho=0$).  The bottom panels show data for the normal distribution with correlation $\rho$, sampled under a bias function proportional to the normal density with correlation $-\rho$. The sampling mechanism cancels the dependence between $X$ and $Y$ and the observed data is independent for all values of $\rho$.}	
	%\end{minipage}
	\label{fig:log_normal}	
\end{figure}
Figure ~\ref{fig:log_normal} displays the scatter plot of samples generated from the two distributions with positive biased sampling functions, described in Section \ref{sec:strictly_positive} in the main text: (1.) A LogNorml distribution with standard marginals and correlation $\rho$ with $\wfun(x,y)=x+y$, and (2.) a bi-variate Gaussian distribution with standard marginals and correlation $\rho$ with $\wfun$ proportional to a bi-variate Gaussian density with standard marginals and correlation $-\rho$.

\subsection{Computing the Expectations for the Statistic Under the Null}
\label{sec:fast_bootstrap}
%\micha{This section can be removed. It is something we explored, but it is not so interesting as it ended up not so useful} \zuk{I think it's good to show also directions that failed, to show that we considered then, and why do we need to estimate correctly the expectations which is a hard task. I referred to this appendix section from Section 5.1}
Our two proposed tests are significantly faster than minP2, as shown in Section \ref{sec:simulations}. However, Tsai's test, which is specialized for truncation, is faster than our tests, especially compared to the bootstrap test which is about an order of magnitude slower than the weighted permutation test. It is thus natural to seek computational improvements for our tests, but this may have statistical, in addition to computational, consequences.
In particular, the computation of our test statistic $T$ in Equation \eqref{eq:modifed_hoeffding} requires estimation of the expected cell counts under the null $e_{i}^{jk}$. It is possible to use a different test statistic with different,
albeit wrong, expectations. The weighted permutation test with this statistic will still be valid, according to Corollary \ref{cor:alpha}, and a bootstrap test with this statistic can also be used. We examined two such modified statistics:
(i) A naive approach, where we ignore the biased sampling function $\wfun$ and
simply use the empirical marginal distributions $\xcdfhat, \ycdfhat$ as our estimators, resulting the $T$ statistic used by \cite{heller2012consistent}, and (ii) a fast bootstrap based approach, where
the marginals are estimated according to Equation \eqref{eq:expected_quardant}, but are not re-estimated
for each bootstrap sample. While these two approaches appear simpler and are faster,
they may suffer from significant loss of power. For example, for bivariate Gaussians with correlation coefficient $\rho=-0.4$,
our bootstrap approach achieves a power of $0.634$, while using the naive expectations yield only
power of $0.394$ and estimating the expectations only once via bootstrap achieves power of $0.406$. Similar trends of reduction in power were observed for other $\rho$ values.

\subsection{Algorithms for Importance Sampling}
\label{sec:IS_appendix}

We have investigated several algorithms for importance sampling, and in particular their accuracy and validity for testing.
In \cite{chen2007sequential}, a sequential importance sampling (SIS) algorithm for truncation was proposed.
\cite{kou2009approximating} generalized the algorithm for approximating the permanent and $\alpha$-permanent of general matrices. The validity of the importance sampling technique was established in \cite{harrison2012conservative}. % for Pvalue calculation
%for the unweighted case, provided that we include the importance weight also for the identity permutation in the calculation of the Pvalue.
%It can easily be generalized to our case of weighted permutations, as shown in Theorem \ref{thm:validity_IS}. \zuk{FILL DETAILS}.

We describe in Algorithm \ref{alg:sis_permutations} the details of the monotonic SIS scheme we have implemented. The idea behind this scheme is to order the rows $i$ of $\wmat$ in a decreasing order of the variance of the log-weights $\log(\wmat(i,j))$. Then, starting with the row with highest variance, for each such row $i$ select $\pi(i)$ maximizing $\log(\wmat(i, \pi(i)))$. This procedure will increase the value of $\log(\wmat(i, \pi(i)))$ for the first rows significantly beyond their expectation, where as we reach the last rows and must settle for lower values of $\log(\wmat(i, \pi(i)))$, the variance between different choices is small and the overall contribution to the product $P_{\wmat}(\pi)$ is not significant. As a result, this scheme samples permutations $\pi$ with high values of $P_{\wmat}(\pi)$. Additional algorithms can be suggested by small modifications of Algorithm \ref{alg:sis_permutations}. In 'uniform' importance sampling, we simply sample permutations based on the uniform measure $P_{IS}(\pi)=\frac{1}{n!}$. The 'uniform grid' method interpolates between the uniform method, that produce permutations with $P_{\wmat}(\pi)$ too low, and the monotonic method that produce permutations with $P_{\wmat}(\pi)$ that may be too high. This is achieved by taking a grid of values $0 = \alpha_1 < \alpha_0 < .. < \alpha_G = 1$ and for each $\alpha_k$ modifying step $9$ of the algorithm to sample with probability proportional to $\wmat(\sigma(i),j)^{\alpha_k}$. In our simulations we used $G=10$ with equidistant $\alpha_i$ values. % \yaniv{how did we set $\alpha$ in our simulations?}.
Finally, the Kou-McCullagh algorithm skips the ordering steps 5,6 and computes the column-sums $C_i = \sum_j \wmat(i,j)$ over the remaining rows at each step. Step  $9$ is replaced by sampling $\pi(i,j)$ with probability proportional to $\frac{\wmat(i,j)}{C_j-\wmat(i,j)}$ (see \cite{kou2009approximating} for more details).

We have studied the empirical performance of the SIS algorithms. To illustrate, we show their empirical rejection rate under the null for the LogNormal example with $\wfun(x,y)=x+y$. While for $n=100$ the different sampling schemes perform reasonably well, as shown in Table \ref{tab:strictly_positive}, as the sample size $n$ grows, approximating the $P-value$ becomes more challenging. We illustrate the limitations of the algorithms in an example with $n=1000$. Empirically, the MCMC approach is robust and maintains an approximately
uniform $U[0,1]$ P-values distribution under the null in all cases we have tested, while the importance sampling approaches are all overly conservative.

Figure \ref{fig:IS_distribution} shows the probabilities of the sampled permutations under different methods, and the corresponding test statistic. Figure \ref{fig:IS_log_prob_ratio} shows the variability in the importance sampling weights
 $\frac{P_{\wmat}(\pi_i)}{P_{IS}(\pi_i)}$ of the different methods, causing a poor estimation of the P-value.

Figure \ref{fig:IS_pvals} shows the resulting performance. The Kou-McCullagh scheme is the least conservative scheme among all SIS methods, but is still far from a uniform $U[0,1]$ P-values distribution. The power of the SIS under the alternative will also be reduced due to their conservative behaviour. %showing the closest performance to the MCM approach. %The 'match.w' version shows the most similar performance to the $P_{\wmat}$ desired distribution.

\begin{algorithm}
	\caption{{\small SIS Permutation Test of Quasi-Independence }} {\small
	\begin{algorithmic}[1]
		\Statex {\bf Input:} $\sample$ - sample, $\wfun(x,y)$ - bias function.
		\Statex {\bf Parameters:} $\nperm$ - number of permutations.
		\State Compute $\wmat(i,j) = \wfun(x_i,y_j), \quad \forall i,j=1,..,n$.
		\State Set $\pi_0$ the identity permutation $\pi_0(i) = i$.
		\For{$b=1$ to $\nperm$}
		\State Initialize the importance sampling and weighted (unnormalized) probabilities \hspace{1cm} $\quad \quad P_{IS}(\pi_b)=P_{\wmat}(\pi_b)=1$.				
		\State Compute $V_i = Var\Big(\log(\wmat(i,1)),..,\log(\wmat(i,n))\Big), \quad \forall i=1,..,n$.
		\State Order the variables by decreasing order of $V_i$:  $V_{\sigma(1)} \geq ... \geq V_{\sigma(n)}$. 		
		\For {$i=1$ to $n$}
		\State Sample $\pi_{b}(\sigma(i))$ with probabilities: \\ \hspace{3cm} $Pr(\pi_{b}(\sigma(i))=j) \propto \wmat(\sigma(i),j) \prod_{k=1}^{i-1} \indicator{\pi_b(\sigma(k)) \neq j}$.
		\State Update the probability $P_{IS}(\pi_{b}) = P_{IS}(\pi_{b}) \times \frac{\wmat(\sigma(i),\pi_{b}(i))}{\sum_{j=1}^n  \wmat(\sigma(i),j) \prod_{k=1}^{i-1} \indicator{\pi_b(\sigma(k)) \neq j}}$ .
		\State Update the (unnormalized) probability $P_{\wmat}(\pi_{b}) = P_{\wmat}(\pi_{b}) \times  \wmat(\sigma(i),\pi_{b}(i))$ .		
		\EndFor		
		\EndFor		
		\State {\bf Output:} $P_{value}$ computed from the $P_{IS}(\pi_b),P_{\wmat}(\pi_b)$'s according to Equation \eqref{eq:IS_pval}.
	\end{algorithmic}
	\label{alg:sis_permutations}}
\end{algorithm}

\newpage
\clearpage
\thispagestyle{empty}

\begin{figure}[!h]
	%\begin{minipage}{1.\columnwidth}
	\begin{center}		
		\includegraphics[width=0.8\columnwidth, height=13cm]{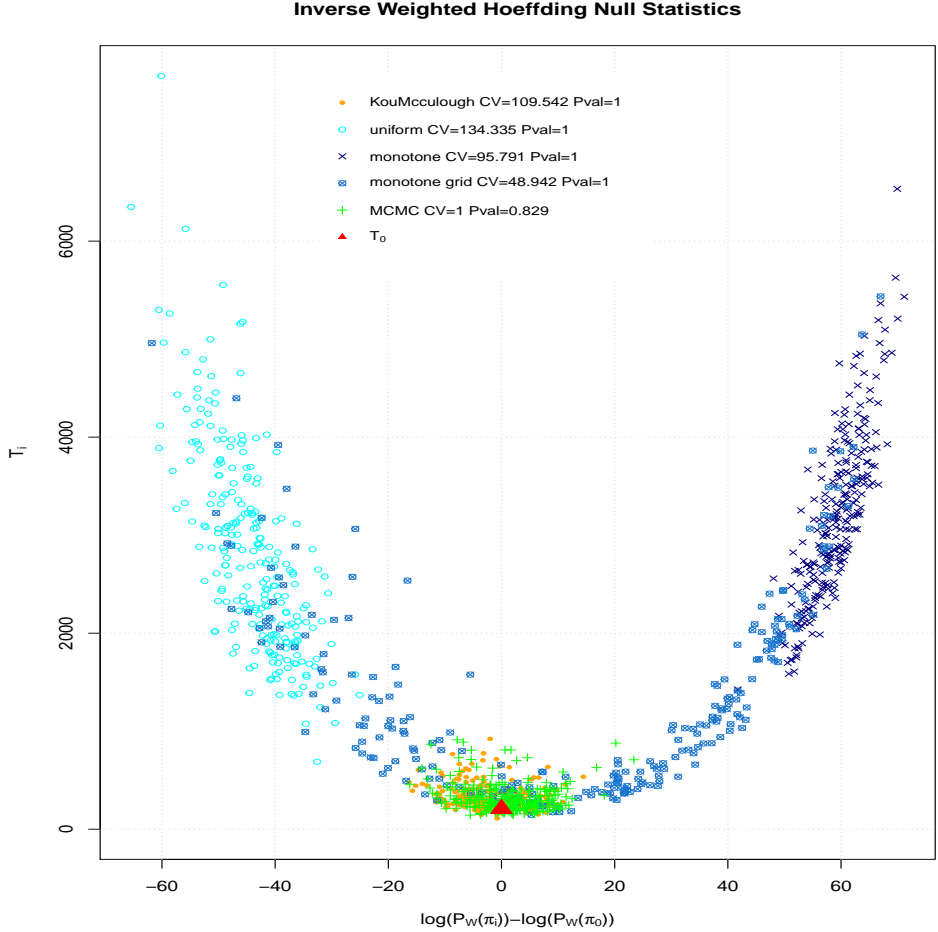}		
	\end{center}
	%\end{minipage}
	\vspace{-0.2in}
	\caption{\footnotesize Log probabilities under $P_{\wmat}$ for permutations sampled under different methods, vs. their computed statistic $T_i=T(\pi_i(\sample))$ for the inverse weighting statistic from Section \ref{sec:inverse_weight_stat}.
		For $n=1000$, a weighted sample was simulated from the independent LogNormal distribution, with $\wmat(x,y)=x+y$. Since $P_{\wmat}$ is known only up to a multiplicative constant, we show the log probabilities minus the log probability of the identity permutation.
		Shown are the scaled log probabilities for $B=250$ random permutations for four importance sampling methods, and the MCMC method (green '+' signs), in addition to the ID permutation (red triangle). For each method, we also display the coefficient of variation diagnostic
		of $\frac{P_{\wmat}(\pi)}{P_{IS}(\pi)}$, and the resulting p-value. 		
		According to the MCMC Method, the permutations with log-probability close to the identity cover most of the probability space under $P_{\wmat}$.
		The uniform (monotone) importance sampling methods sample permutations with too low (high) $P_{\wmat}$ values.
		The monotone grid method interpolates between the two and covers a large range of the log-probabilities, including log-probabilities sampled under the MCMC method. The Kou-McCullagh method \cite{kou2009approximating} also samples permutations with log-probabilities close to the MCMC method. However, the importance sampling probabilities $P_{IS}(\pi_i)$ for the permutations sampled under these methods are vastly different from the true underlying probabilities  $P_{\wmat}(\pi_i)$, as is evident by the large value of the coefficient of variation $CV(\frac{P_{\wmat}(\pi_i)}{P_{IS}(\pi_i)})$.
		Although the true P-value is $\approx 0.24$ (as confirmed by extensive MCMC simulations with higher $B$), all importance sampling methods are too conservative here and give a P-value of $1$, in part due to the inclusion of the identity permutation in the P-value calculation.}	
	\label{fig:IS_distribution}
\end{figure}

\newpage
\clearpage

\begin{figure}[!h]
	%\begin{minipage}{1.\columnwidth}
	\begin{center}		
		\includegraphics[width=0.8\columnwidth, height=13cm]{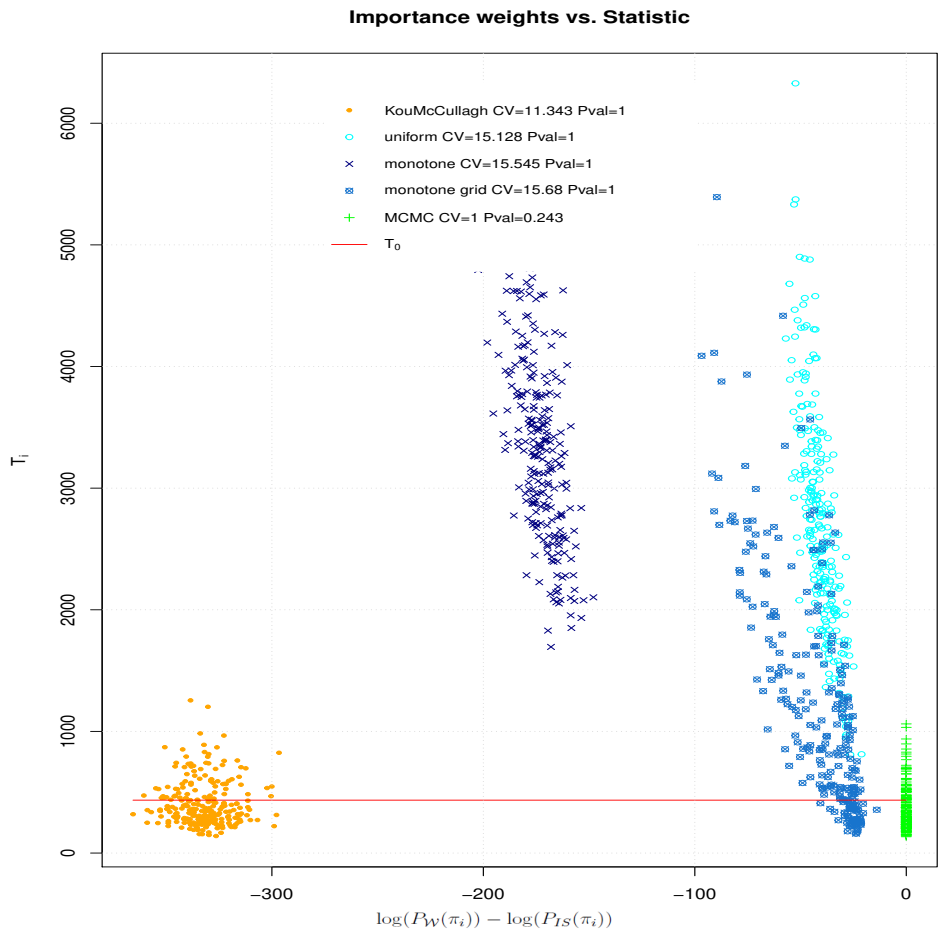}		
	\end{center}
	%\end{minipage}
	\vspace{-0.2in}
	\caption{\footnotesize Log probabilities ratios $\log(\frac{P_{\wmat}}{P_{IS}})$ for permutations sampled under different methods, vs. their computed statistic $T_i=T(\pi_i(\sample))$ for the inverse weighting statistic from Section \ref{sec:inverse_weight_stat}. Parameters are the same as in Figure \ref{fig:IS_distribution}.
Since $P_{\wmat}$ is known only up to a multiplicative constant, we show the log probabilities ratios minus the log probability ratio of the identity permutation.
For the MCMC Method, the ratio is $1$ and all permutations have equal weights, with $\approx 24\%$ of permutations $\pi_i$ having $T_i \geq T_0$. For the uniform and  monotone methods we have $T_i\geq T_0$ for all permutations giving a too conservative P-value of $1$.
For the monotone-grid and Kou-McCullagh method, we do sample permutations with $T_i < T_0$, but their overall probabilities ratios are negligible, giving again a conservative $Pvalue \approx 1$.} % end caption
\label{fig:IS_log_prob_ratio}
\end{figure}

\newpage
\clearpage

\begin{figure}[!h]
	%\begin{minipage}{1.\columnwidth}
	\begin{center}		
			\includegraphics[width=0.8\columnwidth, height=13cm]{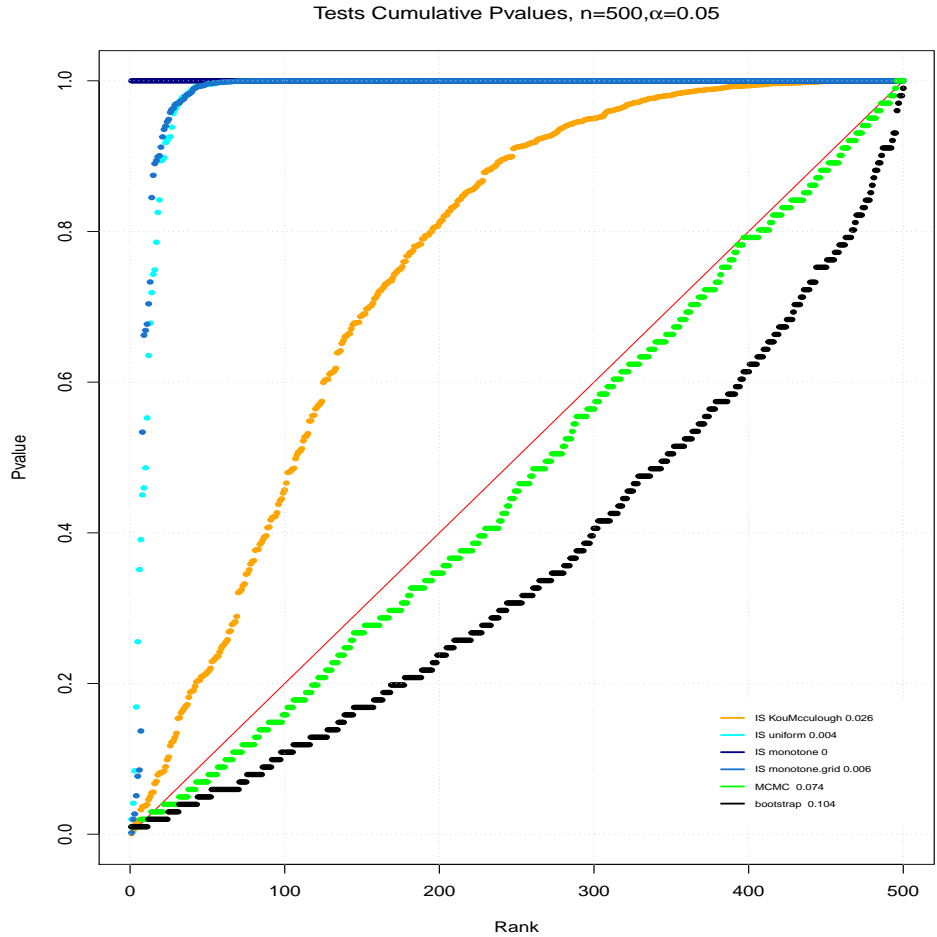}			
	\end{center}
	%\end{minipage}
	\vspace{-0.2in}
	\caption{\footnotesize Cumulative p-value distribution for a the same independence LogNormal distribution from Figure \ref{fig:IS_distribution} and the same tests. The permutation test gives an approximate $U[0,1]$ distribution.
	The bootstrap test is slightly non-calibrated, and gives low p-values under the null.
	The different importance sampling tests are all too conservative, with the Kou-McCullagh technique performing the best.}	
	\label{fig:IS_pvals}
\end{figure}

\subsection{An Iterative Algorithm for Marginal Estimation}
\label{sec:IterativeAlgorithm_appendix}

This section deals with estimation of the marginal distributions under the null independence model presented in Section \ref{sec:marginal_estimation_under_null} in the main text. We consider the more general scenario of $k$ independent random variables ${X}_1\sim {F}_1,\ldots,{X}_k\sim {F}_k$ that become dependent due to selection bias. Specifically, we assume that observations are vectors $(X_1,\ldots,X_k)$ having the density
\begin{equation} \label{basic}
\frac{w(x_1,x_2,\ldots,x_k) F_1(dx_1) F_2(dx_2)\cdots
 F_k(dx_k)}{\mathbb{E}\{w( X_1, X_2,\ldots , X_k)\}},
\end{equation}
where the non-negative weight function $w$ is known and has a finite
expectation with respect to  $ F_1\times\cdots\times
 F_k$. The aim is to estimate $ F_1,\ldots, F_k$ without any parametric assumptions. For identifiability, we assume that for
any measurable subset $\truncregion$ in the support of $ F_j$,
${E}\big[w( X_1,\ldots, X_k) \indicator{X_j\in \truncregion}\big]>0$ ($j=1,\ldots,k$); that is, observations from any subset of the support can be observed.

The problem of non-parametric estimation of a general multivariate
distribution $F$ using weighted data is well known (e.g., \cite{vardi1985empirical}) and the corresponding estimator
of $ F_1,\ldots,  F_k$  can be obtained by marginalization. Specifically, generalizing Equation \eqref{eq:non_parametric_MLE} in the main paper, the estimator for $F_j$ is
\begin{equation} \label{naive}
\hat{F}^n_j(t)=\frac{\sum_{i=1}^n I\{{X_{ij}\le t}\}
\wfun(X_{i1},\ldots,X_{ik})^{-1}} {\sum_{i=1}^n  \wfun(X_{i1},\ldots,X_{ik})^{-1}} \quad (j=1,\ldots,k),
\end{equation}
where $X_i$ is the vector of subject $i$, and $X_{ij}$ is its $j$'th coordinate.
Estimator \eqref{naive} does not exploit independence, and it is especially inappropriate when $\wfun$ vanishes in part of the support (i.e., truncation).

Let ${\tilde{E}}$ denote the conditional expectation operator,
${\tilde{E}}(w ; j,x)=\mathbb{E}\{w( X_1,\ldots, X_k)\mid  X_j=x\}$.
%We will also use the functional version
%${\tilde{E}}(w;j,X^*)={E}\{w(X^*_1,\ldots,X^*_k)\mid X^*_j\} $.
For any $1\le j\le k$, the likelihood of the data
$(x_{i1},\ldots,x_{ik})$ ($i=1,\ldots,n$) can be factorized as follows (see Equation (\ref{basic})):
\begin{equation} \label{splitlike}
\prod_{i=1}^n \frac{w(x_{i1},\ldots,x_{ik})\prod_{\ell\ne j}  F_i(dx_{i\ell})}
{{\tilde{E}}(w; j,x_{ij})} \times \prod_{i=1}^n \frac{{\tilde{E}}(w; j,x_{ij}) F_j(dx_{ij})}
{\mathbb{E}\{\tilde E(w; j, X_j)\}}.
\end{equation}
Due to independence, the first term in Equation \eqref{splitlike} does not depend on $F_j$. Thus, assuming all distributions except $F_j$ are known, the likelihood of
$n$ observations is proportional to the second term of Equation \eqref{splitlike},
and a natural non-parametric estimate of $F_j$ is $\hat{ F}_j(x)\propto \sum_i \indicator{x_{ij}\le
x} {\tilde{E}}(w;j,x_{ij})^{-1}$. This estimator assigns mass only to the observed points and gives rise to Algorithm \ref{alg:est_marg_app}, 
generalizing Algorithm \ref{alg:est_marg} in the main text.
\begin{algorithm}
	\caption{{\small Estimation of Marginals Under Quasi-independence}} {\small
	\begin{algorithmic}[1]
		\Statex {\bf Input:} $(x_{i1},\ldots,x_{ik})$ ($i=1,\ldots,n$) - sample, $w$ - bias function, $d(F_1,F_2)$ - distance function.
		\Statex {\bf Parameters:} $\epsilon$ - convergence criterion.
		\State Generate initial estimates using Equation (\ref{naive}), and set $ F^{new}_j=\hat{F}^n_j$, $ F^{old}_j\equiv 0$.
\While{$\max_j d( F_j^{old}, F_j^{new})<\epsilon$}
\State Set $ F^{old}_j=F^{new}_j$ $(j=1,\ldots,k)$.
		\For{$j=1,\ldots,k$}
\State Calculate $e^{new}_{ij}={\tilde{E}}(w; j,x_{ij})$  with
respect to $\prod_{\ell=1}^k F_\ell^{new}$, for all $x_{ij}$
$(i=1,\ldots,n)$.
\State Set
$$
 F^{new}_j(x)= \frac{\sum_{i=1}^n \indicator{x_{ij}\le x}
[e^{new}_{ij}]^{-1}} {\sum_{i=1}^n [e^{new}_{ij}]^{-1}}.
$$
\EndFor
\EndWhile
\State {\bf Output:} $F^{new}_1,\ldots,F^{new}_k$.
	\end{algorithmic}
\label{alg:est_marg_app}}
\end{algorithm}

Algorithm \ref{alg:est_marg_app} provides a non-parametric estimator for the model under the independence constraint. The most time consuming part of the algorithm is the calculation of $e^{new}_{ij}={\tilde{E}}(w; j,x_{ij})$ in step 5. In most cases, it seems
unavoidable to use full enumeration so the
complexity of each step 5 is $O(n^{k})$. For the important case of a linear weight function, the calculation of ${\tilde{E}}(w; j,x_{ij})$ is rather simple.
Let $w({\bf x})={\bf a}^{t}{\bf x}+b$, where the vector $\bf a$ and the
constant $b$ are known ($\bf a\equiv 1$ and $b=0$ is the sum-bias case $w(\textbf{x})=x_1+\cdots+x_k$). Then
\be
{\tilde{E}}(w ; j,x_{ij})={\bf a}^t\Big(
\mathbb{E} (X_1),\ldots,\mathbb{E}(X_{j-1}),x_{ij},\mathbb{E} (X_{j+1}),\ldots,\mathbb{E
} (X_k)\Big) + b
\ee
where the expectations are calculated for each coordinate separately with respect to
$ F_1^{new},\ldots, F_k^{new}$. Thus, the computational complexity of each step $5$
reduces to $O(n)$.

Several properties of Algorithm \ref{alg:est_marg_app} are readily established; their proofs are standard and hence are omitted.

\begin{proposition} \label{propincreas}
The likelihood increases in each step of the algorithm.
\end{proposition}
%{\bf Proof.} This results from updating of $ F_j$ being a
%maximization of the current likelihood, i.e., it does not change
%the contribution to the likelihood of all coordinates other than
%$j$ and increases the contribution of the $j$'th coordinate.
\begin{proposition}
The algorithm converges in the interior of the parameter space.
\end{proposition}
%{\bf Proof.} Here the parameters are the probabilities at the
%observed points. $w>0$ implies that positive mass is assigned to
%all these points (otherwise the likelihood equals 0). Also, since
%$w>0$, the likelihood is bounded from above and the algorithm
%converges as a result of Proposition \ref{propincreas}.
\begin{proposition}
The NPMLE is a fixed point of the algorithm.
\end{proposition}
%{\bf Proof.} This is immediate from the propositions above.

\subsection{Implementation Details}
\label{sec:Implementation_appendix}
We list below several practical issues that were dealt with when implementing our TIBS algorithm.

\begin{itemize}
	\item
	{\bf Permutations MCMC parameters:} When sampling permutations, we seek permutations that are independent from each other. This is approximately achieved when enough MCMC steps are used between every two permutations chosen such that the Markov Chain is mixing and is close to the stationary distribution. The parameter $M$ of MCMC steps between two consecutive permutations was set to $M=2n$, where $n$ is the sample size. This choice was based on calculating the correlation between $\pi_t(i)$ and $\pi_{t+M}(i)$ as a function of $M$ and showing that this correlation decays.
	
	In addition, a 'burn-in' parameter $M_0$ is often used in MCMC studies, which is the initial number of steps performed after starting from the identity permutation. Here, we chose to not include a 'burn-in' step, i.e. setting $M_0=0$ and taking the first permutation (the identity) as one of our permutations. This ensures that under the null we have one sample identical to the original sample, i.e. the empirical p-value cannot be lower than $\frac{1}{B}$. Some empirical arguments against the usage of 'burn-in' are given in \url{http://users.stat.umn.edu/~geyer/mcmc/burn.html}.

	\item
	{\bf Estimating expectations using MCMC:} The MCMC permutations algorithm was used to estimate the expected cell counts in the test statistic, as shown in Equation \eqref{eq:estimate_expected}, using the estimated transition probabilities $\hat{P}_{ij}$ in Equation \eqref{eq:P_ij_estimator}. For these estimators, we used {\it all} permutations generated during the MCMC algorithm, and not only the $B$ permutations selected later to produce permuted (null) samples. While consecutive permutations are strongly dependent and differ only in one pair ($\pi(i), \pi(j)$), adding
	the estimator for $P_{ij}$ over all permutations reduces the variance of the expected cell estimators, compared to using only every $M$'th permutation, at a negligible additional computational cost.
	
	\item
	{\bf Removing cells with low counts:} The test statistic $T$ in Equation \eqref{eq:modifed_hoeffding} sums over all cells, including cells with small expected counts.
	Such cells yield terms with high variance and may therefore reduce the power of the test statistic, especially recalling that we have variance not only in the observed counts $o_i^{jk}$ but we also use estimators $\hat{e}_i^{jk}$ for the expected counts.
	
	To reduce the variance of these cells, we include in the test statistic the contribution of cells for a data point $i$ only if $\hat{e}_i^{jk}>1$ for all $j,k$. Otherwise, we simply skip these cells, for both the original sample and the permuted and bootstrap samples. Alternatively, one can replace the Pearson chi-square statistic by a likelihood-ratio test statistic that is less sensitive to cells with low counts.
	
	\item
	{\bf Dealing with ties:}
	Even for a continuous distribution $\jointcdf$ and for our test statistic $T$ in Equation \eqref{eq:modifed_hoeffding}, we encounter ties data points $(x_i,y_i)$ determining the quadrants $Q_i^{jk}$ and data points $(x_j, y_{\ell})$ used as counts for the statistic. When $x_j=x_i$ and/or $y_{\ell}=y_i$ the data points lie on the boundaries of the quadrants $Q_i^{jk}$.
	This can occur in the bootstrap test (as the same values may be sampled multiple times),
	and also for permutation testing, where the permuted sample contains data points $(x_i, y_{\ell})$ and $(x_j, y_i)$ for some $j,{\ell}$. We can count these boundary points on one quadrant, another, or ignore them altogether.
	While the effect of such ties is asymptotically negligible - they may affect the statistic and P-value calculation for small samples. For example, if we discard them, then for the original sample the sum of the four quadrants observed counts will be $n-1$, but for a permuted dataset this will be typically $n-2$.
	This causes a small but systematic bias between the original and randomized (permuted or bootstrap) samples,
	and may result in an invalid test with a non-uniform P-values distribution under the null.
	To overcome this bias, our implementation uses a small perturbation of the quadrants $Q_i^{jk}$ (see Equation \eqref{eq:data_quartiles_def}), such that the center point used is $(x_i+\eps^{(x)}_i, y_i+\eps^{(y)}_i)$ with $\eps^{(x)}_i, \eps^{(y)}_i \iid N(0, 10^{-9})$. This randomly assigns any point (in the original or bootstrap/permuted sample) containing $x_i$ or $y_i$ (for example the original data point $(x_i,y_i)$) into one of the relevant quadrants, and does not affect the assignment of other points, thus ensuring that all $n$ points are counted for both the original and the randomized samples.
	
\end{itemize}

\begin{comment}
\item The following table presents all of the different tests and their applicability: \zuk{To remove?} \micha{I'm not sure we need that; if we leave it, Inverse Weighting should be in the rows}

\begin{table}	
	\scriptsize
	\begin{center}
%	\begin{tabular}{|c|c|c|c|c|}
%		\hline
	\begin{tabular}{|l||*{4}{c|}}\hline	
	\backslashbox{Pval Calc. Method}{ Test Statistic}
	&\makebox[5em]{ \makecell{Adjusted \\ Hoeffding}}&\makebox[4em]{\makecell{Inverse \\ Weighted Hoeffding}}&\makebox[3em]{Tsai}
	&\makebox[3em]{minP2}\\\hline\hline			
	Bootstrap	& $\wfun^+$, $(\wfun^{01}, \jointcdf^{\leftrightarrow})$ & $\wfun^+$ &  & \\		\hline
Permutations MCMC	& V & $\wfun^+$ & $\wfun^{01}$ + $C$ & $\wfun^{01}$ + $C$ \\		\hline
Permutations IS	&  V & $\wfun^+$ & &  \\		\hline
\end{tabular}

%	{\bf Pval Calc. Methods $\backslash$ \\ Test Statistics}	& Modified Hoeffding & Inverse Weighting & Tsai & minP2  \\
%		\hline
%%	\end{tabular}
\caption{The list of applicable tests proposed by us and in the literature. $\wfun^+$ indicates positive $\wfun$.
	$\wfun^{01}$ indicates truncation. $\jointcdf^{\leftrightarrow}$ indicates and exchangeable distribution. $C$ indicates dealing with censored data.}
\label{tab:different_tests}
\end{center}
\end{table}
\end{comment}

 % At the end

\end{document}